\documentclass[twoside,11pt]{article}
\usepackage[T1]{fontenc}
\usepackage{lmodern,microtype,mathtools,amsthm,float}
\usepackage[abbrvbib,preprint]{jmlr2e}
\usepackage[letterpaper,margin=1in]{geometry}
\usepackage{booktabs,array,tabularx,enumitem,needspace}
\usepackage{xcolor}
\usepackage{tikz}
\usepackage{algorithm}
\usepackage{algpseudocode}
\makeatletter
\renewenvironment{proof}[1][Proof]{\par\pushQED{\qed}\normalfont
\topsep6\p@\@plus6\p@\relax\trivlist
\item[\hskip\labelsep\bfseries #1\@addpunct{.}]\ignorespaces
}{\popQED\endtrivlist\@endpefalse}
\makeatother

\hypersetup{hidelinks,pdftitle={0.4633-Approximation for Online and Offline Non-Monotone DR-Submodular Maximization},pdfauthor={Vaneet Aggarwal and Donglei Du}}
\allowdisplaybreaks[2]
\numberwithin{theorem}{section}
\newcommand{\R}{\mathbb R}
\newcommand{\E}{\mathbb E}
\newcommand{\one}{\mathbf 1}
\newcommand{\zero}{\mathbf 0}

\newcommand{\ip}[2]{\langle #1,#2\rangle}
\newcommand{\norm}[1]{\lVert #1\rVert}

\newcommand{\cprof}{\mathfrak c}

\newcommand{\dd}{\,\mathrm dt}
\newcommand{\yesbox}{\ensuremath{\boxed{\scriptstyle\checkmark}}}
\newcommand{\nobox}{\ensuremath{\boxed{\scriptstyle\phantom{\checkmark}}}}
\DeclareMathOperator{\Reg}{Reg}
\DeclareMathOperator{\diam}{diam}
\DeclareMathOperator{\diag}{diag}

\DeclareMathOperator{\Unif}{Unif}
\title{0.4633-Approximation for Online and Offline\\
Non-Monotone DR-Submodular Maximization}
\author{\name Vaneet Aggarwal \email vaneet@purdue.edu\\
  \addr Purdue University
  \AND
  \name Donglei Du \email ddu@unb.ca\\
  \addr Faculty of Management, University of New Brunswick}

\firstpageno{1}
\begin{document}
\raggedbottom
\maketitle
\thispagestyle{plain}
\begin{abstract}
We obtain the approximation coefficient
$\alpha_0=54440750000/117486474169>0.463378$ for maximizing nonnegative,
non-monotone DR-submodular functions over compact convex down-closed
subsets of the unit cube, with queries allowed throughout the cube. This
improves the $4/9$ benchmark. A coordinatewise map,
exponential on small coordinates and linear-fractional on large ones,
turns the learner's state, a feasible point, into a feasible action whose value is at least $\alpha_0$
times that of any feasible point, up to a linear term. The proof is
analytical and reduces to one-dimensional inequalities. The linear term
has an unbiased estimator using at most one noisy-gradient query, so
online maximization reduces to linear learning. With conditionally
unbiased, uniformly bounded gradient responses, the expected approximate
regret is $O(\sqrt T)$ with projected ascent and $O(T^{3/4})$ with one
linear-maximization call per round. With uniformly bounded gradients and
one conditionally unbiased, uniformly bounded noisy-value response per
round, the respective rates are $O(T^{3/4})$ and $O(T^{7/8})$.
Online-to-batch conversion gives offline guarantees. With exact values
and gradients, contracting the objective toward a feasible interior point
yields an expected $(\alpha_0-\varepsilon)$-approximation for every
nonnegative $C^1$ DR-submodular objective, with oracle counts polynomial
in $n$ and $\varepsilon^{-1}$ and no gradient or smoothness parameter.
Sampling and standard rounding give $(\alpha_0-\varepsilon)$-approximation
for submodular set functions under one matroid or a fixed number of
knapsack constraints.
\end{abstract}
\begin{keywords}
DR-submodular maximization, submodular maximization, online optimization,
approximation algorithms, upper-linearizability
\end{keywords}
{\small\setcounter{tocdepth}{2}\tableofcontents}
\section{Introduction}\label{sec:introduction}
Diminishing-returns (DR) submodularity requires each marginal derivative
not to increase when any coordinate increases. We study maximization of
nonnegative DR-submodular functions over compact convex down-closed
subsets of the unit cube. Such a set remains feasible when any
coordinate of a feasible vector is decreased toward zero. The problem
includes multilinear relaxations of submodular set-function maximization
under matroid and knapsack constraints
\citep{bian2017continuous,buchbinder2024}. In the online setting, the
learner selects each action before observing that round's objective and
competes with the best fixed feasible point in hindsight.

For arbitrary compact convex down-closed subsets of the unit cube,
\citet{buchbinder2024} obtained an offline coefficient $401/1000$ for
smooth objectives, with a smoothness-dependent additive error and
uniform polynomial oracle complexity. In the online setting,
\citet{lu2026} showed that the problem is $1/e$-linearizable, which
reduces it to online linear optimization with coefficient $1/e$, and
\citet[Corollary~4.2]{geometry2026} established the coefficient $4/9$ for an
inequality of the form \eqref{eq:intro-comparison} below. Their online regret and the resulting offline query bounds
depend on the gradient scale.

\Needspace{5\baselineskip}
We improve the approximation coefficient for this class to
\begin{equation}\label{eq:intro-coefficient}
 \alpha_0=\frac{54440750000}{117486474169}=0.4633788\ldots.
\end{equation}
Figure~\ref{fig:coefficients} compares it with earlier coefficients and
with the known oracle upper bound. We obtain this coefficient by
showing that the problem is
upper-linearizable in the sense of \citet{pedramfar2024linear}: there
is a recovery map $h$ and a comparison vector $g(f,x)$ such that
\begin{equation}\label{eq:intro-comparison}
 \alpha_0f(o)-f(h(x))\le\alpha_0\ip{g(f,x)}{o-x}
\end{equation}
for every feasible state $x$ (the point the learner maintains) and every
feasible point $o$. The map $h$ turns the
learner's state into a feasible action and does not depend on the
objective; only $g$ uses the current objective. The online reduction
of \citet{pedramfar2024linear} therefore applies with either a
projected or a linear-maximization-based learner, preserving the
coefficient in both cases. Table~\ref{tab:comparison} compares these
guarantees with prior work. The known oracle upper bound for
this class is $\beta_*\approx0.47043868$, even with exact
values and full gradients; it applies offline and transfers to online
algorithms by repeating one objective
\citep[Corollary~6.4 and Proposition~7.3]{geometry2026}.

\paragraph{Technical novelty.}
We introduce a recovery map that treats small and large coordinates
differently. With $b=217763/250000$, a state coordinate $r\le1-b$, which
we call small, is mapped to $1-e^{-r}$, and a larger coordinate is
mapped to $r/(b+r)$. A single linear-fractional map would be infeasible
on small coordinates for $b<1$, while the exponential branch alone loses
too much on large coordinates; splitting the state into these two blocks
is what allows $b<1$ and hence a coefficient above $4/9$. Both branches
satisfy $0\le h(r)\le r$, so the recovered action is feasible under any
down-closed constraint. After subtracting a greedy affine lower bound in
an order matched to this map (Lemma~\ref{lem:normalization}), the large
block is compared through secants that move its coordinates one at a
time, as in the proof of \citet[Lemma~4.1]{geometry2026}, but keeping the
intermediate values that their bound discards (Lemma~\ref{lem:ordered}), and the small block is compared along
the path $t\mapsto1-e^{-tr}$ that ends at the recovered value
(Lemma~\ref{lem:two-comparisons}).

The two blocks are linked through one value on the boundary of the cube. After
normalization the objective vanishes when all small coordinates are zero
and all large coordinates are one. The large block can therefore be
compared with the small coordinates set to zero, at the cost of the
value $E(t)$ of the normalized objective on the face where all large
coordinates equal one, which the small-coordinate comparisons also produce,
and the small block can be compared with the large coordinates set to
zero; the inequality $F(o)\le F(o_S,\zero_L)+F(\zero_S,o_L)$ for the
normalized objective $F$ combines the two bounds
(Lemma~\ref{lem:two-comparisons}). We average these comparisons over the
path time $t\in[0,1]$. Early times use only the small-block comparison,
and the matching large-block comparison is used later with the same
total weight. Choosing the small-block weight proportional to the
remaining large-block weight makes the dependence of each block on the
other cancel after averaging, and the nonnegative slack of each
comparison, which the proof keeps, absorbs the difference in timing.

What remains are the intermediate values left by the large-block
comparison. A bound on their second derivative along the path controls how fast
they grow, which lets Lemma~\ref{lem:curvature} split them into layers indexed
by nested time intervals. A single inequality in the path time
(Lemma~\ref{lem:scalar}) then handles all layers at once, and
Appendix~\ref{app:scalar} proves it by elementary calculus. The resulting
vector $g$ is linear in the objective and averages partial derivatives
with bounded total weight, so sampling one term gives an unbiased
estimate from at most one gradient query.

\begin{table}[t]
\centering\small
\renewcommand{\arraystretch}{1.25}
\begin{tabularx}{\textwidth}{@{}>{\raggedright\arraybackslash}Xccccc@{}}
\toprule
Work & Online & Offline & \shortstack{Projection\\based}
 & \shortstack{Projection\\free} & \shortstack{Approximation\\coefficient} \\
\midrule
\citet{buchbinder2024} & \nobox & \yesbox & \nobox & \yesbox$^{*}$ & $401/1000$ \\
\citet{lu2026} & \yesbox & \yesbox & \yesbox & \yesbox & $1/e$ \\
\citet{geometry2026} & \yesbox & \yesbox & \yesbox & \yesbox$^{\dagger}$ & $4/9$ \\
\textbf{This work} & \yesbox & \yesbox & \yesbox & \yesbox & $\mathbf{54440750000/117486474169}$ \\
\bottomrule
\end{tabularx}
\caption{Leading approximation coefficients for nonnegative
DR-submodular objectives under down-closed convex constraints. The
cited results differ in their regularity assumptions and error terms;
for example, the guarantee of \citet{buchbinder2024} has an additive
error that depends on smoothness. The uniform $C^1$ guarantee of this
work is Theorem~\ref{thm:uniform}. An empty box means that the setting
is not claimed in the cited work. Projection columns specify access
to the feasible set. Offline entries for online methods use
online-to-batch conversion.
$^{*}$Linear optimization is required over intersections with additional
explicit polytopes. $^{\dagger}$Projection-free access follows by
substituting online Frank--Wolfe for the base learner, as justified in
Section~\ref{sec:related}.
}
\label{tab:comparison}
\end{table}

\begin{figure}[t]
\centering
\begin{tikzpicture}[x=1.22cm,y=1cm,font=\small]
 \fill[green!20] (0,-0.07) rectangle (10.338,0.07);
 \fill[gray!25] (10.338,-0.07) rectangle (11.044,0.07);
 \fill[red!15] (11.044,-0.07) rectangle (12,0.07);
 \draw[thick,->] (0,0) -- (12.4,0);
 \foreach \k/\lab in {0/0.36,2/0.38,4/0.40,6/0.42,8/0.44,10/0.46,12/0.48}{
  \draw (\k,-0.1) -- (\k,0.1);
  \node[below,font=\scriptsize] at (\k,-0.1) {$\lab$};}
 \fill (0.788,0) circle (2.2pt);
 \draw (0.788,0.1) -- (0.788,0.35);
 \node[above,align=center] at (0.788,0.35) {$1/e\approx0.368$\\ \citet{lu2026}};
 \fill (4.1,0) circle (2.2pt);
 \draw (4.1,0.1) -- (4.1,0.35);
 \node[above,align=center] at (4.1,0.35) {$0.401$\\ \citet{buchbinder2024}};
 \fill (8.444,0) circle (2.2pt);
 \draw (8.444,0.1) -- (8.444,0.35);
 \node[above,align=center] at (8.444,0.35) {$4/9\approx0.444$\\ \citet{geometry2026}};
 \fill[blue!70!black] (10.338,0) circle (3pt);
 \draw[blue!70!black,thick] (10.338,0.1) -- (10.338,1.45);
 \node[above,align=center,text=blue!70!black] at (10.338,1.45)
  {\textbf{Ours}: $\alpha_0\approx0.4634$};
 \draw[red!70!black,very thick] (11.044,-0.22) -- (11.044,0.22);
 \draw[red!70!black] (11.044,-0.22) -- (11.044,-0.6);
 \node[below,align=center,text=red!70!black] at (11.044,-0.6)
  {upper bound $\beta_*\approx0.4704$\\ \citet{geometry2026}};
\end{tikzpicture}
\caption{Approximation coefficients for nonnegative DR-submodular
maximization over compact convex down-closed sets. Dots mark previously
achieved coefficients; our coefficient $\alpha_0$ is marked in blue. No
algorithm making polynomially many value and gradient queries achieves a
coefficient above $\beta_*$ \citep{geometry2026}; the gray interval
between $\alpha_0$ and $\beta_*$ remains open.}
\label{fig:coefficients}
\end{figure}
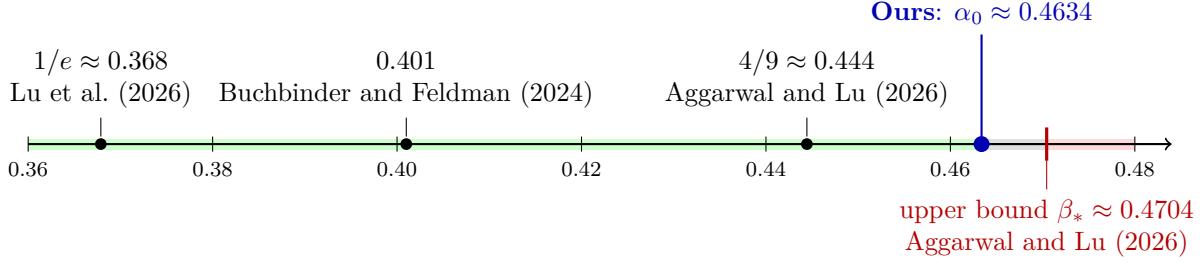

\Needspace{5\baselineskip}
\paragraph{Contributions.}
\begin{enumerate}[leftmargin=*,itemsep=0.45em]
\item \textbf{An improved affine comparison with an analytical proof.}
Theorem~\ref{thm:structural} in Section~\ref{sec:structural} proves
the coefficient $\alpha_0$ for every nonnegative continuously
differentiable DR-submodular objective on the cube, using the new
two-branch recovery map. The map and the comparison vector are explicit, and the inequality holds
simultaneously for every state and comparator. The scalar inequality
behind the coefficient is proved by one-dimensional calculus in
Appendix~\ref{app:scalar}.

\item \textbf{A field implementation with at most one gradient query.}
Lemma~\ref{lem:one-query} in Section~\ref{subsec:one-gradient}
constructs a conditionally unbiased estimate of the comparison vector
from one noisy-gradient response. It gives the second-moment bounds
used by projected ascent and the bounded-response guarantees used by
the projection-free learner.

\item \textbf{Online and offline guarantees with either constraint oracle.}
The established online reduction, with projected ascent or online
Frank--Wolfe as its base learner, gives regret orders $\sqrt T$ and
$T^{3/4}$ with one gradient query or two exact value queries per round, and $T^{3/4}$ and
$T^{7/8}$ with one noisy value
(Corollaries~\ref{thm:online}--\ref{cor:online-values}).
Corollary~\ref{cor:offline-conversion} gives the offline guarantees
by online-to-batch conversion. In the exact-oracle model,
Theorem~\ref{thm:uniform} contracts the objective toward a feasible
interior point and gives an expected $(\alpha_0-\varepsilon)$-approximation for
every nonnegative $C^1$ DR-submodular objective, with polynomial oracle
counts and no gradient, width, or smoothness parameter. It uses ordinary
linear maximization, or, alternatively, projection onto the
width-normalized feasible set.

\item \textbf{Discrete consequences at the improved coefficient.}
Corollaries~\ref{cor:discrete-matroid} and~\ref{cor:discrete-knapsack}
in Section~\ref{subsec:discrete} combine multilinear sampling with
standard rounding to give $(\alpha_0-\varepsilon)$-approximation
under one matroid or a fixed number of knapsack constraints, for every
$\varepsilon>0$; the knapsack complexity is polynomial in $n$ for fixed
$\varepsilon$ and a fixed number of knapsacks. The projection-free matroid implementation uses
set-value and independence oracles. Remark~\ref{rem:matroid-gap}
states the corresponding upper bound under matroid and cardinality
constraints.

\item \textbf{Sharp geometric recovery.} We also determine the sharp coefficient $2\sqrt{3}-3$ for maps $h$ with $0\le h(x)\le x$ that do not depend on the objective, measured against the concave closure $f^+(x)$, the largest average of $f$ over distributions with mean $x$. Theorem~\ref{gcon:main-existence} gives bounded comparison vectors at this coefficient; computing such a vector with few queries remains open.

\end{enumerate}

\section{Related Work}\label{sec:related}
The comparison in Table~\ref{tab:comparison} concerns general
DR-submodular objectives under down-closed convex constraints.

\paragraph{Continuous approximation and coordinate comparisons.}
\citet{bian2017continuous} developed structural inequalities and
approximation algorithms for non-monotone DR-submodular maximization.
\citet[Theorem~1.1]{buchbinder2024} obtained
$401/1000\,\mathrm{OPT}-O(\delta LD^2)$ for an $L$-smooth objective
on a down-closed convex set of diameter $D$, in time polynomial in
dimension and $\delta^{-1}$. Their assumption (meta-solvability) permits
linear optimization over intersections with additional explicit
polytopes. Theorem~\ref{thm:uniform} gives a randomized expected
$(\alpha_0-\varepsilon)\mathrm{OPT}$ guarantee with no additive
smoothness term, using only ordinary linear maximization, with query
counts independent of gradient magnitudes and of any smoothness
modulus.
\citet[Corollary~4.2]{geometry2026} established coefficient $4/9$
for arbitrary compact convex down-closed subsets of the cube using a
linear-fractional recovery map. Their geometry-dependent bound is
stronger than 4/9 when the feasible set is known to contain a point
$\lambda\one$ with $\lambda>0$. Our approximation coefficient $\alpha_0$ improves on $4/9$ for
arbitrary compact convex down-closed sets, using down-closedness alone. Lemma~\ref{lem:ordered} carries out the computation in the proof of
their comparison \citep[Lemma~4.1]{geometry2026} on a slice of the large
block and keeps the intermediate values that their final bound
discards. Lemma~\ref{lem:normalization} uses the greedy affine lower
bound \citep[Proposition~3.2]{bach2013}.

\paragraph{Online reductions and base learners.}
\citet{pedramfar2024linear} introduced upper-linearizable functions,
a class that includes concave functions and several classes of
monotone and non-monotone up-concave functions. Their Online
Maximization by Quadratization (OMBQ)
\citep[Theorem~1]{pedramfar2024linear} converts online linear
optimization algorithms into algorithms for this class, and our
comparison places the present problem in it.
We apply this reduction with projected ascent \citep{zinkevich2003}
and online Frank--Wolfe \citep[Theorem~4.4]{hazan2012projection}.
The former permits conditionally unbiased feedback with bounded second
moments; the latter uses an almost-sure norm bound and one ordinary
linear maximization call per round. Both regret bounds hold for reward
vectors that depend on the learner's past states.

\citet{lu2026} obtained coefficient $1/e$ with one gradient query per
round, by showing that the problem is $1/e$-linearizable and hence
reducible to online linear optimization. \citet[Sections~5 and~7]{geometry2026}
obtained coefficient $4/9$ online and offline, including one-gradient
sampling. We extend their sampling argument, which draws one coordinate and one
point on a segment,
\citep[Lemma~5.1]{geometry2026} to the signed coefficients arising
from the affine normalization and path averaging. 

The reduction of \citet{lu2026} accepts any base learner for online
linear optimization, projection-based or projection-free, and their
static regret bound uses a projection-free learner that calls a
separation oracle \citep[Section~5]{lu2026}. The projection-free
entry for \citet{geometry2026} marked in Table~\ref{tab:comparison}
uses the same property: their base learner enters only through its
regret for linear rewards, and their $4/9$ comparison supplies these
rewards. Under a uniform bound on the gradients the rewards are
bounded, so online Frank--Wolfe can replace the base learner with the
same leading approximation coefficient.

\paragraph{Value feedback and offline conversion.}
The value-feedback reduction of
\citet[Theorems~5--6]{pedramfar2024linear} replaces gradient
observations by randomized value observations. The cube-domain
implementation of \citet[Lemma~5.3]{geometry2026} applies even when
the feasible set has empty interior. Section~\ref{subsec:value-feedback}
uses that implementation with our comparison vector. For a fixed
objective, online-to-batch (OTB) conversion selects a played action
uniformly at random, giving approximation error equal to regret
divided by the number of rounds
\citep[Theorem~8]{pedramfar2024linear}.

\paragraph{Discrete approximation and oracle limits.}
Multilinear extensions express set-function maximization as a continuous
DR-submodular problem. Randomized swap rounding preserves the
multilinear value in expectation under a matroid constraint
\citep{chekuri2010swap}. For a fixed number of knapsack constraints,
\citet[Section~2]{kulik2013knapsack} transfer a continuous approximation
to a feasible set with arbitrarily small coefficient loss.
\citet[Section~8]{geometry2026} applied these reductions at coefficient
$4/9$. Applying them to our comparison vector yields coefficient
$\alpha_0-\varepsilon$, with direct set-value sampling and an
implementation using ordinary linear optimization over the matroid
polytope.

For the continuous DR-submodular model considered here, \citet[Corollary~6.4 and Proposition~7.3]{geometry2026}
prove the oracle upper bound $\beta_*\approx0.47043868$, offline and
online, using the symmetry-gap framework of \citet{vondrak2013}.

\section{Problem Definition and Background}\label{sec:background}
We use the online model of \citet{geometry2026}: the learner commits
a feasible action before obtaining feedback about the current objective.
The comparison is with a single feasible point chosen in hindsight.

\subsection{Objectives, constraints, and feedback}\label{sec:online-model}
Let $n$ be the dimension and $Q=[0,1]^n$. Vector inequalities and products
are coordinatewise. Write $[n]=\{1,\ldots,n\}$, $e_i$ for the $i$th
coordinate vector, and $\one_A$ for the indicator vector of $A\subseteq[n]$.
For a set $B\subseteq\R$, $\one_B(t)$ is its indicator function, and
$\one\{\cdot\}$ equals one when its condition holds and zero otherwise.
The symbols $\zero$ and $\one$ denote the all-zero and all-one vectors
in the indicated space. The notation $\partial_i f$ denotes the partial
derivative in coordinate $i$. Inner products and norms are Euclidean, and
$r_+=\max\{r,0\}$ for a real number $r$. Boundary derivatives are
continuous extensions of the interior derivatives.
We use the differential definition of DR-submodularity.
\begin{definition}[DR-submodularity]\label{def:dr}
A continuously differentiable ($C^1$) function $f:Q\to\R$ is DR-submodular if
\[
 x\le y\quad\Longrightarrow\quad \nabla f(x)\ge\nabla f(y)
 \qquad(x,y\in Q).
\]
\end{definition}
Every objective is nonnegative on $Q$, although its partial derivatives
may have either sign. DR-submodularity implies concavity along segments
with coordinatewise nonnegative or nonpositive direction, since the
directional derivative is nonincreasing along such a segment. In particular,
for comparable $x,y$, directional tangency gives
$f(y)\le f(x)+\ip{\nabla f(x)}{y-x}$.

The feasible region $K\subseteq Q$ is nonempty, compact, convex, and
down-closed: $x\in K$ and $0\le y\le x$ imply $y\in K$. Thus
$\zero\in K$. Let $D$ be a supplied upper bound on $\diam(K)$;
$D=\sqrt n$ always suffices. We consider two forms of constraint access.
A Euclidean projection oracle returns
$\Pi_K(v)=\arg\min_{z\in K}\norm{v-z}$. An ordinary linear
maximization oracle (LMO) returns a maximizer of $\ip c z$ over $z\in K$
for any $c\in\R^n$. Constraint-oracle calls are counted separately from objective queries,
and tie-breaking is fixed and Borel measurable.

An oblivious adversary fixes $f_1,\ldots,f_T$ before the learner's
randomness. At round $t$, the learner selects $p_t\in K$, earns reward
$f_t(p_t)$, and then queries $f_t$ on $Q$. In the noisy-gradient model,
the response $\widehat\nabla f_t(y)$ at a queried point $y$ has conditional mean
$\nabla f_t(y)$ and conditional second moment at most $B_1^2$.
Conditioning includes all previous observations and random choices, the
current objective, and the chosen query point. Exact gradients are a
special case. For the standard projection-free learner, we additionally
assume $\norm{\widehat\nabla f_t(y)}\le B_1$ almost surely; the
projected learner only needs the stated second moment.
Section~\ref{subsec:value-feedback} also considers exact values and
conditionally unbiased noisy values. As in \citet[Sections~5.2--5.3]{geometry2026}, the query domain is
$Q$, while the rewarded action must belong to $K$.

For $\alpha\in(0,1]$, expected static approximate regret is
\begin{equation}\label{eq:regret-definition}
 \Reg_\alpha(T)=\alpha\max_{o\in K}\sum_{t=1}^T f_t(o)
 -\E\sum_{t=1}^T f_t(p_t).
\end{equation}
The comparator is one feasible point used in all rounds.
For exact-gradient bounds, define the endpoint gradient bound
\begin{equation}\label{eq:gradient-envelope}
 H_i(f)=\max\{|\partial_i f(\zero)|,|\partial_i f(\one)|\},\qquad
 H(f)=\left(\sum_{i=1}^n H_i(f)^2\right)^{1/2}.
\end{equation}
DR-submodularity gives
$\partial_i f(\zero)\ge\partial_i f(y)\ge\partial_i f(\one)$ for
$y\in Q$, hence $|\partial_i f(y)|\le H_i(f)$
\citep[Lemma~A.4]{geometry2026}.
The exact-gradient results use a supplied bound
$H(f_t)\le\overline H$. If $\norm{\nabla f_t(y)}\le B$ for every
$y\in Q$, then $H(f_t)^2\le2B^2$ by summing the two squared endpoint
norms. Thus either type of bound can be used to set the learning rate.

For the uniform offline guarantee we use the following exact-oracle
model. The objective $f$ is fixed,
nonnegative, $C^1$, and DR-submodular on $Q$, with
$\mathrm{OPT}=\max_{x\in K}f(x)$. Exact values and full gradients are
available on $Q$, together with an ordinary LMO over $K$. A coefficient
$\alpha$ is uniformly achievable in this model when, for every
$\varepsilon\in(0,1)$, the output $X\in K$ satisfies
\begin{equation}\label{eq:uniform-model}
 \E f(X)\ge(\alpha-\varepsilon)\mathrm{OPT}
\end{equation}
with oracle counts polynomial in $n$ and $\varepsilon^{-1}$, uniformly
over objective and gradient magnitudes and the coordinate widths
$\max_{x\in K}x_i$; no
smoothness parameter is assumed. Counts refer to exact-real oracle
calls, with the scalar arithmetic and sampling operations of the field
construction. The structural comparison and online results use the
$C^1$ assumptions stated above.

\subsection{Reduction to online linear optimization}\label{sec:linearizable}
An objective-independent map $h$ assigns a feasible action to each
state $x\in K$. We use the following specialization of the
upper-linearizable comparison of \citet{pedramfar2024linear}, in
which one vector bounds the action's loss against every comparator.
\begin{definition}[Upper-linearizable comparison]\label{def:linearizable}
For $0<\alpha\le1$, an objective-independent map $h:K\to K$ and a field
$g(f,x)\in\R^n$ give a comparison with coefficient $\alpha$ if, for every nonnegative $C^1$ DR-submodular
$f:Q\to\R$,
\begin{equation}\label{eq:linearizable}
 \alpha f(o)-f(h(x))\le\alpha\ip{g(f,x)}{o-x}
 \qquad(x,o\in K).
\end{equation}
For fixed $f$, the map $x\mapsto g(f,x)$ is called the comparison
field. Its value is independent of the comparator $o$.
\end{definition}
The normalization in \eqref{eq:linearizable} absorbs the scale of the
field. A right-hand side $\beta'\ip{\widetilde g(f,x)}{o-x}$ with
$\beta'>0$, as in \citet{pedramfar2024linear}, takes the displayed form
after setting $g=(\beta'/\alpha)\widetilde g$.

Summing \eqref{eq:linearizable} bounds the approximate reward gap by
$\alpha$ times the linear regret against the same fixed comparator.
The vectors $g(f_t,x_t)$ depend on the learner's states, so the
base-learner bound must allow reward vectors that depend on the past. Conditional
unbiasedness then permits replacement of these vectors by their
estimates in the expected bound, as proved in
Proposition~\ref{thm:online-transfer}.

\section{Structural Result}\label{sec:structural}
We construct a downward map $h$ and a field $g(f,x)$ satisfying
$f(o)\le\alpha_0^{-1}f(h(x))+\ip{g(f,x)}{o-x}$ for every pair of
cube points $x,o$.

The parameters are
\begin{equation}\label{eq:parameters}
\begin{gathered}
 b=\frac{217763}{250000},\qquad s=1-b=\frac{32237}{250000},\qquad
 t_0=\frac14,\\
 \Gamma=1+b+\frac{1}{4b}=\frac{117486474169}{54440750000},\qquad
 \alpha_0=\Gamma^{-1}=\frac{54440750000}{117486474169}>0.463378.
\end{gathered}
\end{equation}
Here $s=1-b$ is both the coordinate threshold and the rate in the
factors $e^{\pm st}$ that appear in the averaging weights below; $t_0$ and the constants in \eqref{eq:beta-c}
below enter only the averaging schedule.
For a scalar coordinate $r\in[0,1]$, the recovery is
\begin{equation}\label{eq:action-map}
 h(r)=\begin{cases}1-e^{-r},&0\le r\le s,\\
 r/(b+r),&s<r\le1.
 \end{cases}
\end{equation}
At $r=s$ the exponential branch applies, so the map jumps upward just
to the right of $s$; only Borel measurability of $h$ is used later.
Its coordinatewise extension is
$h(x)=(h(x_1),\ldots,h(x_n))$.
The first branch satisfies $1-e^{-r}\le r$. For the second branch,
$b+r>1$, so $r/(b+r)\le r$. Thus $0\le h(x)\le x$, and down-closedness
implies $h(x)\in K$ whenever $x\in K$.

\subsection{Construction of the field}\label{sec:field}
The construction has the form
\[
 f=\ell+F,\qquad
 g(f,x)=d+\int_0^1\mathcal J(t)\,dt,
\]
where $\ell(z)=f(\zero)+\ip dz$ is an affine normalizer, $F\ge0$ is the
normalized residual, and $\mathcal J(t)$ is a path contribution.
The definitions below come in three
groups: the normalization $\ell$ and the coordinate orders; the paths,
secants, and gradient combinations used by the coordinate comparisons;
and the scalar averaging weights that combine those comparisons.

Fix $x\in Q$ and $f$. All coordinate orders and scalar weights below
depend on $x$ but are independent of $f$; a subscript restricts a vector
or gradient to the indicated coordinates. The two blocks are
\[
 S=\{j\in[n]:x_j\le s\},\qquad L=[n]\setminus S,\qquad
 u=x_S,\quad v=x_L.
\]
We call the coordinates in $S$ small and those in $L$ large. A pair
of block vectors is merged in the original coordinate order when used
as a function argument. For $t\in[0,1]$, the small-coordinate path and
recovered large coordinates are
\begin{equation}\label{eq:paths-small}
 a_j(t)=1-e^{-tu_j}\quad(j\in S),\qquad
 p_i=\frac{v_i}{b+v_i}\quad(i\in L).
\end{equation}
A prime denotes differentiation with respect to $t$.
Then $h(x)=(a(1),p)$ and $a_j'(t)=u_j(1-a_j(t))$.

We use an affine normalization based on the greedy lower bound
\citep{bach2013}, with an order chosen for the present recovery.
The coefficients are $\gamma_i=h(x_i)-\alpha_0x_i$, and the permutation
$\pi(1),\ldots,\pi(n)$ orders them nonincreasingly, with ties broken
by index. The prefixes are $C_0=\varnothing$ and
$C_\iota=\{\pi(1),\ldots,\pi(\iota)\}$. We call the points $\one_{C_\iota}$ ($0\le\iota\le n$) the chain
vertices. For $z\in Q$, the normalizer is
\begin{equation}\label{eq:normalizer}
 d_{\pi(\iota)}=f(\one_{C_\iota})-f(\one_{C_{\iota-1}}),\qquad
 \ell(z)=f(\zero)+\ip d z,\qquad F(z)=f(z)-\ell(z).
\end{equation}
Lemma~\ref{lem:normalization} proves that $F$ remains nonnegative and
DR-submodular and has the boundary zeros used below.
One dummy coordinate with state $1$ is added, independently of whether
$L$ is empty. The functions $f,\ell,F$ are extended independently of it,
so the increments $d_i$ in \eqref{eq:normalizer} of the original
coordinates are unchanged and its own increment is zero. As a large
coordinate it precedes the original small coordinates in the order
$\pi$. In the order of the large coordinates defined below it is placed
first, including in a tie, and its weight in \eqref{eq:q-weights} is
$q=1$.
Its field component is zero, so removing it gives the original-dimensional
field. The padding makes the coefficient sum $c_1=1+1/b$ in
\eqref{eq:V} independent of the state; all oracle queries omit the dummy
coordinate.
From here on $L$ includes the dummy coordinate, so the formulas below
have $m=|L|\ge1$. If $S$ is empty, all its
vector terms are absent.

The large coordinates use a second order, independent of $\pi$.
Their scalar weights are
\begin{equation}\label{eq:q-weights}
 \begin{aligned}
 q_i&=\max\{v_i,1-v_i\},& \cprof(q)&=1+q/b,\\
 \vartheta(q)&=q(1-q)/b,& \Lambda(q)&=1+b+\vartheta(q)=\cprof(q)(b+1-q).
 \end{aligned}
\end{equation}
The coordinates in $L$ are relabeled as $1,\ldots,m$ in nonincreasing
order of $q_i$, with the dummy coordinate first and ties among original
coordinates broken by their original indices. For
$0\le i\le m$, the vector $P_i^-$ has its first $i$ coordinates equal
to $p$ and all remaining coordinates zero. The vector $P_i^+$ has the
same first $i$ coordinates and all remaining coordinates equal to one. Hence
$P_0^-=\zero_L$, $P_0^+=\one_L$, and $P_m^-=P_m^+=p$.
The associated secants and field components are
\begin{align}
 A_i^L(t)&=\frac{F(a(t),P_i^-)-F(a(t),P_{i-1}^-)}{p_i},
 &B_i^L(t)&=\frac{F(a(t),P_{i-1}^+)-F(a(t),P_i^+)}{1-p_i},
 \label{eq:large-secants}\\
 g_i^L(t)&=(1-p_i)A_i^L(t)+p_iB_i^L(t).
 \label{eq:large-field}
\end{align}
Both denominators are positive: $p_i>s$ and $1-p_i\ge b/(1+b)$.
The proof also uses the paths
\begin{equation}\label{eq:prefix-functions}
 Y_i(t)=F(a(t),P_i^+),\quad E(t)=Y_0(t),\quad
 y(t)=Y_m(t),\quad S_0(t)=F(a(t),\zero_L).
\end{equation}

On the face where every large coordinate equals one, put
$e(r)=F(r,\one_L)$ for $r\in[0,1]^S$ and
$Q_t(r)=e(a(t)+(1-a(t))r)$. If $k=|S|>0$, the coordinates
$j_1,\ldots,j_k$ are ordered nonincreasingly by $u_j$, with ties
broken by index. Their prefixes are $B_0=\varnothing$ and
$B_\iota=\{j_1,\ldots,j_\iota\}$ for $1\le\iota\le k$. The increment vector is
\[
 (\nu_t)_{j_\iota}=Q_t(\one_{B_\iota})-Q_t(\one_{B_{\iota-1}}).
\]
The gradient combination used in the field is
\begin{equation}\label{eq:V}
 V(t)=\sum_{i=1}^m\chi_i\nabla_S F(a(t),P_i^+),\qquad
 \chi_i=\frac{q_i-q_{i+1}}b\ (i<m),\qquad \chi_m=\cprof(q_m).
\end{equation}
The coefficients $\chi_i$ are nonnegative, and their sum is
$c_1=\cprof(q_1)=1+1/b$, because the padding gives $q_1=1$.

The averaging schedule is independent of the objective and of the
state. Its constants are
\begin{equation}\label{eq:beta-c}
\begin{gathered}
 t_1=\frac{319}{500},\qquad \beta=bt_0=\frac{217763}{1000000},\qquad
 \varpi_0=\frac{1524341}{24250000},\qquad
 \varpi_1=\frac{24171693}{9409000},\\
 w_1=\frac{755091259}{1617960000},\qquad
 w_2=\frac{11}{40},\qquad w_3=\frac{121}{1000}.
\end{gathered}
\end{equation}
The averaging density $w$, the density $\varpi$ with which the direct
large-block comparison \eqref{eq:direct-comparison} of
Lemma~\ref{lem:two-comparisons} below is averaged, and their tail
$W$ are
\begin{align}
 w(t)&=\begin{cases}
 0,&0\le t<t_0,\\
 b+w_1(t-t_0)+w_2(t-t_0)^2,&t_0\le t\le t_1,\\
 w(t_1)+w_3(t-t_1),&t_1\le t\le1,
 \end{cases}\nonumber\\
 \varpi(t)&=[\varpi_0+\varpi_1(t-t_0)]\one_{[t_0,t_1]}(t),
 &W(t)&=\int_t^1[w(r)+\varpi(r)]\,dr.
 \label{eq:wW}
\end{align}
The constants satisfy the exact mass identities $\int_0^1\varpi=\beta$
and $\int_0^1w=1-\beta$. With $t_1-t_0=97/250$, the second is equivalent to
\[
 w_1=\frac{1-b-w_2[(t_1-t_0)^3/3+(t_1-t_0)^2(1-t_1)]-w_3(1-t_1)^2/2}
                {(t_1-t_0)^2/2+(t_1-t_0)(1-t_1)}.
\]
Write $\kappa_0(t)=b\one_{[0,t_0)}(t)$. The two small-coordinate weights
are
\begin{equation}\label{eq:lambda-kappa}
 \kappa(t)=bW(t),\qquad \lambda(t)=w(t)+\kappa_0(t)-\kappa(t).
\end{equation}
Thus $W=1$, $\kappa=\kappa_0=b$, and $\lambda=0$ before $t_0$, while
$\lambda=w-bW\ge0$ after $t_0$ because $w$ increases from $b$ and $W$
decreases from one. The measure $\mu(dt)=[w(t)+\varpi(t)]\,dt$ is a
probability measure, all weights are nonnegative, and
\begin{equation}\label{eq:mass-identities}
\begin{gathered}
 \int_0^1\varpi=\beta=\int_0^1\kappa_0,\qquad
 \int_0^1w=1-\beta,\qquad \int_0^1(\lambda+\kappa)=1,\\
 \lambda+\kappa=w+\kappa_0,\qquad \kappa\ge\kappa_0,\qquad
 \int_t^1\varpi(r)\,dr\ge\int_t^1\kappa_0(r)\,dr=b(t_0-t)_+.
\end{gathered}
\end{equation}
The tail inequality holds because the mass $\beta$ of $\varpi$ lies in
$[t_0,t_1]$. The small block thus averages with $(\lambda+\kappa)\,dt$
and the large block with $\mu$. Both are probability measures, but they
place the mass $\beta$ at different times; the proof of
Theorem~\ref{thm:structural} bounds this difference by the
nonnegative remainders of Lemma~\ref{lem:two-comparisons} below.

With $c_1=1+1/b$, put $\rho_1=c_1\lambda+\kappa$. Since
$\int_0^1\lambda=1-\int_0^1\kappa=1-b\int_0^1W$ by
\eqref{eq:mass-identities} and $c_1-1=1/b$,
\[
 \int_0^1\rho_1=c_1-(c_1-1)b\int_0^1W=c_1-\int_0^1W\ge c_1-1=\frac1b .
\]
As $e^{-st}\ge e^{-s}$ and $e^{-s}>1-s=b$, it follows that
$\int_0^1e^{-st}\rho_1(t)\,dt\ge e^{-s}/b>1$. Since $\rho_1>0$ almost
everywhere, there is a unique
$z\in(0,1)$ with
\[
 \int_0^ze^{-st}\rho_1(t)\,dt=1.
\]
Appendix~\ref{app:scalar} proves $z>3/4>t_1$. Define
the weight $\psi(t)=\rho_1(t)\one_{[0,z]}(t)$ and the two auxiliary
functions
\begin{align}
 \theta_+(t)&=e^{st}\left[\Gamma e^{-s}-b
                  +b\int_0^t\rho_1(r)e^{-sr}\,dr\right],
 \label{eq:theta-plus}\\
 \theta_-(t)&=\Gamma e^{s(t-1)},
 &\theta(t)&=\min\{\theta_+(t),\theta_-(t)\}.
 \label{eq:theta-minus}
\end{align}
They cross exactly at $z$, because
$\theta_+(t)-\theta_-(t)=be^{st}\bigl[\int_0^t\rho_1(r)e^{-sr}\,dr-1\bigr]$
is negative before $z$ and positive after $z$. The remaining weights are
\begin{equation}\label{eq:mixing}
 U(t)=\rho_1(t)-\psi(t)=\rho_1(t)\one_{(z,1]}(t),\qquad
 \sigma(t)=\theta(t)-\rho_1(t).
\end{equation}
The comparison field is defined by its two coordinate blocks:
\begin{equation}\label{eq:field}
 \begin{aligned}
 g(f,x)&=d+\int_0^1\mathcal J(t)\,dt,\\
 \mathcal J_S(t)&=\lambda(t)(1-a(t))V(t)
       +\kappa(t)(1-a(t))\nabla_S F(a(t),\zero_L)-U(t)\nu_t,\\
 \mathcal J_L(t)&=[w(t)+\varpi(t)]g^L(t).
 \end{aligned}
\end{equation}
Products with $1-a(t)$ are coordinatewise. Each gradient term carries
this factor because the small-coordinate comparison is written at
$a(t)+(1-a(t))o_S$. The increments $\nu_t$ need no such factor, since
$Q_t$ already evaluates $e$ at $a(t)+(1-a(t))r$.
On $[0,t_0)$ only the term with $\kappa$ is present, since
$\lambda=U=0$ there. The direct large-block comparison~\eqref{eq:direct-comparison} of
Lemma~\ref{lem:two-comparisons} is averaged with density $\varpi$.
The normalizer $\ell$ and the residual $F$ serve to define and analyze
the field. The implementation in Section~\ref{subsec:one-gradient}
samples an equivalent raw-derivative representation directly and
evaluates neither.

The field just defined bounds every comparator with the same vector.
\begin{theorem}[Affine comparison on the cube]\label{thm:structural}
For every nonnegative continuously differentiable DR-submodular
$f:Q\to\R$, the field $g(f,x)$ in \eqref{eq:field} satisfies
\begin{equation}\label{eq:structural}
 f(o)\le\Gamma f(h(x))+\ip{g(f,x)}{o-x}
 \qquad(x,o\in Q).
\end{equation}
For fixed $x$, $g(f,x)$ is linear in $f$. Both $h$ and $g(f,\cdot)$
are Borel measurable.
\end{theorem}
The proof of Theorem~\ref{thm:structural} is given in
Section~\ref{sec:comparison-proof}, after the normalization and path
comparisons of Section~\ref{sec:coordinate} and the scalar averaging
inequality and ordered representation of Section~\ref{sec:curvature}.

For this fixed recovery map, equality is attained by a two-coordinate
cut.
\begin{remark}[Midpoint tightness and the comparison field]
\label{rem:cut-tightness}
Let $f(z)=z_1+z_2-2z_1z_2$ and
$K=\{z\in\R_{\ge0}^2:z_1+z_2\le1\}$.
At $x=(1/2,1/2)$, both coordinates use the linear-fractional
branch of~\eqref{eq:action-map}. Writing $p=(2b+1)^{-1}$ gives
\begin{equation}\label{eq:cut-midpoint-tightness}
 h(x)=(p,p),\qquad
 f(h(x))=2p(1-p)=\frac{4b}{(2b+1)^2}=\alpha_0,
 \qquad f(e_1)=f(e_2)=1.
\end{equation}
The two comparators $e_1,e_2$ have average $x$.
Since $f(z)\le1$ throughout the cube, their average objective value
is the largest possible among all distributions of comparators with
mean $x$.
Averaging~\eqref{eq:structural} over them cancels the linear term,
so $1\le\Gamma f(h(x))=\Gamma\alpha_0=1$.
Thus $\alpha_0$ is the largest coefficient possible in this
comparison for the fixed recovery map. For $b>1/2$, the bound
$4b/(2b+1)^2$ forced by this example decreases strictly with $b$, since its derivative
is $4(1-2b)/(2b+1)^3$.
Theorem~\ref{gcon:main-existence} shows that the largest universal
coefficient among deterministic objective-independent downward maps
is $2\sqrt3-3$, attained by the same two-branch map at $b=\sqrt3/2$.
The vectors in its part~(ii) come from a non-constructive
supporting-hyperplane argument and have no known query-efficient
estimator; the algorithm here uses
$b=217763/250000$, for which the explicit field has the one-query sampler
of Lemma~\ref{lem:one-query}. The three coefficients in this discussion
are ordered as
\[
 \alpha_0\approx0.463379\;<\;2\sqrt3-3\approx0.464102\;<\;
 \beta_*\approx0.470439:
\]
the coefficient attained by the explicit field, the largest universal
coefficient for deterministic objective-independent downward recovery
relative to the concave closure, and the general oracle upper bound of
Remark~\ref{rem:smooth-hardness}.

At the state $x=(u,1/2)$, with one small coordinate $0<u\le s$ and one
large coordinate, the same construction
gives a different vector from the objective gradient:
\begin{equation}\label{eq:cut-gradient-field}
 \begin{gathered}
 \nabla f(x)=(0,1-2u),\qquad
 g(f,x)=\bigl(2\mathcal M_S(u)-1,\;2\mathcal M_L(u)-1\bigr),\\
 \mathcal M_S(u)=\int_0^1(\lambda(t)+\kappa(t))e^{-ut}\,dt,\qquad
 \mathcal M_L(u)=\int_0^1e^{-ut}\,\mu(dt),
 \end{gathered}
\end{equation}
where $\mu$ is the averaging measure defined after~\eqref{eq:lambda-kappa}.
Indeed, the affine normalizer is $\ell(z)=z_2-z_1$. Substituting
$f-\ell=2z_1(1-z_2)$ into~\eqref{eq:field}, the large secants equal
$-2a(t)$, the small-block gradients $\nabla_SF(a(t),l)$ at $l=p$, $\zero_L$, and $\one_L$
are $2(1-p)$, $2$, and $0$, and $\cprof(1/2)(1-p)=1$ because
$p=(2b+1)^{-1}$; the identities in \eqref{eq:mass-identities}
then yield~\eqref{eq:cut-gradient-field}.
Since both $\lambda+\kappa$ and $\mu$ have total mass one,
$\mathcal M_S(u),\mathcal M_L(u)\ge e^{-u}\ge1-u\ge b>1/2$, so the field's
first coordinate is positive although the first gradient coordinate is
zero. Thus the field supplies a positive linear reward in the $e_1$
direction, whereas the gradient supplies zero reward in that direction.
\end{remark}

\subsection{Normalization and coordinate inequalities}\label{sec:coordinate}
The greedy lower bound uses the order of $h(x_i)-\alpha_0x_i$.
This order places the large coordinates first. Lemma~\ref{lem:normalization}
shows that it gives $F(\zero)=F(\zero_S,\one_L)=F(\one)=0$ and
$\ell(h(x))\ge\alpha_0\ell(x)$.
\begin{lemma}[Affine normalization]\label{lem:normalization}
The function $F$ in \eqref{eq:normalizer} is nonnegative and
DR-submodular on $Q$, and
\[
 F(\zero)=F(\zero_S,\one_L)=F(\one)=0,\qquad
 \ell(h(x))\ge\alpha_0\ell(x).
\]
\end{lemma}
\begin{proof}
The restriction $B\mapsto f(\one_B)-f(\zero)$ is a normalized
submodular set function. The greedy lower bound
\citep[Proposition~3.2]{bach2013} therefore gives
$f(\one_B)\ge f(\zero)+\sum_{i\in B}d_i=\ell(\one_B)$ for every
$B\subseteq[n]$. Coordinatewise concavity extends this bound to the
cube: interpolating between the two endpoints in each coordinate gives
\begin{equation}\label{eq:greedy-cube-lower}
 f(z)\ge\sum_{B\subseteq[n]}f(\one_B)
       \prod_{i\in B}z_i\prod_{i\notin B}(1-z_i)
 \ge f(\zero)+\sum_{i=1}^nd_i z_i=\ell(z).
\end{equation}
The last inequality applies $f(\one_B)\ge\ell(\one_B)$ to each term and
uses that $\ell$ is affine and the product weights have mean $z$. Thus $F\ge0$. Subtracting an affine function
preserves the order of partial derivatives and hence DR-submodularity.
The definition of $d$ also gives $F(\one_{C_\iota})=0$ for every chain vertex.

It remains to check the chosen order and the restoration inequality.
The parameters satisfy $b>\sqrt3/2$ and $\alpha_0<b$. Since
$\Gamma=1+b+1/(4b)=(2b+1)^2/(4b)$, we have $\alpha_0=4b/(2b+1)^2$.
For $0\le u\le s$, we have $1-e^{-u}=\int_0^ue^{-r}\,dr$ and
$b=1-s\le e^{-s}\le e^{-r}\le1$ for $0\le r\le u$, so
$bu\le1-e^{-u}\le u$. Subtracting $\alpha_0u$ gives
\[
 0\le (b-\alpha_0)u
 \le1-e^{-u}-\alpha_0u\le(1-\alpha_0)u.
\]
The last inequality is strict for $u>0$. On $[s,1]$, the large-coordinate formula, extended to the closed
interval, $\gamma_L(v)=v/(b+v)-\alpha_0v$, is concave and satisfies
\[
 \gamma_L(s)=s(1-\alpha_0),\qquad
 \gamma_L(1)-s(1-\alpha_0)
 =\frac{b^2(4b^2-3)}{(1+b)(2b+1)^2}>0.
\]
Here $\gamma_L(s)=s-\alpha_0s$ because $b+s=1$, and
$\gamma_L(1)-s(1-\alpha_0)=\frac1{1+b}-(1-b)-b\alpha_0
=\frac{b^2}{1+b}-\frac{4b^2}{(2b+1)^2}$, which equals the displayed
fraction; it is positive because $4b^2>3$.
Concavity therefore gives $\gamma_L(v)\ge s(1-\alpha_0)$ on this
interval. For a small coordinate, the strict inequality above gives
$\gamma_i<(1-\alpha_0)s$ when $u_i>0$, and $\gamma_i=0$ when $u_i=0$.
The two blocks are therefore strictly separated, so the index tie rule
never interleaves them, and every coordinate in $L$ precedes every
coordinate in $S$ in the order $\pi$. In particular, $(\zero_S,\one_L)$ is a chain
vertex, proving the remaining required zero.

The exact identity $\alpha_0=4b/(2b+1)^2$ also gives
\begin{equation}\label{eq:normalization-square}
 \frac1{(2b+1)^2}-\gamma_L(v)
 =\frac{b(2v-1)^2}{(2b+1)^2(b+v)}\ge0.
\end{equation}
Equality holds at $v=1/2\in[s,1]$. For a large coordinate,
$\gamma_i=\gamma_L(v_i)$, so $s(1-\alpha_0)\le\gamma_i\le(2b+1)^{-2}$. For
a small coordinate, the bounds above give
$0\le\gamma_i\le(1-\alpha_0)s=\gamma_L(s)\le(2b+1)^{-2}$. Thus
$0\le\gamma_i\le(2b+1)^{-2}$ for every $i$.
With $\widehat\gamma_\iota=\gamma_{\pi(\iota)}$ and $\widehat\gamma_{n+1}=0$, expanding the chain increments and collecting each chain value gives
\begin{align*}
 \sum_{\iota=1}^n\widehat\gamma_\iota d_{\pi(\iota)}
 &=\sum_{\iota=1}^n\widehat\gamma_\iota
       [f(\one_{C_\iota})-f(\one_{C_{\iota-1}})]\\
 &=\sum_{\iota=1}^n(\widehat\gamma_\iota-\widehat\gamma_{\iota+1})f(\one_{C_\iota})
       -\widehat\gamma_1 f(\zero).
\end{align*}
Since $\ell(h(x))-\alpha_0\ell(x)=(1-\alpha_0)f(\zero)+\sum_i\gamma_id_i$,
it follows that
\[
 \ell(h(x))-\alpha_0\ell(x)
 =(1-\alpha_0-\widehat\gamma_1)f(\zero)
   +\sum_{\iota=1}^n(\widehat\gamma_\iota-\widehat\gamma_{\iota+1})f(\one_{C_\iota})\ge0.
\]
The right side is nonnegative because $f\ge0$ at every chain vertex and
every coefficient is nonnegative: the order gives
$\widehat\gamma_\iota-\widehat\gamma_{\iota+1}\ge0$, and
$1-\alpha_0-\widehat\gamma_1\ge1-\alpha_0-(2b+1)^{-2}
=4b^2/(2b+1)^2>0$. Equal consecutive coefficients contribute zero,
so the same computation includes all ties.
\end{proof}

The next lemma carries out the computation in the proof of
\citet[Lemma~4.1]{geometry2026} for the slice $F(a(t),\cdot)$ of the
large block. The branch $r/(b+r)$ of \eqref{eq:action-map} is their
recovery map with parameter $b$, and the large coordinates are processed
in the same order, the nonincreasing order of $q_i$. We stop before
their final bound: that proof bounds the intermediate upper-prefix
values below by zero, whereas the constant term below keeps them as
$Y_i(t)$, because Lemmas~\ref{lem:scalar} and~\ref{lem:curvature} use them.
\begin{lemma}[Ordered large-coordinate comparison]\label{lem:ordered}
For every $o_L\in[0,1]^L$ and $t\in[0,1]$,
\begin{equation}\label{eq:large-comparison}
 F(a(t),o_L)\le \mathcal C(t)+\ip{g^L(t)}{o_L-v},
\end{equation}
where
\begin{equation}\label{eq:intercept-exact}
 \mathcal C(t)=\Lambda(q_m)y(t)-bS_0(t)-\vartheta(q_1)E(t)
       -\sum_{i=1}^{m-1}[\vartheta(q_{i+1})-\vartheta(q_i)]Y_i(t).
\end{equation}
\end{lemma}
\begin{proof}
Summing display~(18) in the proof of \citet[Lemma~4.1]{geometry2026}
over the large coordinates shows that every DR-submodular function $J$
on the large-coordinate cube satisfies the signed endpoint inequality
\begin{equation}\label{eq:signed-endpoint}
 J(o_L)\le J(p)-\sum_{i=1}^m p_iB_i^L(J)
       +\sum_{i=1}^m[(1-p_i)A_i^L(J)+p_iB_i^L(J)]o_i.
\end{equation}
Here the secants are those in \eqref{eq:large-secants} with $J$ in
place of $F(a(t),\cdot)$.
For $J=F(a(t),\cdot)$, centering the affine term at $v$ leaves
\[
 y(t)+\sum_{i=1}^m v_i(1-p_i)A_i^L(t)
       -\sum_{i=1}^m p_i(1-v_i)B_i^L(t).
\]
The first sum equals $b(y(t)-S_0(t))$, since $v_i(1-p_i)/p_i=b$ and the
lower-prefix value differences telescope from $S_0(t)$ to $y(t)$.
For the second sum, $p_i(1-v_i)/(1-p_i)=\vartheta(q_i)$, so its
contribution is $\sum_{i=1}^m\vartheta(q_i)(Y_i(t)-Y_{i-1}(t))$.
Collecting the coefficient of each $Y_i(t)$ gives
\begin{align*}
 \sum_{i=1}^m\vartheta(q_i)(Y_i(t)-Y_{i-1}(t))
 &=\sum_{i=1}^m\vartheta(q_i)Y_i(t)
      -\sum_{i=0}^{m-1}\vartheta(q_{i+1})Y_i(t)\\
 &=\vartheta(q_m)y(t)-\vartheta(q_1)E(t)
      -\sum_{i=1}^{m-1}[\vartheta(q_{i+1})-\vartheta(q_i)]Y_i(t).
\end{align*}
Substituting both sums into the centered intercept gives
\eqref{eq:intercept-exact}, since $1+b+\vartheta(q_m)=\Lambda(q_m)$.
\end{proof}

For the small-coordinate comparisons, fix $o=(o_S,o_L)$ and write
\[
 \widetilde o_S=a(t)+(1-a(t))o_S,\qquad
 E^*=e(o_S),\qquad \eta(t)=e^{-st}.
\]
The three boundary inequalities are
\begin{align}
 Q_t(o_S)&\ge\eta(t)E^*,\label{eq:boundary-survival}\\
 Q_t(o_S)&\ge bE(t)+\ip{\nu_t}{o_S-u},\label{eq:boundary-greedy}\\
 E'(t)+sE(t)&\ge0.\label{eq:boundary-drift}
\end{align}
When $S$ is empty, normalization gives $E(t)=E^*=0$ and all three
inequalities hold with zero small-coordinate terms. For nonempty $S$,
the bounds follow as below.
For \eqref{eq:boundary-survival}, the path definition and $u_j\le s$
give $0\le\widetilde o_S-o_S=a(t)(1-o_S)\le(1-\eta(t))(1-o_S)$.
If $t>0$, the point
$r=o_S+(\widetilde o_S-o_S)/(1-\eta(t))$ belongs to the small cube and
$\widetilde o_S=\eta(t)o_S+(1-\eta(t))r$. Concavity on the nonnegative
direction $r-o_S$, followed by $e(r)\ge0$, gives
$e(\widetilde o_S)\ge\eta(t)e(o_S)$. At $t=0$, the two sides coincide.

The map $r\mapsto a(t)+(1-a(t))r$ is affine with positive diagonal
coefficients. Consequently $Q_t$ is nonnegative and DR-submodular.
The greedy lower bound and its coordinatewise extension in
\eqref{eq:greedy-cube-lower}, applied to $Q_t$, give
$Q_t(o_S)\ge E(t)+\ip{\nu_t}{o_S}$.
The ordered small coordinates $j_1,\ldots,j_k$ and their prefixes
$B_0,\ldots,B_k$ determine the constant after centering at $u$.
Expansion of the increments and collection of each
$Q_t(\one_{B_\iota})$ coefficient give
\[
 E(t)+\ip{\nu_t}{u}
 =(1-u_{j_1})E(t)
   +\sum_{\iota=1}^{k-1}(u_{j_\iota}-u_{j_{\iota+1}})Q_t(\one_{B_\iota})
   +u_{j_k}Q_t(\one_{B_k})\ge bE(t).
\]
The order makes each coefficient nonnegative, the function values are
nonnegative, and $u_{j_1}\le s=1-b$. Adding
$\ip{\nu_t}{o_S-u}$ proves \eqref{eq:boundary-greedy}.
Along the path, $a_j'=u_j(1-a_j)$ gives
$E'(t)=\ip{u}{(1-a(t))\nabla e(a(t))}$.

For \eqref{eq:boundary-drift}, the set
$N_t=\{j\in S:\partial_j e(a(t))<0\}$ identifies a nonnegative
direction that increases precisely these coordinates to one.
Tangency in that direction yields
\[
 0\le e(a(t)+(1-a(t))\one_{N_t})
 \le E(t)+\sum_{j\in N_t}(1-a_j(t))\partial_j e(a(t)).
\]
Thus
\begin{equation}\label{eq:negative-budget}
 \sum_{j\in S}(1-a_j(t))(-\partial_j e(a(t)))_+\le E(t).
\end{equation}
Since $E'(t)=\sum_{j\in S}u_j(1-a_j(t))\partial_je(a(t))$ and
$0\le u_j\le s$,
\[
 E'(t)\ge-\sum_{j\in S}u_j(1-a_j(t))(-\partial_je(a(t)))_+
 \ge-s\sum_{j\in S}(1-a_j(t))(-\partial_je(a(t)))_+\ge-sE(t),
\]
where the last step is \eqref{eq:negative-budget}.

Normalization also gives a lower bound. For every small vector
$r\in[0,1]^S$ and every large vector $l\in[0,1]^L$,
\begin{equation}\label{eq:floor}
 F(r,l)-F(\zero_S,l)\ge F(r,\one_L)\ge0.
\end{equation}
The first inequality compares, by DR, the small-coordinate increments
from $\zero_S$ to $r$ at the large-coordinate values $l\le\one_L$, and uses
$F(\zero_S,\one_L)=0$; the second is $F\ge0$. This bound allows the
large block to be compared without any small coordinate, after
subtracting the same boundary value $E(t)$ that appears in the
small-coordinate comparisons.

The next lemma combines the large-block bound of
Lemma~\ref{lem:ordered} with the small-coordinate path. The proof of
Theorem~\ref{thm:structural} keeps every slack, so we state them as
identities with explicit nonnegative remainders. For small vectors $y$
and large vectors $l$, put
\[
 \Delta_l(y)=F(y,l)-F(\zero_S,l)-e(y),\qquad
 \Delta_o=\Delta_{o_L},\qquad \Delta_{\zero}=\Delta_{\zero_L}.
\]
Since $F(\zero_S,\one_L)=0$,
\[
 \Delta_l(y)=[F(y,l)-F(\zero_S,l)]-[F(y,\one_L)-F(\zero_S,\one_L)].
\]
By DR-submodularity, the small-block increment $F(y,l)-F(\zero_S,l)$ is
nonincreasing in $l$, and the $y$-gradient
$\nabla_SF(y,l)-\nabla_SF(y,\one_L)$ of $\Delta_l$ is nonnegative because
$l\le\one_L$. Hence $\Delta_l$ is nondecreasing in $y$ and
nonincreasing in $l$, and $\Delta_l(\zero_S)=\Delta_{\one_L}=0$. Hence
$\Delta_o$ and $\Delta_{\zero}-\Delta_o$ are nonnegative and nondecreasing; the
bound \eqref{eq:floor} is the statement $\Delta_l\ge0$. Along the large
prefixes, the upper and lower increments
\[
 \mathcal X_i^+(y)=\Delta_{P_i^+}(y)-\Delta_{P_{i-1}^+}(y),\qquad
 \mathcal X_i^-(y)=\Delta_{P_{i-1}^-}(y)-\Delta_{P_i^-}(y)
\]
vanish at $\zero_S$ and are nonnegative and nondecreasing, and so is
\[
 \Xi_o(y)=\sum_{i=1}^m\frac{q_i-v_i(1-o_i)}b\mathcal X_i^+(y)
          +\sum_{i=1}^m\frac b{v_i}o_i\mathcal X_i^-(y),
\]
whose coefficients are nonnegative because $q_i\ge v_i$. The
monotonicity of $\mathcal X_i^\pm$ follows in the same way: since
$P_i^+\le P_{i-1}^+$ and $P_{i-1}^-\le P_i^-$, their $y$-gradients
$\nabla_SF(y,P_i^+)-\nabla_SF(y,P_{i-1}^+)$ and
$\nabla_SF(y,P_{i-1}^-)-\nabla_SF(y,P_i^-)$ are nonnegative. Let
$\mathcal L_o(y)\ge0$ be the gap in \eqref{eq:signed-endpoint} for
$J=F(y,\cdot)$ at $o_L$. For a large vector $P$, the tangency remainder
\[
 \mathcal T_P(t)=F(a(t),P)+\ip{(1-a(t))\nabla_SF(a(t),P)}{o_S}
                -F(\widetilde o_S,P)
\]
is nonnegative by concavity along the nonnegative direction
$(1-a(t))o_S$.
\begin{lemma}[Path comparisons with remainders]\label{lem:two-comparisons}
For every comparator $o=(o_S,o_L)$ and every $t\in[0,1]$,
\begin{align}
 F(o)+\mathcal S_1(t)={}&\mathcal C(t)+\ip{g^L(t)}{o_L-v}+\ip{(1-a(t))V(t)}{o_S}
     \nonumber\\
   &+(c_1-1)E(t)+E^*-c_1Q_t(o_S),\label{eq:tangent-comparison}\\
 F(o)+\mathcal S_2(t)={}&\mathcal C(t)+\ip{g^L(t)}{o_L-v}
     +\ip{(1-a(t))\nabla_SF(a(t),\zero_L)}{o_S}\nonumber\\
   &+E^*-Q_t(o_S),\label{eq:join-comparison}\\
 F(o_S,\zero_L)+\mathcal S_3(t)={}&S_0(t)
     +\ip{(1-a(t))\nabla_SF(a(t),\zero_L)}{o_S}+E^*-Q_t(o_S),
     \label{eq:small-only-comparison}\\
 F(\zero_S,o_L)+\mathcal S_4(t)={}&\mathcal C(t)-E(t)+\ip{g^L(t)}{o_L-v},
     \label{eq:direct-comparison}\\
 F(o)+\mathcal S_5={}&F(o_S,\zero_L)+F(\zero_S,o_L),
     \label{eq:rectangle-comparison}
\end{align}
with the nonnegative remainders
\begin{align*}
 \mathcal S_1(t)&=\mathcal L_o(\widetilde o_S)+\Xi_o(\widetilde o_S)-\Xi_o(a(t))
   +\Delta_o(\widetilde o_S)-\Delta_o(o_S)
   +\sum_{i=1}^m\chi_i\mathcal T_{P_i^+}(t),\\
 \mathcal S_2(t)&=\mathcal L_o(a(t))+\Delta_o(\widetilde o_S)-\Delta_o(o_S)
   +(\Delta_{\zero}-\Delta_o)(\widetilde o_S)-(\Delta_{\zero}-\Delta_o)(a(t))
   +\mathcal T_{\zero_L}(t),\\
 \mathcal S_3(t)&=\Delta_{\zero}(\widetilde o_S)-\Delta_{\zero}(o_S)+\mathcal T_{\zero_L}(t),\qquad
 \mathcal S_4(t)=\Delta_o(a(t))+\mathcal L_o(a(t)),\\
 \mathcal S_5&=(\Delta_{\zero}-\Delta_o)(o_S).
\end{align*}
\end{lemma}
The last three identities treat the blocks separately: the small block
with the large coordinates at zero, the large block with the small
coordinates at zero, and the rectangle inequality that adds the two.
\begin{proof}
Each remainder is nonnegative because $\widetilde o_S\ge a(t)$ and
$\widetilde o_S\ge o_S$, the functions $\Delta_o$, $\Delta_{\zero}-\Delta_o$, and
$\Xi_o$ are nondecreasing, and $\mathcal L_o,\mathcal T_P\ge0$. The
identities are rearrangements. The definition of $\mathcal L_o$ and the
centered intercept in the proof of Lemma~\ref{lem:ordered} give
\begin{equation}\label{eq:large-gap}
 F(a(t),o_L)=\mathcal C(t)+\ip{g^L(t)}{o_L-v}-\mathcal L_o(a(t)).
\end{equation}
For \eqref{eq:rectangle-comparison}, express the four values through
$\Delta_{\zero}$ and $\Delta_o$ and use $F(\zero)=0$. For
\eqref{eq:direct-comparison}, the definition of $\Delta_o$ gives
$F(\zero_S,o_L)=F(a(t),o_L)-E(t)-\Delta_o(a(t))$; substitute
\eqref{eq:large-gap}.

For the remaining identities, $F(y,l)-e(y)=\Delta_l(y)+F(\zero_S,l)$ gives
\begin{equation}\label{eq:join-start}
 F(o)=F(\widetilde o_S,o_L)+E^*-Q_t(o_S)
      -[\Delta_o(\widetilde o_S)-\Delta_o(o_S)],
\end{equation}
and likewise
$F(o_S,\zero_L)=F(\widetilde o_S,\zero_L)+E^*-Q_t(o_S)
-[\Delta_{\zero}(\widetilde o_S)-\Delta_{\zero}(o_S)]$. The definition of
$\mathcal T_{\zero_L}$, written as
$F(\widetilde o_S,\zero_L)=S_0(t)+\ip{(1-a(t))\nabla_SF(a(t),\zero_L)}{o_S}
-\mathcal T_{\zero_L}(t)$, proves \eqref{eq:small-only-comparison}.
Put $\Psi(l)=F(\widetilde o_S,l)-F(a(t),l)$. Then
$\Psi(\zero_L)-\Psi(o_L)=(\Delta_{\zero}-\Delta_o)(\widetilde o_S)
-(\Delta_{\zero}-\Delta_o)(a(t))$. Write
$F(\widetilde o_S,o_L)=F(a(t),o_L)+\Psi(\zero_L)-[\Psi(\zero_L)-\Psi(o_L)]$.
Substitute \eqref{eq:large-gap} for $F(a(t),o_L)$ and
$\Psi(\zero_L)=F(\widetilde o_S,\zero_L)-S_0(t)
=\ip{(1-a(t))\nabla_SF(a(t),\zero_L)}{o_S}-\mathcal T_{\zero_L}(t)$, and
insert the result in \eqref{eq:join-start}. This gives
\[
\begin{aligned}
 F(o)={}&\mathcal C(t)+\ip{g^L(t)}{o_L-v}
 +\ip{(1-a(t))\nabla_SF(a(t),\zero_L)}{o_S}+E^*-Q_t(o_S)\\
 &-\mathcal L_o(a(t))-\mathcal T_{\zero_L}(t)
 -[(\Delta_{\zero}-\Delta_o)(\widetilde o_S)-(\Delta_{\zero}-\Delta_o)(a(t))]
 -[\Delta_o(\widetilde o_S)-\Delta_o(o_S)].
\end{aligned}
\]
The remainder terms sum to $-\mathcal S_2(t)$, which proves
\eqref{eq:join-comparison}.

For \eqref{eq:tangent-comparison}, apply the definition of
$\mathcal L_o$ to $F(\widetilde o_S,\cdot)=F(a(t),\cdot)+\Psi$. The right side
of \eqref{eq:signed-endpoint} is linear in the function; write
$\mathcal A_\Psi(o_L)$ for its value with $J=\Psi$. Then
\[
 F(\widetilde o_S,o_L)=\mathcal C(t)+\ip{g^L(t)}{o_L-v}+\mathcal A_\Psi(o_L)
                       -\mathcal L_o(\widetilde o_S).
\]
The secant differences of $\Psi$ are increments of $\mathcal X_i^\pm$:
$\Psi(P_i^+)-\Psi(P_{i-1}^+)=\mathcal X_i^+(\widetilde o_S)-\mathcal X_i^+(a(t))$
and
$\Psi(P_{i-1}^-)-\Psi(P_i^-)=\mathcal X_i^-(\widetilde o_S)-\mathcal X_i^-(a(t))$.
Since $(1-p_i)/p_i=b/v_i$ and $p_i/(1-p_i)=v_i/b$, the definition of
the secants gives
\[
 \mathcal A_\Psi(o_L)=\Psi(p)-\sum_{i=1}^m\frac b{v_i}o_i\bigl[\Psi(P_{i-1}^-)-\Psi(P_i^-)\bigr]
 +\sum_{i=1}^m\frac{v_i(1-o_i)}b\bigl[\Psi(P_i^+)-\Psi(P_{i-1}^+)\bigr].
\]
Writing $v_i(1-o_i)=q_i-[q_i-v_i(1-o_i)]$ and using the two increment
identities above and the definition of $\Xi_o$ gives
\[
 \mathcal A_\Psi(o_L)=\Psi(p)+\sum_{i=1}^m\frac{q_i}b[\Psi(P_i^+)-\Psi(P_{i-1}^+)]
       -[\Xi_o(\widetilde o_S)-\Xi_o(a(t))].
\]
With $Z_i=\Psi(P_i^+)$, so that $Z_0=\Psi(\one_L)$ and $Z_m=\Psi(p)$, summation
by parts gives
$\sum_{i=1}^m\frac{q_i}b(Z_i-Z_{i-1})=\frac{q_m}bZ_m
+\sum_{i=1}^{m-1}\frac{q_i-q_{i+1}}bZ_i-\frac{q_1}bZ_0$. Since $q_1=1$,
$1/b=c_1-1$ and $1+q_m/b=\chi_m$,
\[
 \Psi(p)+\sum_{i=1}^m\frac{q_i}b[\Psi(P_i^+)-\Psi(P_{i-1}^+)]
 =\sum_{i=1}^m\chi_i\Psi(P_i^+)-(c_1-1)\Psi(\one_L).
\]
Finally $\Psi(\one_L)=Q_t(o_S)-E(t)$, and
$\Psi(P_i^+)=\ip{(1-a(t))\nabla_SF(a(t),P_i^+)}{o_S}-\mathcal T_{P_i^+}(t)$ by
the definition of the tangency remainder. Substituting these expressions
in \eqref{eq:join-start} proves \eqref{eq:tangent-comparison}.
\end{proof}

\subsection{Scalar averaging and ordered path representations}\label{sec:curvature}

Section~\ref{sec:comparison-proof} applies the next lemma to step
profiles $q(t)$ built from the nested intervals of
Lemma~\ref{lem:curvature}. Each such profile takes the values
$q_1\ge\cdots\ge q_m$ in order, so it is nonincreasing.
\begin{lemma}[Scalar averaging inequality]\label{lem:scalar}
The weights in \eqref{eq:beta-c}--\eqref{eq:mixing} satisfy
\eqref{eq:mass-identities}. The functions $\theta_+$ and $\theta_-$ cross exactly once,
at $z\in(3/4,1)$; the function $\theta$ is absolutely continuous, and
\begin{equation}\label{eq:theta-identities}
\begin{gathered}
 \theta(0)=\Gamma e^{-s}-b,\qquad\theta(1)=\Gamma,\qquad
 \theta'(t)=s\theta(t)+b\psi(t)\ \text{for a.e. }t,\\
 0\le\psi(t)\le\rho_1(t),\qquad
 \int_0^1e^{-st}\psi(t)\,dt=1.
\end{gathered}
\end{equation}
Define
\begin{equation}\label{eq:budget}
 \mathcal B(\tau)=e^{-s\tau}\int_0^\tau[\theta(t)-\rho_1(t)]e^{st}\,dt.
\end{equation}
Then for every $0\le\tau\le T\le1$ and every nonincreasing measurable
$q:[\tau,T]\to[1/2,1]$,
\begin{equation}\label{eq:scalar}
 \int_\tau^T\left\{\Lambda(q(t))W(t)+\cprof(q(t))\lambda(t)-\Gamma\right\}dt
 \le\mathcal B(\tau).
\end{equation}
\end{lemma}
For $t\ge t_0$ the integrand in \eqref{eq:scalar} equals
$\cprof(q)[w(t)+(1-q)W(t)]-\Gamma$, since $\lambda=w-bW$ there and
$\Lambda(q)=\cprof(q)(b+1-q)$. Before $t_0$ it equals
$\Lambda(q)-\Gamma=-(q-1/2)^2/b\le0$, so the initial interval, on which
only the small-block comparison is used, contributes a nonpositive
amount to this inequality.
The proof is in Appendix~\ref{app:scalar}.

The next lemma represents the prefix paths by nested intervals, to
which Lemma~\ref{lem:scalar} applies. Its notation separates each path
into its initial value, the common boundary value, and an increment:
\begin{equation}\label{eq:path-decomposition}
 I_i=Y_i(0),\qquad K_i(t)=Y_i(t)-I_i-E(t),\qquad
 D_i(t)=K_i(t)-K_{i-1}(t).
\end{equation}
Here $I_i,K_i$ are indexed by $0\le i\le m$, and $D_i$ by
$1\le i\le m$. Normalization gives $I_0=E(0)=0$, so $K_0=0$ and
$K_i(0)=0$. Also $I_i=F(\zero_S,P_i^+)\ge0$ because $F\ge0$. For
$i\ge1$, DR and
$P_i^+\le P_{i-1}^+$ imply $Y_i'(t)\ge Y_{i-1}'(t)$.
Consequently $D_i'(t)=Y_i'(t)-Y_{i-1}'(t)\ge0$, and integrating
from zero gives $D_i\ge0$ as well.
In particular, $Y_i(t)=I_i+K_i(t)+E(t)\ge E(t)$ for every prefix.

Adding a common correction to the derivatives $Y_i'(t)$ gives functions
$R_i(t)$ that are ordered in $i$ and nonincreasing in $t$. Each set
$\{t:R_i(t)>\omega\}$ is then an interval starting at $0$, and these
intervals are nested in $i$. This nesting produces the nonincreasing
profiles in Lemma~\ref{lem:scalar}.
\begin{lemma}[Ordered integral representation]\label{lem:curvature}
For $0\le i\le m$, the functions
\begin{equation}\label{eq:R-functions}
 R_i(t)=Y_i'(t)+sE(t)+s\int_t^1[E'(r)+sE(r)]dr
\end{equation}
are nonnegative, nonincreasing in $t$, and nondecreasing in $i$. Put
$H_R=R_m(0)$ and, for a level $\omega\in[0,H_R]$,
$T_i(\omega)=\sup\{t\in[0,1]:R_i(t)>\omega\}$ (with $\sup\varnothing=0$)
and $\tau(\omega)=T_0(\omega)$. These endpoints are measurable and satisfy
$0\le\tau(\omega)\le T_1(\omega)\le\cdots\le T_m(\omega)\le1$. The corresponding paths are
\begin{equation}\label{eq:unit-layer}
 E_\tau(t)=
 \frac{e^{-s(t+\tau)}(e^{2s\min\{t,\tau\}}-1)}{2s},
 \qquad
 K_{i,\omega}(t)=\min\{t,T_i(\omega)\}-\min\{t,\tau(\omega)\}.
\end{equation}
Then the paths have the representations
\[
 E(t)=\int_0^{H_R}E_{\tau(\omega)}(t)\,d\omega,\qquad
 K_i(t)=\int_0^{H_R}K_{i,\omega}(t)\,d\omega.
\]
Moreover $E_\tau'+sE_\tau=r_{0,\tau}$ almost everywhere, where
$r_{0,\tau}(t)=e^{-s(\tau-t)}\one\{t<\tau\}$.
\end{lemma}
\begin{proof}
For a fixed large-coordinate vector $P\le\one_L$, the path
$Y(t)=F(a(t),P)$ satisfies the following two bounds. DR gives
$\partial_jF(a(t),P)\ge\partial_je(a(t))$ for every $j\in S$.
Using $0\le u_j\le s$ and \eqref{eq:negative-budget}, we obtain
\begin{equation}\label{eq:common-derivative-lower}
 Y'(t)=\sum_{j\in S}u_j e^{-tu_j}\partial_jF(a(t),P)
 \ge-s\sum_{j\in S}e^{-tu_j}(-\partial_je(a(t)))_+
 \ge-sE(t).
\end{equation}
For $\delta>0$ with $t+\delta\le1$, DR along $a(t+\delta)\ge a(t)$ gives
\begin{align*}
 Y'(t+\delta)-Y'(t)
 &\le\sum_{j\in S}u_je^{-tu_j}(e^{-\delta u_j}-1)
                         \partial_jF(a(t),P)\\
 &\le\sum_{j\in S}u_je^{-tu_j}(1-e^{-\delta u_j})
                         (-\partial_je(a(t)))_+\\
 &\le \delta\sum_{j\in S}u_j^2e^{-tu_j}
                         (-\partial_je(a(t)))_+
 \le \delta s^2E(t).
\end{align*}
In the second line, terms with nonnegative $\partial_jF$ are dropped;
for the others, $(-\partial_jF)_+\le(-\partial_je)_+$.
The third line uses $1-e^{-\delta u_j}\le\delta u_j$.
For a partition of $[r,t]$ into cells $[\varsigma,\varsigma']$, summing this bound
over the cells, each taken at its left endpoint, gives
\[
 Y'(t)-Y'(r)\le s^2\sum_{[\varsigma,\varsigma']}(\varsigma'-\varsigma)E(\varsigma).
\]
Continuity of $E$ makes the right side converge to its Riemann integral
as the partition mesh tends to zero. Therefore
\begin{equation}\label{eq:integrated-curvature}
 Y'(t)-Y'(r)\le s^2\int_r^t E(\varsigma)\,d\varsigma.
\end{equation}
Subtracting \eqref{eq:R-functions} at $r$ from its value at $t$
gives, for $r<t$,
\[
\begin{aligned}
 R_i(t)-R_i(r)&=Y_i'(t)-Y_i'(r)+s[E(t)-E(r)]-s\int_r^t[E'(\varsigma)+sE(\varsigma)]\,d\varsigma\\
 &=Y_i'(t)-Y_i'(r)-s^2\int_r^tE(\varsigma)\,d\varsigma\le0,
\end{aligned}
\]
where the last step is \eqref{eq:integrated-curvature}.
Nonnegativity follows from \eqref{eq:common-derivative-lower} and
$E'(t)+sE(t)\ge0$. The prefix order gives
$Y_i'(t)\ge Y_{i-1}'(t)$, hence $R_i(t)\ge R_{i-1}(t)$.

For $\omega\ge0$, the endpoint is
$T_i(\omega)=\sup\{t\in[0,1]:R_i(t)>\omega\}$, with supremum zero
for the empty set. Each $T_i$ is nonincreasing in $\omega$ and hence
Borel measurable.
Because $R_i$ is nonincreasing, the set in braces is an
initial interval, up to its endpoint.
The indicator of this interval therefore equals
$\one\{t<T_i(\omega)\}$ for almost every $(t,\omega)$.
The elementary identity
$r=\int_0^\infty\one\{\omega<r\}\,d\omega$ for $r\ge0$ gives
\begin{equation}\label{eq:layer-identity}
 R_i(t)=\int_0^{H_R}\one\{t<T_i(\omega)\}\,d\omega
 \quad\text{for almost every }t.
\end{equation}
The pointwise order of the $R_i$ gives the order of the $T_i$.
Continuity of the first derivatives makes each $R_i$ continuous and
bounded. Their order in $i$ and monotonicity in $t$ give
$0\le R_i(t)\le R_m(0)=H_R$, which explains the common finite
integration range in \eqref{eq:layer-identity}.

The relation between these functions and the paths is linear. First,
$K_i'(t)=Y_i'(t)-E'(t)=R_i(t)-R_0(t)$ and $K_i(0)=0$, so integrating
\eqref{eq:layer-identity} from zero to $t$ gives
\[
 K_i(t)=\int_0^{H_R}
   [\min\{t,T_i(\omega)\}-\min\{t,T_0(\omega)\}]\,d\omega.
\]
Second, with $r_0(t)=E'(t)+sE(t)$, the equation for $R_0$ is
\[
 R_0(t)=r_0(t)+s\int_t^1r_0(r)dr.
\]
With $\Phi(t)=\int_t^1r_0(r)\,dr$, this reads $R_0=-\Phi'+s\Phi$ with
$\Phi(1)=0$. Since
$\frac{d}{dt}\bigl[e^{-st}\Phi(t)\bigr]=e^{-st}(\Phi'-s\Phi)(t)=-e^{-st}R_0(t)$,
integrating from $t$ to $1$ gives
$\Phi(t)=\int_t^1e^{-s(r-t)}R_0(r)\,dr$, and
$r_0=-\Phi'=R_0-s\Phi$ gives
\begin{equation}\label{eq:volterra-inverse}
 r_0(t)=R_0(t)-s\int_t^1e^{-s(r-t)}R_0(r)dr.
\end{equation} For one indicator
$R_0(t)=\one\{t<\tau\}$, formula \eqref{eq:volterra-inverse} becomes
\[
 r_{0,\tau}(t)=e^{-s(\tau-t)}\one\{t<\tau\}.
\]
Solving $E_\tau'+sE_\tau=r_{0,\tau}$ with $E_\tau(0)=0$ gives
\[
 E_\tau(t)=e^{-st}\int_0^{\min\{t,\tau\}}e^{sr}e^{-s(\tau-r)}dr,
\]
which evaluates to \eqref{eq:unit-layer}. Applying the same bounded
linear operations to \eqref{eq:layer-identity} and interchanging its
finite integrals yields
$E(t)=\int_0^{H_R}E_{T_0(\omega)}(t)\,d\omega$.
This proves both asserted representations.
\end{proof}

\subsection{Proof of Theorem~\ref{thm:structural}}\label{sec:comparison-proof}
The proof averages the identities of Lemma~\ref{lem:two-comparisons} in
time with the weights of Section~\ref{sec:field} and keeps their
remainders. Two features of the schedule make the averaged constant term
computable. The small-block weight $\kappa=bW$ is $b$ times the tail of
the large-block measure $\mu$, so the change of the cross-block
difference $\Delta_{\zero}$ along the small path enters the two blocks with
opposite signs and cancels after integration. The weight
$\psi$ removes the comparator boundary value $E^*$.
\begin{proof}[Proof of Theorem~\ref{thm:structural}]
Fix $o=(o_S,o_L)$, and abbreviate $\Delta_{\zero}(t)=\Delta_{\zero}(a(t))$ and
$\Delta_o(t)=\Delta_o(a(t))$. Since $F(\zero)=0$, the definition of
$\Delta_l$ gives $S_0(t)=\Delta_{\zero}(t)+E(t)$.

We first form a weighted sum. Multiply \eqref{eq:tangent-comparison} by
$\lambda(t)$, \eqref{eq:join-comparison} by $\kappa(t)-\kappa_0(t)$,
\eqref{eq:small-only-comparison} by $\kappa_0(t)$, and
\eqref{eq:direct-comparison} by $\varpi(t)$, integrate over $[0,1]$, and
add $\beta$ times \eqref{eq:rectangle-comparison}. All weights are
nonnegative by \eqref{eq:mass-identities}. Since
$\int(\lambda+\kappa-\kappa_0)=1-\beta$ and
$\int\kappa_0=\int\varpi=\beta$, after cancelling
$\beta F(o_S,\zero_L)+\beta F(\zero_S,o_L)$, which appears on both sides,
the left side is $F(o)+S_{\mathrm{int}}$, where
\[
 S_{\mathrm{int}}=\int_0^1\bigl[\lambda\mathcal S_1+(\kappa-\kappa_0)\mathcal S_2
       +\kappa_0\mathcal S_3+\varpi\mathcal S_4\bigr]dt+\beta\mathcal S_5\ge0.
\]
On the right side, the large-block comparisons carry the total weight
$\lambda+\kappa-\kappa_0+\varpi=w+\varpi$, so their linear terms form the
large part of $\ip{g(f,x)-d}{o-x}$. The small-block gradient terms carry
$\lambda(1-a(t))V(t)+\kappa(1-a(t))\nabla_SF(a(t),\zero_L)$. Centering
them at $u$ gives the small part of the field, except for its term
$-U\nu_t$, together with the constants
\[
 \ip{u}{(1-a(t))V(t)}=c_1E'(t)+\sum_{i=1}^m\cprof(q_i)D_i'(t),\qquad
 \ip{u}{(1-a(t))\nabla_SF(a(t),\zero_L)}=\Delta_{\zero}'(t)+E'(t).
\]
Both follow from $a_j'=u_j(1-a_j)$, which gives
$\ip{u}{(1-a(t))\nabla_SF(a(t),P)}=\frac{d}{dt}F(a(t),P)$ for every large
vector $P$. For the first, $\frac d{dt}Y_i=E'+\sum_{\iota\le i}D_\iota'$ and
$\sum_{i\ge\iota}\chi_i=\cprof(q_\iota)$; for the second, $S_0=\Delta_{\zero}+E$.
Next we rewrite $\mathcal C$. In \eqref{eq:intercept-exact}, substitute
$Y_i=I_i+K_i+E$ and $S_0=\Delta_{\zero}+E$. The coefficient of $E$ becomes
$\Lambda(q_m)-b-\vartheta(q_1)-\sum_{i=1}^{m-1}[\vartheta(q_{i+1})-\vartheta(q_i)]
=\Lambda(q_m)-b-\vartheta(q_m)=1$. Since
$\Lambda(q_{i+1})-\Lambda(q_i)=\vartheta(q_{i+1})-\vartheta(q_i)$ and
$I_0=K_0=0$, summation by parts gives, for $Z=I$ and for $Z=K$,
\[
 \Lambda(q_m)Z_m-\sum_{i=1}^{m-1}[\vartheta(q_{i+1})-\vartheta(q_i)]Z_i
 =\sum_{i=1}^m\Lambda(q_i)(Z_i-Z_{i-1}).
\]
Therefore
\[
 \mathcal C(t)=C_I+E(t)+\sum_{i=1}^m\Lambda(q_i)D_i(t)-b\Delta_{\zero}(t),\qquad
 C_I=\sum_{i=1}^m\Lambda(q_i)(I_i-I_{i-1}).
\]
Since $D_i(0)=\Delta_{\zero}(0)=0$ and $W(t)=\mu((t,1])$, Fubini's
theorem gives
\[
 \int_0^1D_i\,d\mu=\int_0^1WD_i'\,dt,\qquad
 \int_0^1\Delta_{\zero}\,d\mu=\int_0^1W\Delta_{\zero}'\,dt.
\]
Hence the term $-b\int W\Delta_{\zero}'$ from the large-block comparisons and the term
$\int\kappa\Delta_{\zero}'$ from the centered small-block gradient terms cancel, because
$\kappa=bW$, while the small-only comparison adds $\int\kappa_0\Delta_{\zero}$
through $S_0=\Delta_{\zero}+E$. Among the remaining terms, the coefficient of
$E$ collects $w+\varpi$ from $\mathcal C$, $(c_1-1)\lambda$ from
\eqref{eq:tangent-comparison}, $-\varpi$ from \eqref{eq:direct-comparison},
and $\kappa_0$ from $S_0=\Delta_{\zero}+E$; since $w+\kappa_0=\lambda+\kappa$,
\[
 (w+\varpi)+(c_1-1)\lambda-\varpi+\kappa_0=c_1\lambda+\kappa=\rho_1 .
\]
The derivative $E'$ receives $c_1\lambda+\kappa=\rho_1$ from the two
centered constants. The comparator value $E^*$ receives
$\lambda+(\kappa-\kappa_0)+\kappa_0$, whose integral is
$\int(\lambda+\kappa)=1$, and $Q_t(o_S)$ receives
$-c_1\lambda-(\kappa-\kappa_0)-\kappa_0=-\rho_1$. With
\begin{equation}\label{eq:IK}
 I_K=\int_0^1\sum_{i=1}^m
       [\Lambda(q_i)W(t)+\cprof(q_i)\lambda(t)]D_i'(t)\,dt,
\end{equation}
the weighted sum is the identity
\begin{align}
 F(o)+S_{\mathrm{int}}={}&C_I+\ip{g(f,x)-d}{o-x}+I_K
       +\int_0^1\kappa_0\Delta_{\zero}\,dt
       +\int_0^1U(t)\ip{\nu_t}{o_S-u}\,dt\nonumber\\
 &+\int_0^1\rho_1(t)[E'(t)+E(t)]\,dt+E^*
       -\int_0^1\rho_1(t)Q_t(o_S)\,dt.
 \label{eq:weighted-sum}
\end{align}

Next we treat the boundary terms. Put $r_0=E'+sE$, which is nonnegative by
\eqref{eq:boundary-drift}. The greedy minorant in the proof of
\eqref{eq:boundary-greedy} gives
$Q_t(o_S)=E(t)+\ip{\nu_t}{o_S}+\mathcal Y(t)$ with $\mathcal Y(t)\ge0$, and
the same proof gives $E(t)+\ip{\nu_t}u\ge bE(t)$, that is,
$\ip{\nu_t}u+sE(t)\ge0$ because $b=1-s$. By
\eqref{eq:boundary-survival} and $\int_0^1e^{-st}\psi(t)\,dt=1$, we have
$\int_0^1\psi(t)Q_t(o_S)\,dt\ge E^*$. Hence
\[
 B_{\mathrm{rem}}=\int_0^1U(t)\bigl[\ip{\nu_t}u+sE(t)+\mathcal Y(t)\bigr]dt
   +\Bigl[\int_0^1\psi(t)Q_t(o_S)\,dt-E^*\Bigr]\ge0.
\]
We now evaluate the boundary terms of \eqref{eq:weighted-sum}. Since
$Q_t(o_S)-\ip{\nu_t}{o_S-u}=E(t)+\ip{\nu_t}u+\mathcal Y(t)$,
$\rho_1=\psi+U$ and $1-s=b$,
\[
 \int_0^1\rho_1Q_t(o_S)\,dt-\int_0^1U\ip{\nu_t}{o_S-u}\,dt-E^*
 =B_{\mathrm{rem}}+b\int_0^1UE\,dt .
\]
Next, $\theta'=s\theta+b\psi$ gives $(\theta E)'=\theta r_0+b\psi E$.
Integrating over $[0,1]$ with $E(0)=0$ and $\theta(1)=\Gamma$,
\[
 \Gamma E(1)=\int_0^1\theta r_0\,dt+b\int_0^1\psi E\,dt .
\]
Finally, $\rho_1(E'+E)=\rho_1r_0+b\rho_1E$. Adding the three displays,
the terms in $E$ carry the weight $b(\psi+U-\rho_1)=0$, and
$\theta-\rho_1=\sigma$ gives
\begin{equation}\label{eq:boundary-integration}
 \begin{aligned}
 &\Gamma E(1)-\int_0^1\rho_1(E'+E)\,dt-\int_0^1U\ip{\nu_t}{o_S-u}\,dt
 +\int_0^1\rho_1Q_t(o_S)\,dt-E^*\\
 &\qquad=\int_0^1\sigma r_0\,dt+B_{\mathrm{rem}},
 \end{aligned}
\end{equation}
where $\sigma=\theta-\rho_1$ is the weight in \eqref{eq:mixing}.

We can now conclude. Since $F(h(x))=Y_m(1)=I_m+K_m(1)+E(1)$, subtracting
\eqref{eq:weighted-sum} from $\Gamma F(h(x))+\ip{g(f,x)-d}{o-x}$ and
applying \eqref{eq:boundary-integration} give
\begin{align}
 &\Gamma F(h(x))+\ip{g(f,x)-d}{o-x}-F(o)\nonumber\\
 &\quad=(\Gamma I_m-C_I)+(\Gamma K_m(1)-I_K)+\int_0^1\sigma r_0\,dt
       +\Bigl(S_{\mathrm{int}}-\int_0^1\kappa_0\Delta_{\zero}\,dt\Bigr)
       +B_{\mathrm{rem}}.
 \label{eq:total-intercept}
\end{align}
The first, fourth, and fifth terms are nonnegative. For the first, since
$I_0=0$, summation by parts gives
\begin{equation}\label{eq:initial-payment}
 C_I=\Lambda(q_m)I_m
       -\sum_{i<m}[\vartheta(q_{i+1})-\vartheta(q_i)]I_i
 \le\Gamma I_m,
\end{equation}
because $I_i\ge0$, $\vartheta$ decreases on $[1/2,1]$ so that
$\vartheta(q_{i+1})-\vartheta(q_i)\ge0$ along $q_1\ge\cdots\ge q_m\ge1/2$,
and $\Lambda(q_m)-\Gamma=-(q_m-1/2)^2/b\le0$. For the
fourth, write $\beta\mathcal S_5=\int_0^1\kappa_0\mathcal S_5\,dt$ and use
the identity
$\mathcal S_3+\mathcal S_5-\mathcal S_2=\Delta_{\zero}(t)-\Delta_o(t)-\mathcal L_o(a(t))$,
which follows from the definitions of the remainders (the terms
$\mathcal T_{\zero_L}$ and those evaluated at $\widetilde o_S$ and $o_S$
cancel). Since $(\kappa-\kappa_0)\mathcal S_2+\kappa_0(\mathcal S_3+\mathcal S_5)
=\kappa\mathcal S_2+\kappa_0(\mathcal S_3+\mathcal S_5-\mathcal S_2)$,
\[
 S_{\mathrm{int}}=\int_0^1\bigl[\lambda\mathcal S_1+\kappa\mathcal S_2
   +\kappa_0(\Delta_{\zero}-\Delta_o-\mathcal L_o(a(t)))
   +\varpi(\Delta_o+\mathcal L_o(a(t)))\bigr]dt.
\]
Since $\mathcal S_1\ge0$ and $\mathcal S_2\ge\mathcal L_o(a(t))$,
\[
 S_{\mathrm{int}}\ge\int_0^1\bigl[\kappa_0\Delta_{\zero}+(\varpi-\kappa_0)\Delta_o
 +(\kappa-\kappa_0+\varpi)\mathcal L_o(a(t))\bigr]dt,
\]
and the last term is nonnegative because $\kappa\ge\kappa_0$ and
$\mathcal L_o\ge0$. Hence
\begin{equation}\label{eq:retained-source}
 \begin{aligned}
 S_{\mathrm{int}}
 &\ge\int_0^1\kappa_0\Delta_{\zero}\,dt+\int_0^1(\varpi-\kappa_0)\Delta_o\,dt\\
 &=\int_0^1\kappa_0\Delta_{\zero}\,dt
   +\int_0^1\Bigl[\int_t^1(\varpi-\kappa_0)(r)\,dr\Bigr]\Delta_o'(t)\,dt
 \ge\int_0^1\kappa_0\Delta_{\zero}\,dt.
 \end{aligned}
\end{equation}
The equality integrates by parts with $\Delta_o(0)=0$, and the last step
uses $\Delta_o'\ge0$ and the tail inequality in \eqref{eq:mass-identities}.
The fifth term is $B_{\mathrm{rem}}\ge0$. By \eqref{eq:total-intercept},
it remains to prove
\begin{equation}\label{eq:intercept-reduced}
 I_K-\Gamma K_m(1)\le\int_0^1\sigma(t)[E'(t)+sE(t)]\,dt.
\end{equation}
Lemma~\ref{lem:curvature} expresses both sides over the same range of
levels, where the scalar inequality compares their integrands.

For a level $\omega\in[0,H_R]$ in Lemma~\ref{lem:curvature},
write $T_i=T_i(\omega)$, $\tau=T_0$, and $T=T_m$.
For almost every $t$,
\[
 K_{i,\omega}'(t)-K_{i-1,\omega}'(t)
 =\one\{t<T_i\}-\one\{t<T_{i-1}\}
 =\one\{T_{i-1}\le t<T_i\}.
\]
These disjoint intervals partition $[\tau,T)$, apart from endpoints;
an interval is empty when its two endpoints coincide.
Also $K_{m,\omega}(1)=T-\tau$.
The profile $q(t)$ equals $q_i$ on $[T_{i-1},T_i)$ and $q_m$
at $T$.
The index $i$ increases with $t$, so $q(t)$ is nonincreasing.
Thus the contribution of this level to $I_K-\Gamma K_m(1)$ is
\[
 \mathcal I_\omega=\int_\tau^T
 \{\Lambda(q(t))W(t)+\cprof(q(t))\lambda(t)-\Gamma\}dt.
\]
On the same level, the formula $E_\tau'+sE_\tau=r_{0,\tau}$ from the
proof of Lemma~\ref{lem:curvature} gives
\[
 \int_0^1\sigma(t)[E_\tau'(t)+sE_\tau(t)]dt
 =e^{-s\tau}\int_0^\tau\sigma(t)e^{st}dt.
\]
Since $\sigma=\theta-\rho_1$, Lemma~\ref{lem:scalar} gives
\begin{equation}\label{eq:height-comparison}
 \mathcal I_\omega\le\mathcal B(\tau)=e^{-s\tau}\int_0^\tau\sigma(t)e^{st}\,dt
 =\int_0^1\sigma(t)[E_\tau'(t)+sE_\tau(t)]\,dt.
\end{equation}
It remains to integrate \eqref{eq:height-comparison} over
$0\le\omega\le H_R$. Subtracting \eqref{eq:layer-identity} for $i-1$ from
that for $i$ gives, for almost every $t$,
\[
 D_i'(t)=R_i(t)-R_{i-1}(t)
 =\int_0^{H_R}\one\{T_{i-1}(\omega)\le t<T_i(\omega)\}\,d\omega ,
\]
and applying \eqref{eq:volterra-inverse} to \eqref{eq:layer-identity}
with $i=0$ gives $r_0(t)=\int_0^{H_R}r_{0,T_0(\omega)}(t)\,d\omega$. All
integrands are bounded on $[0,1]\times[0,H_R]$, so Fubini's theorem gives
\[
 I_K-\Gamma K_m(1)=\int_0^{H_R}\mathcal I_\omega\,d\omega,\qquad
 \int_0^1\sigma r_0\,dt
 =\int_0^{H_R}\int_0^1\sigma(t)r_{0,T_0(\omega)}(t)\,dt\,d\omega .
\]
Integrating \eqref{eq:height-comparison} over $\omega$ therefore proves
\eqref{eq:intercept-reduced}.
Substituting \eqref{eq:intercept-reduced} in
\eqref{eq:total-intercept}, with \eqref{eq:initial-payment},
\eqref{eq:retained-source}, and $B_{\mathrm{rem}}\ge0$, gives
\begin{equation}\label{eq:normalized-comparison}
 F(o)\le\Gamma F(h(x))+\ip{g(f,x)-d}{o-x}.
\end{equation}
The affine normalization satisfies
\begin{equation}\label{eq:affine-restoration}
 \ell(o)=\ell(x)+\ip d{o-x}
       \le\Gamma\ell(h(x))+\ip d{o-x}
\end{equation}
by Lemma~\ref{lem:normalization}. Adding
\eqref{eq:normalized-comparison} and \eqref{eq:affine-restoration}
proves \eqref{eq:structural} with the field \eqref{eq:field}.

Finally, each chain difference, secant, and gradient is linear in $f$.
Their orders and weights are objective-independent, so the field is
linear in $f$. Deterministic tie-breaking partitions the state cube
into finitely many Borel regions with fixed coordinate orders. The
formulas on each region are Borel functions of $(x,t)$. Their secant
denominators are bounded away from zero on the relevant large-coordinate
branch, and the objective and its derivatives are bounded on the cube.
Integrating over $t\in[0,1]$ therefore preserves Borel measurability.
The same observation applies to $h$.
\end{proof}

\section{Online Achievability}\label{sec:online}
We apply Online Maximization by Quadratization (OMBQ)
\citep[Theorem~1]{pedramfar2024linear} to Theorem~\ref{thm:structural}.
OMBQ turns an online linear-optimization learner into a learner with
bounded $\alpha$-regret whenever a comparison of the form
\eqref{eq:linearizable} holds.
The field estimator below supplies its linear feedback; projected ascent
and online Frank--Wolfe provide the two constraint-oracle implementations.

\subsection{A single noisy-gradient query}\label{subsec:one-gradient}
Each value difference in \eqref{eq:field} is a coordinate integral,
so its contribution can be estimated by sampling one partial derivative.
For a gradient response $\widehat\nabla f(y)$, let $\mathcal H$
contain the current objective, the query point $y$, and all observations
and random choices preceding the response. The oracle model is
\begin{equation}\label{eq:noisy-gradient-model}
 \E[\widehat\nabla f(y)\mid\mathcal H]=\nabla f(y),\qquad
 \E[\norm{\widehat\nabla f(y)}^2\mid\mathcal H]\le B_1^2,
\end{equation}
with exact gradients as a special case.

Throughout the construction, $x$ is a fixed state and $t$ is a path time.
For two points $z^-,z^+\in Q$ differing only in coordinate $i$, the
normalized segment average is
\[
 \mathcal D_i(z^-,z^+)
 =\int_0^1\partial_i f((1-r)z^-+rz^+)\,dr.
\]
This notation also applies when the two points coincide. If
$i=\pi(\iota)$ in the chain $C_0\subset\cdots\subset C_n$ of
\eqref{eq:normalizer}, its chain average is
$\mathcal D_i^N=\mathcal D_i(\one_{C_{\iota-1}},\one_{C_\iota})$.
Thus $\mathcal D_i^N=d_i$ by the fundamental theorem of calculus.

For $j\in S$, the set $S_j^-$ consists of the small coordinates preceding
$j$ in the nonincreasing order of $u$, with the index rule for ties from
Section~\ref{sec:field}. With $S_j=S_j^-\cup\{j\}$, the relevant segment is
\begin{align*}
 R_j^-(t)&=(a(t)+(1-a(t))\one_{S_j^-},\one_L),\\
 R_j^+(t)&=(a(t)+(1-a(t))\one_{S_j},\one_L),\qquad
 \mathcal D_j^S(t)=\mathcal D_j(R_j^-(t),R_j^+(t)).
\end{align*}
For a large coordinate $i$ in the ordering of
\eqref{eq:large-secants}, the two endpoint averages are
\[
 \mathcal D_i^-(t)=\mathcal D_i((a(t),P_{i-1}^-),(a(t),P_i^-)),\qquad
 \mathcal D_i^+(t)=\mathcal D_i((a(t),P_i^+),(a(t),P_{i-1}^+)).
\]
The coefficients $\chi_i$ of \eqref{eq:V} sum to $c_1$. All these points, orders, and coefficients depend
only on $x$ and $t$. Each relabeled large coordinate is identified with
its original coordinate when forming the vector estimator.
The auxiliary coordinate added in
Section~\ref{sec:field} determines these coefficients and is omitted
from every query point. The sampler below chooses only among the $n$
original coordinates. When $S$ is empty, its small-coordinate terms are
absent.

The large block of \eqref{eq:field} is averaged against the probability
measure $\mu(dt)=(w(t)+\varpi(t))\,dt$. For $j\in S$, put
\begin{equation}\label{eq:normalizer-coefficient}
 \psi_j=\int_0^1e^{-tu_j}\psi(t)\,dt
 \ge\int_0^1e^{-st}\psi(t)\,dt=1.
\end{equation}
Integrating the normalizer coefficient before sampling writes every
term of \eqref{eq:field} as a signed average of single partial
derivatives:
\begin{align}
 g_j(f,x)={}&(1-\psi_j)\mathcal D_j^N
 +\int_0^1e^{-tu_j}\Bigl[
 \lambda(t)\sum_{i=1}^m\chi_i\partial_j f(a(t),P_i^+)
 \nonumber\\
 &\hspace{17mm}+\kappa(t)\partial_j f(a(t),\zero_L)
 -U(t)\mathcal D_j^S(t)\Bigr]\,dt,
 &&j\in S,\label{eq:gradient-only-small}\\
 g_i(f,x)={}&\int_0^1\bigl[(1-p_i)\mathcal D_i^-(t)
                   +p_i\mathcal D_i^+(t)\bigr]\,\mu(dt),
 &&i\in L.\label{eq:gradient-only-large}
\end{align}
All terms, coefficients, and segment-sampling laws depend only on the state
and the fixed schedule.

Let $\mathfrak m_i$ denote the total absolute coefficient mass of the displayed
representation for coordinate $i$, counting each segment average as a
probability average of one raw partial derivative. For $j\in S$ write
$\eta_j(t)=e^{-tu_j}$. Since $\psi_j\ge1$ and $\lambda,\kappa,U,\chi_i\ge0$,
and since $\sum_i\chi_i=c_1$, $c_1\lambda+\kappa=\rho_1$ and
$\psi+U=\rho_1$,
\begin{align}
 \mathfrak m_j&=(\psi_j-1)+\int_0^1\eta_j(t)\bigl(c_1\lambda(t)+\kappa(t)+U(t)\bigr)\,dt
 \nonumber\\
 &=\int_0^1\eta_j(\psi+U)\,dt+\int_0^1\eta_j\rho_1\,dt-1
 =2\int_0^1\eta_j(t)\rho_1(t)\,dt-1.
 \label{eq:coordinate-coefficient-mass}
\end{align}
As in Section~\ref{sec:field}, $\int_0^1\rho_1=c_1-\int_0^1W$, and
$\int_0^1W\ge t_0$ because $W=1$ on $[0,t_0]$. With $\eta_j\le1$ and
$b>4/5$,
\[
 \mathfrak m_j\le2c_1-2t_0-1=\frac12+\frac2b<3 .
\]
For a large coordinate, $\mathfrak m_i=1$ because
$\mu$ has mass one. Thus we use the fixed sampling constant
\begin{equation}\label{eq:sampling-mass}
 C_{\mathrm{int}}=3.
\end{equation}
The sampling argument of
\citet[Lemma~5.1]{geometry2026} (draw a coordinate, then a term, then a uniform point on its segment)
applies to these signed coefficients.
Algorithm~\ref{alg:sampler} lists the resulting sampler step by step.
\begin{lemma}[One-query field implementation]\label{lem:one-query}
Fix an objective $f$ and a state $x$, and let $\mathcal G$ contain
$f$, $x$, and the history before the sampler draws its path time,
coordinate, and term. With sampler randomness independent of $\mathcal G$ and conditionally unbiased
responses as in \eqref{eq:noisy-gradient-model}, at most one
noisy-gradient query gives $\widehat g=\widehat g(f,x)$ satisfying
\begin{equation}\label{eq:one-query-moments}
 \E[\widehat g\mid\mathcal G]=g(f,x),\qquad
 \E[\norm{\widehat g}^2\mid\mathcal G]\le C_{\mathrm{int}}^2n^2B_1^2.
\end{equation}
With exact gradients, the conditional second moment is at most
$C_{\mathrm{int}}^2nH(f)^2$.
In particular, if $\norm{\nabla f(y)}\le B$ on $Q$, it is at most
$2C_{\mathrm{int}}^2nB^2$.
If the noisy response satisfies
$\norm{\widehat\nabla f(y)}\le B_1$ almost surely, then
$\norm{\widehat g}\le nC_{\mathrm{int}}B_1$ almost surely.
\end{lemma}
\begin{proof}
For $j\in S$, we have $1-a_j(t)=\eta_j(t)$ and $\nabla F=\nabla f-d$, so
$[(1-a(t))\nabla_SF(z)]_j=\eta_j(t)(\partial_jf(z)-d_j)$ at every point
$z$. The points $R_j^-(t)$ and $R_j^+(t)$ differ only in coordinate $j$,
where they equal $a_j(t)$ and $1$; the segment has length $\eta_j(t)$, and
\begin{equation}\label{eq:boundary-increment-gradient}
 (\nu_t)_j=F(R_j^+(t))-F(R_j^-(t))=\eta_j(t)(\mathcal D_j^S(t)-d_j).
\end{equation}
Substituting these expressions and $\sum_i\chi_i=c_1$ into
\eqref{eq:field} gives
\[
\begin{aligned}
 g_j(f,x)={}&d_j\Bigl[1-\int_0^1\eta_j\,(c_1\lambda+\kappa-U)\,dt\Bigr]\\
 &+\int_0^1\eta_j(t)\Bigl[\lambda(t)\sum_{i=1}^m\chi_i\partial_jf(a(t),P_i^+)
 +\kappa(t)\partial_jf(a(t),\zero_L)-U(t)\mathcal D_j^S(t)\Bigr]dt .
\end{aligned}
\]
Since $c_1\lambda+\kappa-U=\rho_1-U=\psi$, the bracket equals
$1-\psi_j$, and $d_j=\mathcal D_j^N$; this is
\eqref{eq:gradient-only-small}. For $i\in L$, the two points in $A_i^L$
differ only in coordinate $i$, from $0$ to $p_i$, and the two points in
$B_i^L$ from $p_i$ to $1$. Dividing by these segment lengths gives
\[
 A_i^L=\mathcal D_i^--d_i,\qquad B_i^L=\mathcal D_i^+-d_i,\qquad
 g_i^L=(1-p_i)\mathcal D_i^-+p_i\mathcal D_i^+-d_i,
\]
and $\mu([0,1])=1$ cancels $d_i$ in $g_i=d_i+\int_0^1g_i^L\,d\mu$. This
gives \eqref{eq:gradient-only-large}.
Equivalently, each coordinate has a signed measure $\zeta_i$ on $Q$,
obtained by spreading each segment average uniformly over its segment,
with
\begin{equation}\label{eq:field-measure-representation}
 g_i(f,x)=\int_Q\partial_i f(y)\,\zeta_i(dy),\qquad
 \|\zeta_i\|_{\mathrm{TV}}\le \mathfrak m_i\le C_{\mathrm{int}}.
\end{equation}
The sampler first chooses $I$ uniformly from $[n]$. Within the displayed
representation for $g_I$, it selects a term with probability equal to
its absolute coefficient measure divided by $C_{\mathrm{int}}$; with
remaining probability $1-\mathfrak m_I/C_{\mathrm{int}}$ it returns zero and makes
no query. A segment term also draws a uniform point $Y$ on its segment.
For a selected coefficient $\omega\ne0$, one noisy-gradient response
at $Y$ gives
\begin{equation}\label{eq:one-gradient-estimator}
 \widehat g=nC_{\mathrm{int}}\operatorname{sign}(\omega)
                 [\widehat\nabla f(Y)]_I e_I.
\end{equation}
Conditional on $I$, its mean is $ng_I(f,x)e_I$ by
\eqref{eq:field-measure-representation}. Averaging over $I$ proves
conditional unbiasedness.

The multiplier in \eqref{eq:one-gradient-estimator} also controls the
second moment. Since
$[\widehat\nabla f(Y)]_I^2\le\norm{\widehat\nabla f(Y)}^2$, the conditional
second moment is at most $C_{\mathrm{int}}^2n^2B_1^2$.
The same multiplier bound gives
$\norm{\widehat g}\le nC_{\mathrm{int}}B_1$ when every response
has norm at most $B_1$.

For exact gradients, $|\partial_i f(Y)|\le H_i(f)$ at every queried
point. Since $I$ is uniform and a term is selected with probability at
most one,
\[
 \E[\norm{\widehat g}^2\mid\mathcal G]
 =\frac1n\sum_{i=1}^nn^2C_{\mathrm{int}}^2\,
 \E\bigl[\one\{\text{a term is selected}\}\,\partial_if(Y)^2\mid\mathcal G,I=i\bigr]
 \le nC_{\mathrm{int}}^2\sum_{i=1}^nH_i(f)^2 .
\]
Finally,
\[
 H(f)^2\le\norm{\nabla f(\zero)}^2+
                 \norm{\nabla f(\one)}^2\le2B^2,
\]
which gives the bound in terms of $B$. All expectations in the proof
may be conditioned on $\mathcal G$, since the state and objective
are fixed under that conditioning.
\end{proof}

\begin{algorithm}[t]
\caption{One-query field sampler at a state $x$ (Lemma~\ref{lem:one-query})}
\label{alg:sampler}
\begin{algorithmic}[1]
\State Form the coordinate orders, the coefficients, and the segments in
\eqref{eq:gradient-only-small}--\eqref{eq:gradient-only-large}; they
depend only on $x$.
\State Draw an original coordinate $I$ uniformly from $[n]$.
\State With probability $1-\mathfrak m_I/C_{\mathrm{int}}$,
\Return $\widehat g=\zero$ without a query.
\State Select one signed term of the representation of $g_I$, including
its path time $t$, with probability proportional to its absolute
coefficient; let $\omega$ be its coefficient.
\If{the term is a segment average}
\State draw $Y$ uniformly on the segment;
\Else
\State let $Y$ be the term's point.
\EndIf
\State Make one gradient query at $Y$.
\State \Return
$\widehat g=nC_{\mathrm{int}}\operatorname{sign}(\omega)[\widehat\nabla f(Y)]_Ie_I$.
\end{algorithmic}
\end{algorithm}
Algorithm~\ref{alg:sampler} makes at most one gradient query.
Section~\ref{subsec:uniform} replaces the uniform draw in Line~2 by a
nonuniform one and adjusts the factor in the last line accordingly.

For the direct exact-oracle implementation, independently apply the
same signed-term sampler to every original coordinate $i$, with unused
probability assigned to zero, and return
\begin{equation}\label{eq:direct-scaled-field}
 \widetilde g_i=C_{\mathrm{int}}\operatorname{sign}(\omega_i)\partial_i f(Y_i)
\end{equation}
when a term is selected, and $\widetilde g_i=0$ otherwise. Each coordinate is
unbiased for $g_i(f,x)$ and satisfies $|\widetilde g_i|\le C_{\mathrm{int}}H_i(f)$.
Thus $\E[\widetilde g\mid\mathcal G]=g(f,x)$, $\norm{\widetilde g}\le C_{\mathrm{int}}H(f)$,
and the conditional second moment is at most
$C_{\mathrm{int}}^2H(f)^2$. Each sample uses at most $n$ exact
full-gradient queries.

\subsection{OMBQ with standard base learners}\label{subsec:ombq}
Theorem~\ref{thm:structural} supplies the upper-linearizable comparison
required by OMBQ \citep[Theorem~1]{pedramfar2024linear}: multiplying
\eqref{eq:structural} by $\alpha_0$ and restricting to $x,o\in K$ gives
\eqref{eq:linearizable} with $\alpha=\alpha_0$.
Let $\mathcal H_t$ collect $f_t$, $x_t$, and all observations and random
choices before the field sample of round $t$. After the base learner
chooses $x_t\in K$ and the action $p_t=h(x_t)$ is committed, the sample
$\widehat g_t$ of Lemma~\ref{lem:one-query}, with
$f=f_t$, $x=x_t$, and $\mathcal G=\mathcal H_t$, is supplied to that
learner as its linear reward vector.

\begin{proposition}[OMBQ transfer]\label{thm:online-transfer}
Suppose \eqref{eq:linearizable} holds with coefficient $\alpha$, and
$\mathcal H_t$ contains the history and state before the field sample
$\widehat g_t$ is drawn. If
\begin{equation}\label{eq:ombq-unbiased-field}
 \E[\widehat g_t\mid\mathcal H_t]=g(f_t,x_t)
\end{equation}
and the base learner has expected linear regret at most
$\mathcal R_{\mathcal A}(T)$ on these vectors, which may depend on the past, then OMBQ has
\begin{equation}\label{eq:online-transfer}
 \Reg_\alpha(T)\le\alpha\mathcal R_{\mathcal A}(T).
\end{equation}
Integrable samples suffice; a linear-regret bound that holds for every
realization may be
averaged over the field sampling.
\end{proposition}
\begin{proof}

Fix $o\in K$. Since $x_t$ is determined by $\mathcal H_t$, the tower
property and \eqref{eq:ombq-unbiased-field} give
$\E\ip{\widehat g_t}{o-x_t}=\E\ip{\E[\widehat g_t\mid\mathcal H_t]}{o-x_t}$,
that is,
\begin{equation}\label{eq:conditional-field-transfer}
 \E\ip{g(f_t,x_t)}{o-x_t}
       =\E\ip{\widehat g_t}{o-x_t}.
\end{equation}

Summing \eqref{eq:linearizable} and using
\eqref{eq:conditional-field-transfer} gives, for each fixed $o\in K$,
\begin{equation}\label{eq:ombq-summed-comparison}
 \alpha\sum_{t=1}^T f_t(o)-\E\sum_{t=1}^T f_t(h(x_t))
 \le\alpha\E\sum_{t=1}^T\ip{\widehat g_t}{o-x_t}
 \le\alpha\mathcal R_{\mathcal A}(T).
\end{equation}
The maximizing static comparator is fixed because the objective
sequence is oblivious. Maximizing the left side of
\eqref{eq:ombq-summed-comparison} over $o$ proves
\eqref{eq:online-transfer}. For bounded samples, this
is \citet[Theorem~1]{pedramfar2024linear} with its quadratic correction
term set to zero; the same conditional-expectation
calculation establishes the stated integrable-sample formulation.
\end{proof}

Two standard base learners give the constraint-oracle alternatives.
For $D,\mathcal V>0$, projected online gradient ascent
\citep{zinkevich2003} starts at $x_1=\zero$ and uses
\begin{equation}\label{eq:ombq-update}
 x_{t+1}=\Pi_K(x_t+\varkappa \widehat g_t),\qquad
 \varkappa=\frac{D}{\sqrt{T\mathcal V}},
\end{equation}
where $\E[\norm{\widehat g_t}^2\mid\mathcal H_t]\le\mathcal V$.
Its standard bound, valid for every realization, is
\begin{equation}\label{eq:projection-linear-bound}
 \sum_{t=1}^T\ip{\widehat g_t}{o-x_t}
 \le\frac{D^2}{2\varkappa}+\frac\varkappa2\sum_{t=1}^T\norm{\widehat g_t}^2.
\end{equation}
It applies because $\norm{o-x_1}\le D$ for $o\in K$. Taking
expectations and using
$\E\sum_t\norm{\widehat g_t}^2=\sum_t\E\,\E[\norm{\widehat g_t}^2\mid\mathcal H_t]\le T\mathcal V$
gives
$\mathcal R_{\mathcal A}(T)\le\frac{D^2}{2\varkappa}+\frac{\varkappa T\mathcal V}2
=\frac{D\sqrt{T\mathcal V}}2+\frac{D\sqrt{T\mathcal V}}2=D\sqrt{T\mathcal V}$,
so this learner
only needs a second-moment bound.

For LMO access, the base learner is online Frank--Wolfe (OFW)
\citep[Algorithm~1 and Theorem~4.4]{hazan2012projection}, applied to
the linear losses $x\mapsto-\ip{\widehat g_t}{x}$. Let
$\norm{\widehat g_t}\le B_{\mathrm{lin}}$ for every realized linear reward. OFW
starts at $x_1=\zero$ and forms the regularized surrogates and their
running average
\[
 \widehat\ell_t(x)=-\ip{\widehat g_t}{x}+\frac{B_{\mathrm{lin}}}{D}t^{-1/4}\norm{x-x_1}^2,\qquad
 \mathcal F_t=\frac1t\sum_{\tau=1}^t\widehat\ell_\tau.
\]
It then calls the LMO once to obtain
$v_t\in\arg\min_{v\in K}\ip{\nabla\mathcal F_t(x_t)}{v}$ and sets
$x_{t+1}=(1-t^{-1/2})x_t+t^{-1/2}v_t$. Its guarantee, valid for every realization, is
\begin{equation}\label{eq:online-lmo-linear-bound}
 \sum_{t=1}^T\ip{\widehat g_t}{o-x_t}
       \le c_{\mathrm{FW}}DB_{\mathrm{lin}}\,T^{3/4},
       \qquad c_{\mathrm{FW}}=57,
\end{equation}
by \citet[Theorem~4.4]{hazan2012projection}.
Each round uses one ordinary LMO and updates the state by a convex
combination with its output. Starting at $\zero$ therefore retains
an explicit decomposition into earlier oracle outputs and $\zero$,
which will be used for matroid rounding.

Applying Proposition~\ref{thm:online-transfer} with the field estimator
and either of these base learners gives the following bounds.
\begin{corollary}[Online first-order guarantees from OMBQ]\label{thm:online}
Supplying the estimator of Lemma~\ref{lem:one-query} to the preceding
standard base learners gives
\begin{equation}\label{eq:online-regret}
 \Reg_{\alpha_0}(T)\le\alpha_0D
 \begin{cases}
  \sqrt{T\mathcal V},&\text{projected ascent},\\
  c_{\mathrm{FW}}B_{\mathrm{lin}} T^{3/4},&\text{online Frank--Wolfe}.
 \end{cases}
\end{equation}
Each round uses at most one gradient query in $Q$ and one projection or one
ordinary LMO, respectively. For projected ascent, valid moment bounds
are
\begin{equation}\label{eq:online-field-moments}
 \mathcal V=
 \begin{cases}
  C_{\mathrm{int}}^2n\overline H^2,
      &\text{exact gradients with }H(f_t)\le\overline H,\\
  C_{\mathrm{int}}^2n^2B_1^2,
      &\text{noisy gradients satisfying \eqref{eq:noisy-gradient-model}}.
 \end{cases}
\end{equation}
For OFW, valid almost-sure bounds are
\begin{equation}\label{eq:online-field-bounds}
 B_{\mathrm{lin}}=
 \begin{cases}
  C_{\mathrm{int}}n\overline H,
      &\text{exact gradients with }H(f_t)\le\overline H,\\
  C_{\mathrm{int}}nB_1,
      &\text{bounded noisy gradients}.
 \end{cases}
\end{equation}
Here bounded noisy gradients satisfy
$\norm{\widehat\nabla f_t(y)}\le B_1$ almost surely at every query,
in addition to conditional unbiasedness.
When exact gradients satisfy $\norm{\nabla f_t(y)}\le B$ on $Q$,
one may instead use $\mathcal V=2C_{\mathrm{int}}^2nB^2$ and
$B_{\mathrm{lin}}=C_{\mathrm{int}}nB$.
The direct estimator $\widetilde g$ in \eqref{eq:direct-scaled-field} uses at most $n$ exact
full-gradient queries per round and gives
$\mathcal V=C_{\mathrm{int}}^2\overline H^2$ and
$B_{\mathrm{lin}}=C_{\mathrm{int}}\overline H$.
\end{corollary}
\begin{proof}

Lemma~\ref{lem:one-query} gives conditional unbiasedness and
\eqref{eq:online-field-moments}. The multiplier in
\eqref{eq:one-gradient-estimator} has absolute value at most
$nC_{\mathrm{int}}$. Thus a bounded oracle response gives the
corresponding bound in \eqref{eq:online-field-bounds}; exact responses
also satisfy $|\partial_I f_t(Y)|\le H_I(f_t)\le\overline H$.
The direct-estimator bounds were established after
\eqref{eq:direct-scaled-field}.

Multiplying the linear-regret bounds
\eqref{eq:projection-linear-bound} and
\eqref{eq:online-lmo-linear-bound} by $\alpha_0$, as in
Proposition~\ref{thm:online-transfer}, gives \eqref{eq:online-regret}.
If $D=0$, or if the applicable bound $\mathcal V$ or $B_{\mathrm{lin}}$ is zero, a
fixed state gives nonpositive approximate regret directly from
\eqref{eq:linearizable} and nonnegativity.
\end{proof}

\subsection{Function-value feedback}\label{subsec:value-feedback}
Following \citet[Lemma~5.3]{geometry2026}, we smooth $f_t$ by
averaging it over a small box around a point moved slightly toward the
center of the cube. The
one gradient coordinate needed by the sampler can then be estimated
from one noisy value or two exact values. Assume
$\norm{\nabla f_t(y)}\le B$ for all $y\in Q$. For $0<r<1/2$, points are rescaled toward the center by the factor
$1-2r$, and $\mathsf U$ has independent uniform coordinates
on $[-1,1]$. The smoothed objective is
\begin{equation}\label{eq:cube-smoothing}
 f_{t,r}(y)=\E_{\mathsf U}
 f_t((1-2r) y+r\one+r\mathsf U),\qquad y\in Q.
\end{equation}
The cited lemma gives nonnegativity, DR-submodularity, and the gradient
bound below; averaging the squared displacement gives the stated
refinement of its uniform value-error bound:
\begin{equation}\label{eq:smoothing-error}
 \norm{\nabla f_{t,r}(y)}\le B,\qquad
 \sup_{y\in Q}|f_{t,r}(y)-f_t(y)|\le2Br\sqrt{n/3}.
\end{equation}
For the displacement $\Delta=(1-2r) y+r\one+r\mathsf U-y$,
$\E\norm\Delta^2=r^2\sum_i[(1-2y_i)^2+1/3]\le4r^2n/3$,
so Jensen's inequality and the $B$-Lipschitz property give the value
bound in \eqref{eq:smoothing-error}.

After $I,\omega,Y$ have been sampled as in Lemma~\ref{lem:one-query},
an independent vector $\mathsf U^{(-I)}$ has coordinate
$\mathsf U_I^{(-I)}=0$ and independent uniform coordinates on $[-1,1]$
at all other indices.
The two evaluation points are
\[
 Z^{\pm}=(1-2r)Y+r\one+r\mathsf U^{(-I)}\pm r e_I.
\]
With exact values, the requested coordinate
$[\widehat\nabla f(Y)]_I$ in
\eqref{eq:one-gradient-estimator} is estimated by
\begin{equation}\label{eq:value-two-point}
 \widehat v_I^{(2)}=
 \frac{1-2r}{2r}[f_t(Z^+)-f_t(Z^-)].
\end{equation}
For noisy values, $\widehat f_t(y)$ denotes a response with conditional
mean $f_t(y)$ and conditional second moment at most $M_0^2$, under the
same conditioning as in \eqref{eq:noisy-gradient-model}.
An independent uniform sign $\varsigma\in\{-1,1\}$ gives the alternative
\begin{equation}\label{eq:value-one-point}
 \widehat v_I^{(1)}=\frac{1-2r}{r}\varsigma\widehat f_t(Z^\varsigma).
\end{equation}
The exact difference quotient has a second moment independent of $r$,
whereas the noisy estimator has moment of order $r^{-2}$. The radius
choices below respectively make the accumulated smoothing error constant
or balance it with the estimation error.
\begin{corollary}[Online value-feedback guarantees]\label{cor:online-values}
Assume $\norm{\nabla f_t(y)}\le B$ for every $t$ and every $y\in Q$.
The projected learner in Corollary~\ref{thm:online} accepts the preceding
value substitutions under their conditional second-moment assumptions.
The standard OFW learner accepts exact values; for noisy values, it
additionally assumes $|\widehat f_t(y)|\le M_0$ almost surely at every
query.

Two exact value queries per round use $r=(4T)^{-1}$, with
$\mathcal V=2C_{\mathrm{int}}^2nB^2$ for projected ascent and
$B_{\mathrm{lin}}=C_{\mathrm{int}}nB$ for OFW. One noisy value query uses
$\mathcal V=C_{\mathrm{int}}^2n^2M_0^2/r^2$ for projected ascent or
$B_{\mathrm{lin}}=C_{\mathrm{int}}nM_0/r$ for OFW, with the respective radii
\begin{equation}\label{eq:value-noisy-radii}
 r=(4T^{1/4})^{-1}\quad\text{for projected ascent},\qquad
 r=(4T^{1/8})^{-1}\quad\text{for OFW}.
\end{equation}
The constants are
\begin{align}
 A_2&=\sqrt2\,C_{\mathrm{int}}\alpha_0DB\sqrt n,&
 A_{2,\mathrm{LMO}}&=c_{\mathrm{FW}}C_{\mathrm{int}}\alpha_0DnB,
       \label{eq:value-two-constants}\\
 C_2&=\frac{1+\alpha_0}{2}B\sqrt{n/3},&
 A_1&=4C_{\mathrm{int}}\alpha_0DnM_0+C_2,
       \label{eq:value-common-constants}\\
 A_{1,\mathrm{LMO}}&=4c_{\mathrm{FW}}C_{\mathrm{int}}\alpha_0DnM_0+C_2.
       \label{eq:value-lmo-constant}
\end{align}
The expected regrets, with superscripts indicating the number of value
queries, satisfy
\begin{align}
 \Reg_{\alpha_0}^{(2)}(T)&\le
 \begin{cases}
  A_2\sqrt T+C_2,&\text{projected ascent},\\
  A_{2,\mathrm{LMO}}T^{3/4}+C_2,&\text{online Frank--Wolfe},
 \end{cases}
 \label{eq:value-exact-rate}\\
 \Reg_{\alpha_0}^{(1)}(T)&\le
 \begin{cases}
  A_1T^{3/4},&\text{projected ascent},\\
  A_{1,\mathrm{LMO}}T^{7/8},&\text{online Frank--Wolfe}.
 \end{cases}
 \label{eq:value-noisy-rate}
\end{align}
Each round uses one projection or one ordinary LMO, respectively.
\end{corollary}
\begin{proof}
By \citet[Lemma~5.3]{geometry2026}, the identity
\begin{equation}\label{eq:value-coordinate-mean}
 \partial_I f_{t,r}(Y)=
 \frac{1-2r}{2r}\E_{\mathsf U^{(-I)}}[f_t(Z^+)-f_t(Z^-)]
\end{equation}
makes both value substitutions conditionally unbiased, and the
objective-independent sampler of Lemma~\ref{lem:one-query} therefore
estimates $g(f_{t,r},x_t)$.
To see the identity, put $z=(1-2r)Y+r\one+r\mathsf U^{(-I)}$. Since
$\mathsf U_I$ has density $1/2$ on $[-1,1]$ and $\partial_If_{t,r}(Y)=
(1-2r)\E[\partial_If_t((1-2r)Y+r\one+r\mathsf U)]$,
\[
 \partial_If_{t,r}(Y)=(1-2r)\E\Bigl[\frac12\int_{-1}^1\partial_If_t(z+rue_I)\,du\Bigr]
 =\frac{1-2r}{2r}\E\bigl[f_t(Z^+)-f_t(Z^-)\bigr].
\]
For the same reason
$\widehat v_I^{(2)}=\frac{1-2r}2\int_{-1}^1\partial_If_t(z+rue_I)\,du$,
so $|\widehat v_I^{(2)}|\le(1-2r)H_I(f_t)\le H_I(f_t)$, and the
exact-gradient bound of Lemma~\ref{lem:one-query} applies unchanged.
For one value, averaging over the sign gives
$\E_\varsigma\bigl[\frac{1-2r}r\varsigma f_t(Z^\varsigma)\bigr]
=\frac{1-2r}{2r}[f_t(Z^+)-f_t(Z^-)]$, and
$\E[(\widehat v_I^{(1)})^2]\le(1-2r)^2M_0^2/r^2\le M_0^2/r^2$. Combined with the multiplier
$nC_{\mathrm{int}}$ in \eqref{eq:one-gradient-estimator}, these
give the stated $\mathcal V$ bounds,
and the corresponding pointwise estimates give the stated
$B_{\mathrm{lin}}$ bounds under almost-sure bounded feedback.
Finally, \eqref{eq:smoothing-error} gives
\begin{equation}\label{eq:value-smoothing-transfer}
 \alpha_0 f_t(o)-f_t(h(x_t))
 \le \alpha_0 f_{t,r}(o)-f_{t,r}(h(x_t))
       +2(1+\alpha_0)Br\sqrt{n/3}.
\end{equation}
If $\mathcal R_{\mathcal A}(T)$ is the linear regret
of the corresponding base learner evaluated with the smoothed
field estimates, the approximate regret on the original objectives
satisfies
\begin{equation}\label{eq:value-regret-substitution}
 \Reg_{\alpha_0}(T)
 \le\alpha_0\mathcal R_{\mathcal A}(T)+4C_2rT.
\end{equation}
For two exact values, $r=(4T)^{-1}$ makes the last term $C_2$,
and substituting the displayed moment bounds gives the two
coefficients $A_2$ and $A_{2,\mathrm{LMO}}$.
For one noisy value, the projected linear-regret bound gives
$\alpha_0D\sqrt{T\mathcal V}
 =4C_{\mathrm{int}}\alpha_0DnM_0T^{3/4}$
when $r=(4T^{1/4})^{-1}$; the last term in
\eqref{eq:value-regret-substitution} is $C_2T^{3/4}$.
For online Frank--Wolfe, the corresponding bound
$\alpha_0c_{\mathrm{FW}}DB_{\mathrm{lin}}T^{3/4}$
becomes $4c_{\mathrm{FW}}C_{\mathrm{int}}\alpha_0DnM_0T^{7/8}$
when $r=(4T^{1/8})^{-1}$, and the smoothing term is $C_2T^{7/8}$.
Adding the estimation and smoothing terms gives $A_1$ and
$A_{1,\mathrm{LMO}}$
in~\eqref{eq:value-common-constants}--\eqref{eq:value-lmo-constant}
and proves both stated rates.
\end{proof}

\section{Offline and Discrete Consequences}\label{sec:consequences}
For a fixed objective, online-to-batch (OTB) conversion
\citep[Theorem~8]{pedramfar2024linear} turns the accumulated regret
into an additive approximation error. Applied to a multilinear extension,
this guarantee combines with matroid or knapsack rounding.

\subsection{Offline approximation by OTB}\label{subsec:offline}
Let $f:Q\to\R$ be a fixed nonnegative $C^1$ DR-submodular objective
and $K$ satisfy Section~\ref{sec:background}, with
$\operatorname{OPT}=\max_{o\in K}f(o)$. The OTB output is a
uniformly selected action from a run on $f_t=f$.
\begin{corollary}[Offline approximation by OTB]
\label{cor:offline-conversion}
Either online implementation is run on the constant sequence $f_t=f$,
under its stated feedback assumptions. If its expected approximate
regret is at most $\mathcal R(T)$, an independent uniformly selected
played action $Y$ satisfies
\begin{equation}\label{eq:offline-conversion}
 \E f(Y)\ge\alpha_0\operatorname{OPT}-\frac{\mathcal R(T)}T.
\end{equation}
In particular, the regret bounds of Corollaries~\ref{thm:online}
and~\ref{cor:online-values} give the following sufficient objective-query
counts for additive error $\epsilon$, with dimension and response bounds
fixed:
\[
\begin{array}{c|cc}
 \text{feedback}&\text{projection}&\text{LMO}\\
 \hline
 \text{one gradient or two exact values per round}
   &O(\epsilon^{-2})&O(\epsilon^{-4})\\
 \text{one noisy value per round}
   &O(\epsilon^{-4})&O(\epsilon^{-8}).
\end{array}
\]
The projection guarantees use conditional second-moment bounds;
the LMO guarantees use the almost-sure bounds stated in those
corollaries. Each round uses one call to the corresponding constraint
oracle. A supplied lower bound $0<\mathsf v\le\operatorname{OPT}$ converts
additive error $\epsilon=\varepsilon \mathsf v$ into coefficient
$\alpha_0-\varepsilon$.
\end{corollary}
\begin{proof}
Let $\jmath$ be independent of the online run and uniform on
$\{1,\ldots,T\}$, and return $Y=p_\jmath$, where $p_t=h(x_t)$ is the
played action at round $t$.
Conditional on $p_1,\ldots,p_T$, the expected output value is
$T^{-1}\sum_{t=1}^T f(p_t)$; taking expectation and using the definition
of approximate regret gives
\begin{equation}\label{eq:offline-uniform-value}
 \E f(Y)=\frac1T\sum_{t=1}^T\E f(p_t)
 =\alpha_0\operatorname{OPT}-\frac{\Reg_{\alpha_0}(T)}T
 \ge\alpha_0\operatorname{OPT}-\frac{\mathcal R(T)}T.
\end{equation}
More generally, an online bound $\mathcal R(T)\le \bar aT^{\bar\nu}+\bar c$, with
$\bar a,\bar c\ge0$ and $0<\bar\nu<1$, gives error $\bar aT^{\bar\nu-1}+\bar c/T$. A sufficient
horizon for error at most $\epsilon$ is
\begin{equation}\label{eq:offline-horizon}
 T\ge\left\lceil\max\left\{1,
       \left(\frac{2\bar a}{\epsilon}\right)^{1/(1-\bar\nu)},
       \frac{2\bar c}{\epsilon}\right\}\right\rceil.
\end{equation}
Each term of the error is then at most $\epsilon/2$. Substituting
$\bar\nu=1/2,3/4,$ or $7/8$ from
Corollaries~\ref{thm:online} and~\ref{cor:online-values}, with
$\bar c=C_2$ for two exact values and $\bar c=0$ otherwise, gives the table.
For the feedback implementations in the table, the number of objective
queries is at most $2T$, and $Y\in\{p_1,\ldots,p_T\}\subseteq K$.
Finally, $\epsilon=\varepsilon \mathsf v\le\varepsilon\operatorname{OPT}$ in
\eqref{eq:offline-uniform-value} gives
$\E f(Y)\ge(\alpha_0-\varepsilon)\operatorname{OPT}$.
\end{proof}

\subsection{A uniform algorithm for the full \texorpdfstring{$C^1$}{C1} class}\label{subsec:uniform}
The response bounds in Corollary~\ref{cor:offline-conversion} are fixed
instance parameters. For a fixed objective with exact value and
full-gradient access, we now remove every such parameter in the exact-oracle
model \eqref{eq:uniform-model}. The construction contracts the objective
toward the feasible interior point $\bar z$ defined below;
nonnegativity and down-closedness then bound the gradients of the
contracted objective.

\emph{Width normalization} computes the width
$\mathsf w_i=\max_{x\in K}x_i$ of each coordinate with one LMO call in direction $e_i$,
fixes the coordinates of zero width at zero, relabels the remaining
$n'$ coordinates, and sets
\begin{equation}\label{eq:width-rescaling}
 \mathsf W=\diag(\mathsf w_1,\ldots,\mathsf w_{n'}),\qquad
 K_{\mathsf W}=\{z\in[0,1]^{n'}:\mathsf Wz\in K\},\qquad
 \widetilde f(z)=f(\mathsf Wz),
\end{equation}
where deleted coordinates are reinserted as zeros in $\mathsf Wz$.
If no coordinate remains, then $K=\{\zero\}$ and returning $\zero$ is
optimal. Otherwise the rescaled problem has the same optimum
$\mathrm{OPT}$, $\widetilde f$ is nonnegative, $C^1$, and
DR-submodular on $[0,1]^{n'}$, and $K_{\mathsf W}$ is compact, convex,
and down-closed, with diameter at most $\sqrt{n'}$. Down-closedness gives
$e_i\in K_{\mathsf W}$ for every remaining coordinate $i$. Put
\[
 \varrho=\frac1{n'+1},\qquad \bar z=\varrho\one\in K_{\mathsf W},
 \qquad A_0=\widetilde f(\bar z).
\]
We call $\bar z$ the anchor. It is a convex combination of $\zero$ and the vectors
$e_i$, so it is feasible.

\begin{lemma}[Anchor, contraction, and endpoint gradient bounds]\label{lem:contraction}
For every $o\in[0,1]^{n'}$,
\begin{equation}\label{eq:anchor}
 A_0\ge\varrho(1-\varrho)\widetilde f(o).
\end{equation}
For $0<\delta<1$, define
\[
 T_\delta(z)=(1-\delta)z+\delta\bar z,\qquad
 f_\delta(z)=\widetilde f(T_\delta(z)).
\]
Then $f_\delta$ is nonnegative, $C^1$, and DR-submodular on
$[0,1]^{n'}$, and
\begin{equation}\label{eq:contraction-value}
 f_\delta(o)\ge(1-\delta)\widetilde f(o)\qquad(o\in[0,1]^{n'}).
\end{equation}
Its endpoint gradient bounds
$H_i(f_\delta)=\max\{|\partial_if_\delta(\zero)|,|\partial_if_\delta(\one)|\}$
satisfy $|\partial_if_\delta(z)|\le H_i(f_\delta)$ for every $z$, and
\begin{equation}\label{eq:uniform-gradient}
 H_{\Sigma}(f_\delta):=\sum_{i=1}^{n'}H_i(f_\delta)
 \le\frac{(1-\delta)(2-\delta)}{\delta\varrho(1-\varrho)}A_0.
\end{equation}
For every $x\in K_{\mathsf W}$, the original action
$\mathsf WT_\delta(h(x))$ belongs to $K$.
\end{lemma}
\begin{proof}
Directional concavity from $\zero$ and nonnegativity give
$\widetilde f(\varrho o)\ge\varrho\widetilde f(o)$. The point
$y=\one-(1-\varrho)o$ lies above $\varrho o$, and
$\bar z=(1-\varrho)\varrho o+\varrho y$. Concavity along the nonnegative
direction $y-\varrho o$ and $\widetilde f(y)\ge0$ give
$A_0\ge(1-\varrho)\widetilde f(\varrho o)$, which proves
\eqref{eq:anchor}.

The map $T_\delta$ is affine with positive diagonal coefficients, so
$f_\delta$ inherits nonnegativity, continuous differentiability, and
DR-submodularity. For $a',r'\ge0$ with $s'=a'+r'\in(0,1]$, put $y'=(a'o+r'\one)/s'$, a
point of the segment from $o$ to $\one$. Concavity along the nonnegative
direction $\one-o$, and then along the segment from $\zero$ to $y'$, give
\[
 \widetilde f(y')\ge\frac{a'}{s'}\widetilde f(o)+\frac{r'}{s'}\widetilde f(\one),\qquad
 \widetilde f(a'o+r'\one)=\widetilde f(s'y')\ge s'\widetilde f(y')+(1-s')\widetilde f(\zero),
\]
and therefore
\[
 \widetilde f(a'o+r'\one)\ge a'\widetilde f(o)+r'\widetilde f(\one)
 +(1-a'-r')\widetilde f(\zero).
\]
Taking $a'=1-\delta$ and $r'=\delta\varrho$, so that
$a'o+r'\one=T_\delta(o)$, and using nonnegativity proves
\eqref{eq:contraction-value}.

DR-submodularity gives
$\partial_if_\delta(\zero)\ge\partial_if_\delta(z)\ge\partial_if_\delta(\one)$,
hence $|\partial_if_\delta(z)|\le H_i(f_\delta)$. The same inequality
gives $(\partial_if_\delta(\one))_+\le(\partial_if_\delta(\zero))_+$ and
$(-\partial_if_\delta(\zero))_+\le(-\partial_if_\delta(\one))_+$, so both
$|\partial_if_\delta(\zero)|$ and $|\partial_if_\delta(\one)|$ are at most
$(\partial_if_\delta(\zero))_++(-\partial_if_\delta(\one))_+$; that is,
$H_i(f_\delta)\le(\partial_if_\delta(\zero))_++(-\partial_if_\delta(\one))_+$.
For a nonnegative DR function $\phi$ and a point $y$, tangency after
lowering to zero every coordinate $i$ with $\partial_i\phi(y)>0$, or after
raising to one every coordinate with $\partial_i\phi(y)<0$, gives
\[
 \sum_iy_i(\partial_i\phi(y))_+\le\phi(y),\qquad
 \sum_i(1-y_i)(-\partial_i\phi(y))_+\le\phi(y).
\]
Indeed, with $P=\{i:\partial_i\phi(y)>0\}$ and
$y'=y-\sum_{i\in P}y_ie_i\le y$, tangency for comparable points and
$\phi(y')\ge0$ give $0\le\phi(y')\le\phi(y)-\sum_{i\in P}y_i\partial_i\phi(y)$.
The second bound follows in the same way from
$y''=y+\sum_{i:\partial_i\phi(y)<0}(1-y_i)e_i\ge y$.
Apply the first bound to $\widetilde f$ at $y^-=T_\delta(\zero)=\delta\varrho\one$
and the second at $y^+=T_\delta(\one)=(1-\delta+\delta\varrho)\one$, whose
distance to $\one$ is $\delta(1-\varrho)$. With the chain-rule factor
$1-\delta$,
\[
 H_{\Sigma}(f_\delta)\le(1-\delta)\left[\frac{\widetilde f(y^-)}{\delta\varrho}
 +\frac{\widetilde f(y^+)}{\delta(1-\varrho)}\right].
\]
The function $\varphi(t)=\widetilde f(t\one)$ is concave and nonnegative on
$[0,1]$ with $\varphi(\varrho)=A_0$. For $t\le\varrho$, write
$\varrho=\frac{1-\varrho}{1-t}\,t+\frac{\varrho-t}{1-t}\cdot1$; concavity and
$\varphi(1)\ge0$ give $A_0\ge\frac{1-\varrho}{1-t}\varphi(t)$, that is,
$\varphi(t)\le(1-t)A_0/(1-\varrho)$. For $t\ge\varrho$, write
$\varrho=\frac\varrho t\,t+\bigl(1-\frac\varrho t\bigr)\cdot0$; concavity and
$\varphi(0)\ge0$ give $\varphi(t)\le tA_0/\varrho$. Hence
$\widetilde f(y^-)\le(1-\delta\varrho)A_0/(1-\varrho)$ and
$\widetilde f(y^+)\le(1-\delta+\delta\varrho)A_0/\varrho$. Substituting,
\[
 H_{\Sigma}(f_\delta)\le(1-\delta)\Bigl[\frac{(1-\delta\varrho)A_0}{\delta\varrho(1-\varrho)}
 +\frac{(1-\delta+\delta\varrho)A_0}{\delta\varrho(1-\varrho)}\Bigr]
 =\frac{(1-\delta)(2-\delta)}{\delta\varrho(1-\varrho)}A_0,
\]
which is \eqref{eq:uniform-gradient}. Finally,
$h(x)\in K_{\mathsf W}$ by down-closedness, $\bar z\in K_{\mathsf W}$, and
convexity give $T_\delta(h(x))\in K_{\mathsf W}$.
\end{proof}

\begin{theorem}[Uniform $C^1$ offline approximation]\label{thm:uniform}
In the exact-oracle model \eqref{eq:uniform-model}, let $f$ be
nonnegative, $C^1$, and DR-submodular on the cube, and let $K$ be
nonempty, compact, convex, and down-closed. For every
$0<\varepsilon<1$ there is a randomized algorithm returning $y\in K$
with
\begin{equation}\label{eq:uniform-guarantee}
 \E f(y)\ge(\alpha_0-\varepsilon)\mathrm{OPT}.
\end{equation}
Using ordinary LMOs only, it makes $O(n+n^6\varepsilon^{-8})$ LMO calls
and $O(1+n^6\varepsilon^{-8})$ exact cube value and gradient queries.
If, in addition to the $n$ initial LMO calls, a Euclidean projection
oracle over the width-normalized set $K_{\mathsf W}$ is available, it
instead uses $O(n^3\varepsilon^{-4})$ projections and
$O(1+n^3\varepsilon^{-4})$ objective queries. The bounds do not depend on
coordinate widths, gradient magnitudes, or a smoothness modulus.
\end{theorem}
\begin{proof}
Perform the width normalization. If $A_0=0$, \eqref{eq:anchor} shows
that $\widetilde f$ vanishes on $[0,1]^{n'}$, and any feasible output is
optimal. Otherwise take $\delta=\varepsilon/2$. Two exact full-gradient
queries, at $\mathsf WT_\delta(\zero)$ and $\mathsf WT_\delta(\one)$,
give every $H_i(f_\delta)$ and $H_{\Sigma}(f_\delta)$. If $H_{\Sigma}(f_\delta)=0$, return $y=\mathsf WT_\delta(\zero)$.
Since $f_\delta$ is constant, \eqref{eq:contraction-value}, applied at an
optimizer $o^\star$ of $\widetilde f$ over $K_{\mathsf W}$, gives
\[
 f(y)=f_\delta(\zero)=f_\delta(o^\star)
 \ge(1-\delta)\mathrm{OPT},
\]
which implies \eqref{eq:uniform-guarantee}.
Otherwise run the online algorithm of Corollary~\ref{thm:online} on the
constant sequence $f_t=f_\delta$ over $K_{\mathsf W}$. An LMO over
$K_{\mathsf W}$ in direction $c'$ is implemented by an LMO over $K$ in
direction $\mathsf W^{-1}c'$, followed by coordinatewise division by the
widths. A query of $\partial_if_\delta$ is a scaled original gradient
query at $\mathsf WT_\delta(z)$, which lies in the original cube.

The field estimate uses the sampler of Lemma~\ref{lem:one-query}
(Algorithm~\ref{alg:sampler}) with one change: the coordinate $I$ is drawn with probability
$\pi_I=H_I(f_\delta)/H_{\Sigma}(f_\delta)$ instead of $1/n'$, and the
multiplier $n'C_{\mathrm{int}}$ in \eqref{eq:one-gradient-estimator} is
replaced by $C_{\mathrm{int}}/\pi_I$. Conditional on $I$, signed-term
sampling gives mean $g_I(f_\delta,x)e_I/\pi_I$, so averaging over $I$
preserves unbiasedness. A coordinate with
$H_I(f_\delta)=0$ has $\partial_If_\delta\equiv0$, so it is never drawn and
its field component vanishes. Since
$|\partial_If_\delta(Y)|\le H_I(f_\delta)$ at every query point,
every estimate satisfies
$\norm{\widehat g}\le C_{\mathrm{int}}H_I(f_\delta)/\pi_I
=C_{\mathrm{int}}H_{\Sigma}(f_\delta)$ almost surely. Thus
$B_{\mathrm{lin}}=C_{\mathrm{int}}H_{\Sigma}(f_\delta)$ and
$\mathcal V=B_{\mathrm{lin}}^2$ are valid in Corollary~\ref{thm:online},
and $D=\sqrt{n'}$.

By \eqref{eq:uniform-gradient},
$H_{\Sigma}(f_\delta)/A_0\le2(n'+1)^2/(\delta n')$, and both $A_0$ and the
$H_i(f_\delta)$ are known after the initial queries. For online
Frank--Wolfe choose
\[
 T\ge\left(\frac{2\alpha_0c_{\mathrm{FW}}C_{\mathrm{int}}
     \sqrt{n'}H_{\Sigma}(f_\delta)}{\varepsilon A_0}\right)^4,
\]
and for projected ascent choose
\[
 T\ge\left(\frac{2\alpha_0C_{\mathrm{int}}\sqrt{n'}H_{\Sigma}(f_\delta)}
     {\varepsilon A_0}\right)^2.
\]
In both cases the average regret is at most $\varepsilon A_0/2$: by
Corollary~\ref{thm:online} with $D=\sqrt{n'}$ and
$B_{\mathrm{lin}}=C_{\mathrm{int}}H_\Sigma(f_\delta)$, it is at most
\[
 \alpha_0c_{\mathrm{FW}}\sqrt{n'}\,C_{\mathrm{int}}H_\Sigma(f_\delta)\,T^{-1/4}
 \quad\text{(OFW)},\qquad
 \alpha_0\sqrt{n'}\,C_{\mathrm{int}}H_\Sigma(f_\delta)\,T^{-1/2}
 \quad\text{(projected ascent)}.
\]
Select one played state $x_\jmath$ uniformly, as in
Corollary~\ref{cor:offline-conversion}, and return
$y=\mathsf WT_\delta(h(x_\jmath))\in K$. Then $f(y)=f_\delta(h(x_\jmath))$. Let $o^\star$ maximize $\widetilde f$
over $K_{\mathsf W}$, so that $\widetilde f(o^\star)=\mathrm{OPT}$.
Corollary~\ref{cor:offline-conversion} gives
$\E f(y)\ge\alpha_0f_\delta(o^\star)-\frac\varepsilon2A_0$, and
\eqref{eq:contraction-value} gives $f_\delta(o^\star)\ge(1-\delta)\mathrm{OPT}$.
Hence
\[
 \E f(y)\ge\alpha_0(1-\delta)\mathrm{OPT}-\frac\varepsilon2A_0
 \ge\left(\alpha_0-\frac{\alpha_0\varepsilon}2-\frac\varepsilon2\right)
 \mathrm{OPT}\ge(\alpha_0-\varepsilon)\mathrm{OPT},
\]
because $A_0=\widetilde f(\bar z)\le\mathrm{OPT}$ (as $\bar z\in K_{\mathsf W}$)
and $\alpha_0<1$.
Since $\varrho(1-\varrho)=n'/(n'+1)^2$ and $\delta=\varepsilon/2$,
\eqref{eq:uniform-gradient} gives
$\sqrt{n'}H_\Sigma(f_\delta)/(\varepsilon A_0)\le4(n'+1)^2/(\sqrt{n'}\varepsilon^2)
=O(n'^{3/2}\varepsilon^{-2})$, so the horizons are
$O(n'^6\varepsilon^{-8})$ and $O(n'^3\varepsilon^{-4})$, respectively.
Each round uses at most one gradient query and one constraint-oracle
call; adding the $n$ width computations, the anchor value, and the two
endpoint gradients gives the stated counts. The guarantee is in
expectation over the field sampler and the selected iterate.
\end{proof}

Rescaling preserves ordinary LMOs but not Euclidean projection: a
projection oracle over $K_{\mathsf W}$ is a weighted projection oracle
over $K$, and ordinary Euclidean projection onto $K$ alone does not
supply it. The LMO route of Theorem~\ref{thm:uniform} needs only the
oracles of Section~\ref{sec:online-model}.

The exact-oracle upper bound applies to the same exact-oracle model.
\begin{remark}[Upper bound in the uniform offline model]
\label{rem:smooth-hardness}
No algorithm in the model \eqref{eq:uniform-model} achieves a
coefficient above
\begin{equation}\label{eq:global-beta}
 \beta_* = \frac{4\tau_*}{1+6\tau_*+\tau_*^2}
           \approx0.47043868,
\end{equation}
where $\chi_*=1-e^{-1/2}$,
$q_*=\sqrt{(3\chi_*-1)/(2\chi_*)}$, and $\tau_*=2q_*(1-q_*)$.
This is the symmetry-gap bound of
\citet[Lemma~6.2 and Corollary~6.4]{geometry2026}, whose hard instances
are nonnegative $C^1$ DR-submodular objectives with exact values and full
gradients. Ordinary LMO replies on the explicit constraint set reveal
nothing about the random relabeling of coordinates that these
instances hide, and exact replies are a special
case of the noisy oracles, so the bound covers those models as well.
\end{remark}

\subsection{Standard discrete reductions}\label{subsec:discrete}
For a nonnegative set function $f:2^N\to\R_{\ge0}$, an exact
set-value oracle returns $f(S)$ for each queried subset $S$.
Submodularity means
$f(A)+f(B)\ge f(A\cup B)+f(A\cap B)$ for $A,B\subseteq N$.
Its multilinear extension is
\begin{equation}\label{eq:discrete-multilinear}
 \bar f(y)=\E f(R_y),\qquad y\in[0,1]^N,
\end{equation}
where $R_y$ includes each element independently with probability $y_i$.
This polynomial is nonnegative and DR-submodular: its mixed second
partials are expectations of submodular differences, and its pure
second partials vanish. The matroid and knapsack relaxations below
contain every feasible indicator, so their fractional optimum
$O_{\rm frac}=\max_{y\in K}\bar f(y)$ is at least the discrete optimum $O$.

Elements with infeasible singletons are removed. On the remaining
ground set of size $n\ge1$, the scale
\begin{equation}\label{eq:discrete-value-scale}
 \mathsf v=\max\{f(\varnothing),\max_{i\in N}f(\{i\})\}\le O
\end{equation}
satisfies $0\le f(S)\le n\mathsf v$ for every $S$. Indeed, for
$S\ne\varnothing$, submodularity gives
$f(S)\le f(\varnothing)+\sum_{i\in S}[f(\{i\})-f(\varnothing)]
\le|S|\mathsf v-(|S|-1)f(\varnothing)\le|S|\mathsf v$. If no elements remain or $\mathsf v=0$, the empty set is optimal.
For each requested partial derivative, the sample
\[
 \Delta_i=f(R\cup\{i\})-f(R),\qquad
 R\subseteq N\setminus\{i\},
\]
with independent inclusion probabilities from the query point, is
unbiased and satisfies $|\Delta_i|\le n\mathsf v$.
At round $t$, apply the direct estimator \eqref{eq:direct-scaled-field}:
for each coordinate $i$, independently, draw one term of the
representation \eqref{eq:gradient-only-small}--\eqref{eq:gradient-only-large}
with probability equal to its absolute coefficient divided by
$C_{\mathrm{int}}$, or a zero outcome with the remaining probability,
and multiply the sampled derivative by $C_{\mathrm{int}}\operatorname{sign}(\omega)$.
Each sampled derivative is replaced by the sample $\Delta_i$ above, so a
full field estimate uses at most $2n$ set-value queries. Conditioned on
the history $\mathcal H_t$ before the sampling, this representation and
$|\Delta_i|\le n\mathsf v$ give
\begin{equation}\label{eq:discrete-field-moment}
 \E[\widehat g\mid\mathcal H_t]=g(\bar f,x_t),\qquad
 |\widehat g_i|\le C_{\mathrm{int}}n\mathsf v,\qquad
 \norm{\widehat g}\le C_{\mathrm{int}}n^{3/2}\mathsf v.
\end{equation}
Fix $0<\zeta<\alpha_0$. Run the online algorithm of
Section~\ref{subsec:ombq} on the constant sequence $\bar f$ with these field
estimates and $D=\sqrt n$, and select a played action uniformly at
random as in the proof of Corollary~\ref{cor:offline-conversion}. This
gives a feasible fractional output $X$ with
$\E \bar f(X)\ge\alpha_0O_{\rm frac}-\Reg_{\alpha_0}(T)/T$. With the horizon
chosen below, the average regret is at most $\zeta \mathsf v$, so
\begin{equation}\label{eq:discrete-fractional-direct}
 \E \bar f(X)\ge\alpha_0O_{\rm frac}-\zeta \mathsf v
             \ge(\alpha_0-\zeta)O_{\rm frac},
\end{equation}
where the second inequality uses $\mathsf v\le O\le O_{\rm frac}$. By
\eqref{eq:discrete-field-moment}, Corollary~\ref{thm:online} applies with
$\mathcal V=B_{\mathrm{lin}}^2=C_{\mathrm{int}}^2n^3\mathsf v^2$ and $D=\sqrt n$,
and gives the average-regret bounds
\begin{equation}\label{eq:discrete-average-error}
 \frac{\Reg_{\alpha_0}(T)}T\le
 \alpha_0C_{\mathrm{int}}n^2\mathsf v
 \begin{cases}
 T^{-1/2},&\text{projected ascent},\\
 c_{\mathrm{FW}}T^{-1/4},&\text{online Frank--Wolfe}.
 \end{cases}
\end{equation}
Requiring the right side to be at most
$\zeta \mathsf v$ gives $T=O(n^4\zeta^{-2})$ in the projected route and
$T=O(n^8\zeta^{-4})$ in the LMO route.
Each round uses $2n$ set-value queries and one constraint-oracle call,
so the total set-value counts are $O(n^5\zeta^{-2})$ and
$O(n^9\zeta^{-4})$, respectively.

For a matroid constraint, swap rounding preserves the fractional
guarantee in expectation. A matroid $\mathcal M=(N,\mathcal I)$ consists
of a ground set $N$ and a family $\mathcal I$ of independent sets
containing $\varnothing$, closed under subsets, and satisfying exchange:
if $I,J\in\mathcal I$ and $|I|<|J|$, some $j\in J\setminus I$ has
$I\cup\{j\}\in\mathcal I$. An independence oracle tests membership
in $\mathcal I$. Its relaxation is
$P(\mathcal M)=\operatorname{conv}\{\one_I:I\in\mathcal I\}$.
\begin{corollary}[Matroid consequence of swap rounding]
\label{cor:discrete-matroid}
For every $0<\varepsilon<\alpha_0$, randomized swap rounding
\citep[Theorem~II.1]{chekuri2010swap} applied to the fractional output
\eqref{eq:discrete-fractional-direct} gives $I\in\mathcal I$ with
\[
 \E f(I)\ge(\alpha_0-\varepsilon)\max_{J\in\mathcal I}f(J).
\]
The set-value and constraint-oracle counts are those above with
$\zeta=\varepsilon$, with polynomially many additional independence
queries and rounding operations. The projected route uses a supplied polynomial-size convex
decomposition of its output and an independence oracle. The LMO route
uses only set-value and independence oracles.
\end{corollary}
\begin{proof}
The cited theorem applies to points of the matroid polytope
$P(\mathcal M)$ and holds for every submodular function, not
necessarily monotone or nonnegative. To round independent sets of
different sizes, augment them to bases with dummy elements and extend
the objective to ignore those elements. Each swap preserves the mean
and moves along a direction $e_i-e_j$ in the augmented space. In the
original coordinates, a nonzero move has direction $e_i-e_j$ or
$\pm e_i$. The multilinear extension is convex along the former,
with second derivative $-2\partial_{ij}\bar f\ge0$, and affine along
the latter. Conditional on $X$, the theorem gives
$\E[f(I)\mid X]\ge \bar f(X)$; hence
\begin{equation}\label{eq:matroid-final-value}
 \E f(I)\ge\E \bar f(X)\ge\alpha_0O_{\rm frac}-\varepsilon \mathsf v
             \ge(\alpha_0-\varepsilon)O.
\end{equation}
For the LMO route, greedy linear optimization omits negative-weight
elements and uses the independence oracle. The convex-combination
updates retain a decomposition of
each state $x$ into independent-set indicators. A decomposition of the
played action $h(x)\le x$ follows by deleting element $i$ from support
sets until their total deleted weight is $x_i-h_i(x)$. At most one
weight is split for each coordinate, so the support grows by at most
$n$. These deletions preserve independence and leave the other
coordinate marginals unchanged. The decomposition required by swap
rounding on the independence polytope is therefore available in
polynomial time.
\end{proof}

The knapsack reduction below is polynomial in $n$ for a fixed accuracy
and a fixed number of knapsacks; unlike the continuous guarantees above,
its complexity is not jointly polynomial in $n$ and $\varepsilon^{-1}$.
For fixed $k\ge1$, nonnegative costs
$A\in\R_{\ge0}^{k\times n}$ and capacities
$b^{\rm cap}\in\R_{\ge0}^k$ impose $A\one_S\le b^{\rm cap}$.
The relaxation is $K=\{x\in[0,1]^N:Ax\le b^{\rm cap}\}$.
The reduction of \citet{kulik2013knapsack} enumerates small sets of
items and rounds the residual fractional solutions.
\begin{corollary}[Knapsack consequence of continuous-to-discrete rounding]
\label{cor:discrete-knapsack}
For every $0<\varepsilon<\alpha_0$, the continuous-to-discrete
reduction of \citet[Section~2]{kulik2013knapsack}, supplied with
\eqref{eq:discrete-fractional-direct}, returns a feasible set $S$ with
\[
 \E f(S)\ge(\alpha_0-\varepsilon)
             \max\{f(J):A\one_J\le b^{\rm cap}\}.
\]
In the exact-real arithmetic and sampling model of
Section~\ref{sec:online-model}, its set-value queries and other
operations number
$(n+1)^{O_k(\varepsilon^{-4})}\operatorname{poly}(n,\varepsilon^{-1})$.
The fractional calls may use either projections or ordinary linear
maximization over the explicitly specified residual knapsack polytopes.
\end{corollary}
\begin{proof}
For $0<\xi<1$, the cited reduction enumerates feasible sets $J$ of at
most $\lceil k\xi^{-4}\rceil$ items. For each guess, it reduces the
capacities to $b^{\rm res}=b^{\rm cap}-A\one_J$, retains the items of
$J$ as optional items with zero residual cost, and retains only those
unguessed items $i$ with $A_{\cdot i}\le\xi^3b^{\rm res}$
\citep[Section~2]{kulik2013knapsack}. The objective is the restriction
of the unchanged function $f$ to the retained ground set. The residual
polytope is down-closed, and every residual-feasible set is feasible in
the original instance: its original cost is at most its residual cost
plus $A\one_J$. Rows with zero residual capacity have zero residual
costs on all retained items and may be dropped. One residual instance
has optimum $O_{\rm res}\ge(1-\xi)O$. The lower
bound $\mathsf v$ in \eqref{eq:discrete-value-scale} is evaluated for each
residual instance.
For a feasible fractional point $x$ of a residual instance, let $R_x$
contain each element independently with probability $x_i$, so that
$\E f(R_x)=\bar f(x)$. Following \citet[Section~2.2]{kulik2013knapsack}, let
$D_x=R_x$ if $R_x$ is $\xi$-nearly feasible (it exceeds each residual capacity by at most
a factor $1+\xi$) and $D_x=\varnothing$
otherwise, and let $S\subseteq D_x$ be the feasible set produced by
their feasibility repair \citep[Section~2.3]{kulik2013knapsack}. We use
two estimates from that paper, with constants $C_k\ge1$ and
$\xi_k>0$ depending only on $k$, valid for $0<\xi\le\xi_k$; decrease
$\xi_k$ if necessary so that $\xi_k<1$ and $C_k\xi_k<1$.

First, the discarded-value argument of \citet[Section~2.2]{kulik2013knapsack} bounds the
expected value of $R_x$ on the event that $R_x$ is not
$\xi$-nearly feasible by tail-probability and discarded-value
estimates, which give
\[
 \E[f(R_x)]\le\E\bigl[f(R_x)\one\{D_x=R_x\}\bigr]+C_k\xi O_{\rm res}.
\]
This step does not use a lower bound on $\bar f(x)$; that paper uses such a
lower bound only afterwards, to turn the additive term into a relative
one. Since $f\ge0$, $\E f(D_x)\ge\E[f(R_x)\one\{D_x=R_x\}]$, and
therefore
\[
 \E f(D_x)\ge \bar f(x)-C_k\xi O_{\rm res}.
\]
Second, the feasibility repair satisfies $f(S)\ge(1-C_k\xi)f(D_x)$
for every realization of $D_x$. Taking expectations and using
$0\le1-C_k\xi\le1$,
\[
 \E f(S)\ge(1-C_k\xi)\E f(D_x)\ge(1-C_k\xi)\bar f(x)-C_k\xi O_{\rm res}.
\]
Both estimates hold for every feasible $x$, with no lower bound on
$\bar f(x)$. They therefore hold conditionally on the random fractional
output:
\begin{equation}\label{eq:knapsack-conditional-rounding}
 \E[f(S)\mid X=x]\ge
 (1-C_k\xi)\bar f(x)-C_k\xi O_{\rm res}.
\end{equation}

Take $\zeta=\varepsilon/2$ and
$\xi=\min\{\xi_k,\varepsilon/[2(2C_k+1)]\}$. In the residual
instance, $\mathsf v\le O_{\rm res}\le O^{\rm res}_{\rm frac}$, where
$O^{\rm res}_{\rm frac}$ is the fractional optimum of the residual
instance, so
\eqref{eq:discrete-fractional-direct} gives
$\E \bar f(X)\ge(\alpha_0-\zeta)O_{\rm res}$. Taking expectations in
\eqref{eq:knapsack-conditional-rounding} and using $C_k\xi<1$ and
$\alpha_0-\zeta\le1$,
\[
 \E f(S)\ge[(1-C_k\xi)(\alpha_0-\zeta)-C_k\xi]O_{\rm res}
 \ge(\alpha_0-\zeta-2C_k\xi)O_{\rm res}.
\]
The last coefficient is positive because $\zeta+2C_k\xi<\varepsilon<\alpha_0$.
Using $O_{\rm res}\ge(1-\xi)O$,
\begin{equation}\label{eq:knapsack-final-value}
 \E f(S)\ge(\alpha_0-\zeta-2C_k\xi)(1-\xi)O
 \ge[\alpha_0-\zeta-(2C_k+1)\xi]O\ge(\alpha_0-\varepsilon)O .
\end{equation}
The algorithm returns, among the rounded sets of all enumerated residual
instances, one with the largest exact set value. Its value is at least
that of the set obtained from the residual instance above, so it
satisfies the same bound. The cited
enumeration uses $(n+1)^{O_k(\xi^{-4})}$ fractional calls, and
$\xi=\Theta_k(\varepsilon)$, proving the stated complexity.
\end{proof}

\begin{remark}[The remaining discrete approximation gap]
\label{rem:matroid-gap}
Corollary~\ref{cor:discrete-matroid} makes $\alpha_0$ achievable
arbitrarily closely with polynomially many set-value and independence
queries. The upper bound $\beta_*\approx0.47043868$ in
\eqref{eq:global-beta} also holds for nonnegative submodular set functions
under cardinality constraints, as stated in the remark following
Theorem~1.1 of \citet{feldman2026hardness}; it therefore applies to the
matroid class as well.
\end{remark}

\Needspace{14\baselineskip}
\section*{Conclusions}\label{sec:conclusion}
The comparison \eqref{eq:structural} gives the approximation
coefficient $\alpha_0=54440750000/117486474169>0.463378$ for general
nonnegative DR-submodular objectives over down-closed convex sets, with
an analytical proof. Its averaging weights $w,\varpi$ in \eqref{eq:wW}
are piecewise polynomial in the path time, and on the initial interval
$[0,t_0)$ only the small-block comparison is used. The one-query field
estimator of Algorithm~\ref{alg:sampler} gives the same coefficient with
projected ascent or online Frank--Wolfe; standard smoothing,
online-to-batch conversion, and rounding give the value-feedback,
offline, and discrete consequences. Contracting the objective toward a
feasible interior point gives the uniform offline guarantee of
Theorem~\ref{thm:uniform} for the full nonnegative $C^1$ class, with no
gradient, width, or smoothness parameter. For the continuous
DR-submodular oracle model, the known oracle upper bound is
$\beta_*\approx0.47043868$
\citep[Corollary~6.4 and Proposition~7.3]{geometry2026}; it also applies
to the exact-oracle model of Theorem~\ref{thm:uniform} by
Remark~\ref{rem:smooth-hardness}. Appendix~\ref{sec:geometric} shows that
$2\sqrt3-3$ is the best coefficient attainable by any map $h$ with
$0\le h(x)\le x$ that does not depend on the objective, measured against
the concave closure $f^+(x)$ (the largest average of $f$ over
distributions with mean $x$), and that bounded comparison vectors exist
at this coefficient.

\Needspace{5\baselineskip}
\section*{Acknowledgments}
Vaneet Aggarwal thanks Hillclimb for providing access to generative
artificial intelligence tools. GPT and Claude assisted with developing
proof arguments, checking algebraic calculations, and revising the
exposition. The authors take responsibility for the final manuscript.
\appendix
\section{Proof of Lemma~\ref{lem:scalar}}\label{app:scalar}

This appendix proves Lemma~\ref{lem:scalar} for the rational schedule of
Section~\ref{sec:field}. Section~\ref{app:parameters} records the
constants used below and reduces a nonincreasing profile $q$ to a
constant one (Lemma~\ref{lem:ordered-scalar-reduction}).
Section~\ref{app:prefix} proves the inequality for intervals that start
at time zero (Lemma~\ref{lem:prefix-inequality}) and for intervals
$[\tau,1]$ with $\tau\le t_0$ (Lemma~\ref{lem:head-terminal}).
Section~\ref{app:boundary-proof} locates the crossing point $z$ and
bounds $\mathcal B$ from below
(Lemma~\ref{lem:boundary-estimates}). Section~\ref{app:all-intervals}
proves four bounds on the integrands $F_q$ for constant $q$
(Lemma~\ref{lem:terminal-reduction}). Section~\ref{app:complete}
combines these results into the proof of Lemma~\ref{lem:scalar}.

\subsection{Elementary parameter bounds and ordered profiles}\label{app:parameters}

The mass identities of Section~\ref{sec:field} give
$\int_0^1w=1-\beta$ and $\int_0^1\varpi=\beta$, so
$\mu(dt)=(w(t)+\varpi(t))dt$ is a probability measure.
Its tail $W$ is continuous, equals one before $t_0$, and satisfies
$W'=-(w+\varpi)$ almost everywhere. Since $w(t_0+)=b$, $w$ increases after $t_0$,
and $W$ decreases, $\lambda=w-bW\ge0$ there. On the initial interval
$\lambda=0$ and $\kappa=b$. Thus the mass and tail inequalities in
\eqref{eq:mass-identities} hold. In this appendix write
\[
 \mathcal U(T)=\int_0^TW(t)\dd,\qquad
 M(T)=\int_0^Tw(t)\dd,\qquad L(T)=\int_0^T\lambda(t)\dd.
\]
We call $[t_0,t_1]$ the ramp. On it, with $x=t-t_0$,
\[
 w=b+w_1x+w_2x^2,\qquad
 \varpi=\varpi_0+\varpi_1x,\qquad
 W=1-(b+\varpi_0)x-\frac{w_1+\varpi_1}{2}x^2-\frac{w_2}{3}x^3,
\]
On $(t_1,1]$, $\varpi=0$, and $w$ is linear on $[t_1,1]$.
With $\delta_1=t_1-t_0=97/250$, the formulas of Section~\ref{sec:field}
give $w(t_1)=b+w_1\delta_1+w_2\delta_1^2$, $w(1)=w(t_1)+w_3(1-t_1)$,
$W(t_1)=(1-t_1)w(t_1)+\tfrac{w_3}2(1-t_1)^2$, and
$J:=\mathcal U(1)=\mathcal U(t_1)+(1-t_1)W(t_1)-\tfrac{(1-t_1)^2}2w(t_1)-\tfrac{w_3}6(1-t_1)^3$,
where $\mathcal U(t_1)=\int_0^{t_1}W$ follows by integrating the ramp
formula for $W$. Evaluating these expressions
gives
\begin{equation}\label{eq:J}
 J=\frac{3823017045539993}{6255000000000000},
 \qquad J_-:=\frac{611193}{10^6}<J<\frac{611194}{10^6}=:J_+.
\end{equation}
The following rational evaluations will be used below:
\begin{equation}\label{eq:scalar-brackets}
\begin{gathered}
 w(t_1)=\frac{4560014431}{4170000000},\qquad
 W(t_1)=\frac{841892829781}{2085000000000},\\
 \mathcal U(t_1)=\frac{168143863477627}{312750000000000},\qquad
 w(1)=\frac{4742668771}{4170000000}.
\end{gathered}
\end{equation}
The exponential lower bounds used below follow from the alternating
Taylor estimate
\begin{equation}\label{eq:exp-remainder}
 e^{-x}\ge1-x+\frac{x^2}{2}-\frac{x^3}{6}
                         +\frac{x^4}{24}-\frac{x^5}{120},\qquad 0<x<1.
\end{equation}
Polynomial integration and exact rational comparisons establish the
remaining parameter bounds.

\begin{lemma}[Reduction of ordered scalar profiles]\label{lem:ordered-scalar-reduction}
Let $0\le \tau<T\le1$, and let $q:[\tau,T]\to[1/2,1]$ be nonincreasing. Write
\[
 \mathcal U_{\tau,T}=\int_\tau^T W(t)\,dt,\qquad
 L_{\tau,T}=\int_\tau^T\lambda(t)\,dt,\qquad
 \overline q=\frac{1}{\mathcal U_{\tau,T}}\int_\tau^T W(t)q(t)\,dt.
\]
Then
\begin{align}
 &\int_\tau^T\left\{\Lambda(q(t))W(t)
            +\left(1+\frac{q(t)}b\right)\lambda(t)-\Gamma\right\}dt
 \notag\\
 &\hspace{1cm}\le
 \left(1+\frac{\overline q}b\right)
       \left\{L_{\tau,T}+(b+1-\overline q)\mathcal U_{\tau,T}\right\}-\Gamma(T-\tau).
 \label{eq:constant-profile-reduction}
\end{align}
Consequently it suffices to prove the scalar comparison for constant $q$.
\end{lemma}
\begin{proof}
On $[0,1)$, where $W>0$, the quotient $\lambda/W$ vanishes before
$t_0$ and increases afterwards. Indeed, $\lambda/W=w/W-b$ after
$t_0$, and
$(w/W)'=[w'W+w(w+\varpi)]/W^2>0$ on each open piece.
Consequently $\lambda(t)W(r)\ge\lambda(r)W(t)$ whenever $r\le t$;
this also holds for $t=1$, since $W(1)=0$ and $\lambda\ge0$.
Because $q$ is nonincreasing,
\[
 (q(t)-q(r))\bigl(\lambda(t)W(r)-\lambda(r)W(t)\bigr)\le0.
\]
For $r>t$ the same product is also nonpositive, by symmetry.
Integrating over $(r,t)\in[\tau,T]^2$ and expanding gives
$2\bigl(\mathcal U_{\tau,T}\int_\tau^T\lambda q-L_{\tau,T}\int_\tau^TWq\bigr)\le0$, that is,
$\int_\tau^T\lambda q\le L_{\tau,T}\overline q$.
Also $\int_\tau^TWq^2\ge\mathcal U_{\tau,T}\overline q^{\,2}$ by Jensen's
inequality. Expanding $\Lambda(q)=1+b+q(1-q)/b$, the left side of
\eqref{eq:constant-profile-reduction} equals
\[
\begin{aligned}
 &(1+b)\mathcal U_{\tau,T}+L_{\tau,T}
 +\frac1b\Bigl(\int_\tau^TWq-\int_\tau^TWq^2+\int_\tau^T\lambda q\Bigr)-\Gamma(T-\tau)\\
 &\qquad\le(1+b)\mathcal U_{\tau,T}+L_{\tau,T}
 +\frac1b\bigl(\mathcal U_{\tau,T}\overline q-\mathcal U_{\tau,T}\overline q^{\,2}
 +L_{\tau,T}\overline q\bigr)-\Gamma(T-\tau),
\end{aligned}
\]
and the right side is the right side of
\eqref{eq:constant-profile-reduction}, because
$(1+\overline q/b)(b+1-\overline q)=\Lambda(\overline q)$.
\end{proof}

For a constant $q\in\R$ (later $q\in[1/2,1]$), define
\begin{equation}\label{eq:Fq-definition}
 F_q(t)=\Lambda(q)W(t)+\left(1+\frac qb\right)\lambda(t)-\Gamma,
 \qquad P_q(T)=\int_0^T F_q(t)\,dt.
\end{equation}
Before $t_0$, $F_q(t)=-(q-1/2)^2/b$; after $t_0$,
\begin{equation}\label{eq:Fq-after-t0}
 F_q(t)=\left(1+\frac qb\right)\{w(t)+(1-q)W(t)\}-\Gamma.
\end{equation}

\subsection{Prefix intervals}\label{app:prefix}
For intervals that start at time zero, the right side of
\eqref{eq:scalar} is $\mathcal B(0)=0$. We use the following elementary sign test.
If $p\in C^3([a,d])$, $p'''\le0$, and $p'(a)>0$, then
\begin{equation}\label{eq:one-change-derivative}
 p(t)\ge\min\{p(a),p(d)\}\qquad(a\le t\le d).
\end{equation}
Indeed, $p'$ is concave, so its nonnegative set is an interval
containing $a$. Thus $p$ first increases and can subsequently only
decrease.

\begin{lemma}[Prefix inequality]\label{lem:prefix-inequality}
For every $T\in[0,1]$ and $q\in\mathbb R$, $P_q(T)\le0$.
\end{lemma}
\begin{proof}
For $T\ge t_0$, put $\Upsilon=2b+1$, $M_\beta=M+\beta$, and
$Z=(1+b)\mathcal U+M_\beta$. Since $\lambda=w-bW$ after $t_0$, $\lambda=0$
before $t_0$, and $bt_0=\beta$,
\[
 L=\int_{t_0}^T(w-bW)=M-b(\mathcal U-t_0)=M_\beta-b\mathcal U .
\]
Expanding $\Lambda(q)=1+b+q(1-q)/b$ and $\cprof(q)=1+q/b$ in $P_q(T)$,
\[
 bP_q(T)=-\mathcal Uq^2+(\mathcal U+L)q+b[(1+b)\mathcal U+L-\Gamma T],
\]
with $\mathcal U+L=s\mathcal U+M_\beta$ and $(1+b)\mathcal U+L=\mathcal U+M_\beta$.
Its discriminant is
\[
 (s\mathcal U+M_\beta)^2+4b\,\mathcal U(\mathcal U+M_\beta)-4b\Gamma\,\mathcal UT
 =Z^2-\Upsilon^2\mathcal UT,
\]
because $s^2+4b=(1+b)^2$, $s+2b=1+b$, and $4b\Gamma=\Upsilon^2$. Since the
leading coefficient $-\mathcal U$ is negative, $P_q(T)\le0$ for all real
$q$ exactly when
\begin{equation}\label{eq:prefix-unrestricted-square}
 N(T):=\Upsilon^2\mathcal U(T)T-Z(T)^2\ge0.
\end{equation}

For $t_0\le T\le t_1$, set $y=T-t_0$ and
\[
 \upsilon_2=\frac{b+\varpi_0}2=\frac{2830919}{6062500},\quad
 \upsilon_3=\frac{w_1+\varpi_1}6=\frac{476427691363}{941652720000},\quad
 \upsilon_4=\frac{w_2}{12}=\frac{11}{480}.
\]
Direct integration gives
\[
 \mathcal U=t_0+y-\upsilon_2y^2-\upsilon_3y^3-\upsilon_4y^4,\qquad
 M_\beta=bt_0+by+\frac{w_1}2y^2+\frac{w_2}3y^3.
\]
Hence $Z=\Upsilon t_0+\Upsilon y+Z_2y^2+Z_3y^3+Z_4y^4$, where
\[
 Z_2=\frac{w_1}2-(1+b)\upsilon_2,\qquad
 Z_3=\frac{w_2}3-(1+b)\upsilon_3,\qquad Z_4=-(1+b)\upsilon_4.
\]
Here $\mathcal U T=t_0^2+2t_0y+O(y^2)$ and $Z=\Upsilon(t_0+y)+O(y^2)$, so
the terms of degree $0$ and $1$ in $N(t_0+y)$ cancel; thus
$N(t_0+y)=y^2N_0(y)$ with
$N_0(y)=n_2+n_3y+n_4y^2+n_5y^3+n_6y^4+n_7y^5+n_8y^6$. The argument
below uses only the coefficients
\begin{align*}
 n_2&=\Upsilon t_0(\upsilon_2-w_1),&
 n_3&=\Upsilon[\upsilon_2-w_1+t_0(\upsilon_3-2w_2/3)],\\
 n_5&=\Upsilon\upsilon_4-2Z_2Z_3,&
 n_6&=-Z_3^2-2Z_2Z_4,\\
 n_7&=-2Z_3Z_4,&n_8&=-Z_4^2.
\end{align*}
These follow by collecting powers of $y$ and using $2(1+b)-\Upsilon=1$;
for example,
$n_2=\Upsilon^2(1-t_0\upsilon_2)-\Upsilon^2-2\Upsilon t_0Z_2=\Upsilon t_0(\upsilon_2-w_1)$
and $n_5=-\Upsilon^2\upsilon_4-2\Upsilon Z_4-2Z_2Z_3=\Upsilon\upsilon_4-2Z_2Z_3$.
The exact constants imply
\[
 \upsilon_2-w_1=\frac{10610093}{40449000000}>\frac1{4000},\quad
 \upsilon_3>\frac12,\quad Z_2<-\frac35,\quad Z_3<-\frac45,\quad Z_4<0.
\]
Since $\Upsilon<3$ and $\upsilon_4<1/40$, one has
$n_2>0$, $n_3>0$ (because $\upsilon_3>1/2>2w_2/3=11/60$), and
$n_5<3/40-24/25<0$, while $n_6,n_7,n_8<0$. Since $N_0'(0)=n_3$ and
$N_0'''(y)=6n_5+24n_6y+60n_7y^2+120n_8y^3$, we get $N_0'(0)>0$ and
$N_0'''(y)<0$ for $y\ge0$.
By \eqref{eq:one-change-derivative}, $N_0$ attains its minimum on
$[0,t_1-t_0]$ at an endpoint. Its initial value $n_2$ is positive; its
other endpoint is $N(t_1)/(t_1-t_0)^2$, proved positive next.

For $t_1\le T\le1$, put $y=T-t_1$, $\bar w=w(t_1)$, $\bar W=W(t_1)$, and
$\bar{\mathcal U}=\mathcal U(t_1)$. Then
\[
 W=\bar W-\bar wy-\frac{w_3}2y^2,\quad
 \mathcal U=\bar{\mathcal U}+\bar Wy-\frac{\bar w}2y^2-\frac{w_3}6y^3,\quad
 M_\beta=1-W.
\]
Write $Z=\bar Z_0+\bar Z_1y+\bar Z_2y^2+\bar Z_3y^3$, where
\[
 \bar Z_0=(1+b)\bar{\mathcal U}+1-\bar W,\quad \bar Z_1=(1+b)\bar W+\bar w,\quad
 \bar Z_2=\frac{w_3-(1+b)\bar w}2,\quad \bar Z_3=-\frac{(1+b)w_3}6.
\]
The coefficients of degree three and higher in $N(t_1+y)$ are
\begin{align*}
 \bar n_3&=-\Upsilon^2(\bar w/2+t_1w_3/6)-2\bar Z_0\bar Z_3-2\bar Z_1\bar Z_2,\\
 \bar n_4&=-\Upsilon^2w_3/6-\bar Z_2^2-2\bar Z_1\bar Z_3,\qquad
 \bar n_5=-2\bar Z_2\bar Z_3,\qquad \bar n_6=-\bar Z_3^2.
\end{align*}
The rational bounds $\Upsilon^2>15/2$, $\bar w>109/100$, $t_1>63/100$,
$w_3>3/25$, $0<\bar Z_0,\bar Z_1<2$, $-1<\bar Z_2<-9/10$, and
$-1/25<\bar Z_3<0$ give
\[
 \bar n_3<-\frac{15}2\left(\frac{109}{200}+\frac{63}{5000}\right)
               +\frac4{25}+4=-\frac{11}{500}<0,\quad
 \bar n_4<-\frac{81}{100}+\frac4{25}<0,\quad \bar n_5,\bar n_6<0.
\]
Therefore $N'''<0$ on $[t_1,1]$. Direct substitution of
\eqref{eq:scalar-brackets} gives
\begin{equation}\label{eq:tail-endpoint-bounds}
\begin{gathered}
 N(t_1)>\frac{7429967}{625000000}>0,\qquad
 N'(t_1)>\frac{267782557}{5000000000}>0,\\
 N(1)>\Upsilon^2J_--[(1+b)J_++1]^2
       =\frac{11332136644853420879}{15625000000000000000000}>0.
\end{gathered}
\end{equation}
Here $N(t_1)=\Upsilon^2\bar{\mathcal U}t_1-\bar Z_0^2$ and
$N'(t_1)=\Upsilon^2(\bar{\mathcal U}+t_1\bar W)-2\bar Z_0\bar Z_1$, since
$\mathcal U'=W$. For these two bounds it suffices to use
$\Upsilon^2>7519/1000$, $\bar{\mathcal U}>5376/10000$, $\bar W>4037/10000$,
$\bar Z_0<16022/10000$, and $\bar Z_1<18491/10000$.
Equation~\eqref{eq:one-change-derivative} now proves $N\ge0$ on
$[t_1,1]$, and the preceding argument proves it on $[t_0,t_1]$.
For $T\le t_0$, $P_q(T)=-T(q-1/2)^2/b\le0$ directly.
\end{proof}

We also need intervals $[\tau,1]$ with $\tau\le t_0$.
Let
\begin{equation}\label{eq:boundary-early}
 \mathfrak k_0=\frac{77437}{500000},\qquad \mathfrak k=\Gamma e^{-s}-2b>\mathfrak k_0.
\end{equation}
For the last strict inequality, write
$p_5(x)=1-x+x^2/2-x^3/6+x^4/24-x^5/120$ for the right side of
\eqref{eq:exp-remainder}; then $\mathfrak k\ge\Gamma p_5(s)-2b$ and
$\Gamma p_5(s)-2b-\mathfrak k_0>3\cdot10^{-7}$.

On $[0,t_0]$, $\mathcal B$ has a closed form. By \eqref{eq:budget},
$\mathcal B'=\theta-\rho_1-s\mathcal B$, and $\theta'=s\theta+b\psi$ by
\eqref{eq:theta-identities}. Hence
\begin{equation}\label{eq:budget-second-derivative}
 \mathcal B''=s^2\mathcal B+b\psi+s\rho_1-\rho_1'
 \qquad\text{almost everywhere.}
\end{equation}
On $[0,t_0]$ we have $\lambda=0$ and $\kappa=b$, so $\rho_1=b$, and
$\psi=\rho_1$ because $z>3/4>t_0$; Lemma~\ref{lem:boundary-estimates}
proves $z>3/4$ from \eqref{eq:boundary-moment-bounds} alone, without
using the formula below. Thus
$\mathcal B''=s^2\mathcal B+b(b+s)=s^2\mathcal B+b$, with $\mathcal B(0)=0$
and $\mathcal B'(0)=\theta(0)-\rho_1(0)=\Gamma e^{-s}-2b=\mathfrak k$.
Solving this equation and using $\sinh(st)\ge st$,
$\cosh(st)-1\ge s^2t^2/2$, $\cosh(st)\ge1$, and $\mathfrak k>\mathfrak k_0$
gives
\begin{equation}\label{eq:boundary-head-formula}
 \mathcal B(t)=\mathfrak k\frac{\sinh(st)}s
          +b\frac{\cosh(st)-1}{s^2}
 \ge\mathfrak k_0t+\frac b2t^2,\qquad \mathcal B'(t)\ge\mathfrak k_0+bt.
\end{equation}

\begin{lemma}[{Intervals $[\tau,1]$ with $\tau\le t_0$}]
\label{lem:head-terminal}
For $0\le\tau\le t_0$ and $1/2\le q\le1$,
$P_q(1)-P_q(\tau)\le\mathcal B(\tau)$.
\end{lemma}
\begin{proof}
Put $r=q-1/2\in[0,1/2]$ and
$\mathfrak a(r)=(1/2-r)[1+(1/2+r)/b]$. Since
$\int_0^1\lambda=1-\int_0^1\kappa=1-bJ$ and $\Lambda(q)=\cprof(q)(b+1-q)$,
\[
 P_q(1)=\Lambda(q)J+\cprof(q)(1-bJ)-\Gamma=\cprof(q)(1-q)J+\cprof(q)-\Gamma .
\]
With $q=1/2+r$, $\cprof(q)(1-q)=\mathfrak a(r)$ and
$\cprof(q)-\Gamma=(r+1/4)/b-b$. Hence, and by the value of $F_q$ before
$t_0$,
\[
 P_{1/2+r}(1)=\mathfrak a(r)J+\frac{r+1/4}b-b,\qquad
 P_q(\tau)=-\frac{\tau r^2}b.
\]
By \eqref{eq:boundary-head-formula}, the desired gap
$\mathcal B(\tau)-P_q(1)+P_q(\tau)$ is at least
$-P_q(1)+(\mathfrak k_0-r^2/b)\tau+b\tau^2/2$.
For $r\le11/30$ this is nonnegative by Lemma~\ref{lem:prefix-inequality} and
\[
 b\mathfrak k_0-(11/30)^2=\frac{516220879}{1125000000000}>0.
\]
For the remaining $r$, minimizing over all real $\tau$ gives the lower
bound
\[
 \mathcal Q(r)=b-\frac{r+1/4}b-\mathfrak a(r)J
          -\frac{(r^2-b\mathfrak k_0)^2}{2b^3}.
\]
Its derivatives satisfy
\[
 \mathcal Q'(r)=\frac{(b+2r)J-1}b-\frac{2r(r^2-b\mathfrak k_0)}{b^3},\quad
 \mathcal Q''(r)=\frac{2J}b-\frac{6r^2-2b\mathfrak k_0}{b^3},\quad
 \mathcal Q'''(r)=-\frac{12r}{b^3}<0.
\]
Set $r_\star=43/100$. Since $\mathcal Q$ decreases in $J$ (its
$J$-coefficient is $-\mathfrak a(r)\le0$) while $\mathcal Q'$ and
$\mathcal Q''$ increase in $J$, we use $J_+$ for the lower bound on
$\mathcal Q$, $J_-$ for the lower bounds on $\mathcal Q'$ and
$\mathcal Q''$, and $J_+$ for the upper bound on $\mathcal Q'$. This gives
\[
 \mathcal Q(r_\star)>\frac3{10^5},\qquad
 0<\mathcal Q'(r_\star)<\frac{31}{20000},\qquad
 \mathcal Q''(r_\star)>\frac{13}{100},\qquad
 \mathcal Q(1/2)>\frac1{600000}.
\]
Since $\mathfrak a(1/2)=0$, the last value is exact:
\[
 \mathcal Q(1/2)=b-\frac3{4b}-\frac1{2b}\Bigl(\frac1{4b}-\mathfrak k_0\Bigr)^2
 =\frac{40822380145732727}{20652958314427894000000}.
\]
The same rational reappears in the proof of
Lemma~\ref{lem:boundary-estimates}.
On $[11/30,r_\star]$, $\mathcal Q''\ge\mathcal Q''(r_\star)>13/100$. Thus,
writing $y=r_\star-r$,
\[
 \mathcal Q(r)>\frac3{10^5}-\frac{31}{20000}y+\frac{13}{200}y^2
 \ge\frac{2159}{104000000}>0,
\]
where the quadratic in $y$ is minimized at $y=31/2600$.
On $[r_\star,1/2]$, \eqref{eq:one-change-derivative} gives
\[
 \mathcal Q(r)\ge\min\{\mathcal Q(r_\star),\mathcal Q(1/2)\}>0 .
\]
\end{proof}

\subsection{The boundary comparison}\label{app:boundary-proof}
For intervals that start after $t_0$ we need lower bounds on
$\mathcal B(\tau)$; the next lemma provides them.

\begin{lemma}[Boundary estimates]\label{lem:boundary-estimates}
The crossing point $z$ satisfies $z>3/4>t_1$; the identities
\eqref{eq:theta-identities} hold. The function $\mathcal B$ is
nonnegative on $[0,1]$, and
\begin{equation}\label{eq:boundary-terminal}
 \mathcal B(t)\ge \mathcal I(t):=\int_t^1[c_1w(r)-\Gamma]\,dr
 \qquad(t_0\le t\le1).
\end{equation}
Moreover,
\begin{equation}\label{eq:boundary-ramp-floor}
 \mathcal B(t)\ge\mathcal B(t_0)>\frac{527511}{8000000}
                                    >\frac{21}{320}
 \qquad(t_0\le t\le t_1).
\end{equation}
\end{lemma}
\begin{proof}
Since $\int(\lambda+\kappa)=1$ and $\kappa=bW$,
\[
 \int_0^1\rho_1=c_1-\int_0^1W\ge\frac1b,\qquad
 \int_0^1e^{-st}\rho_1(t)\,dt\ge e^{-s}/b>1.
\]
The last inequality follows from $e^{-s}>1-s=b$.
Since $\rho_1>0$ almost everywhere, there is a unique $z\in(0,1)$
with $\int_0^ze^{-st}\rho_1(t)dt=1$.
Consequently $\psi=\rho_1\mathbf1_{[0,z]}$ satisfies the normalization
in \eqref{eq:theta-identities}. The formula for $\theta$ gives
$\theta'=s\theta+b\psi$ and $\theta(1)=\Gamma$.
Also $\theta(0)=\Gamma e^{-s}-b=\mathfrak k+b>0$, so
$0<\theta\le\Gamma$.

We record five polynomial moment bounds. After $t_0$,
$\rho_1=c_1(w-bW)+bW=c_1w-W$ because $bc_1=1+b$. With the ramp
formulas for $w$ and $W$ and $x=t-t_0$, this gives
$\rho_1=b+(c_1w_1+b+\varpi_0)x+\bigl(c_1w_2+\tfrac{w_1+\varpi_1}2\bigr)x^2+\tfrac{w_2}3x^3$,
that is,
\[
 \rho_1=b+\frac{437340043016287}{225853733000000}x
   +\frac{144124767636938969}{68352373755120000}x^2+\frac{11}{120}x^3,
\]
and on the final interval, with $y=t-t_1$,
$\rho_1=(c_1\bar w-\bar W)+(c_1w_3+\bar w)y+\tfrac{w_3}2y^2$, that is,
\[
 \rho_1=\frac{126167129550332371}{64862265000000000}
   +\frac{1229021599447853}{908071710000000}y+\frac{121}{2000}y^2.
\]
These give coefficientwise polynomial bounds $\rho_-\le\rho_1\le\rho_+$:
\[
\begin{array}{c|c|c}
 \text{interval}&100\rho_-&100\rho_+\\\hline
 \lbrack0,t_0\rbrack&87&88\\
 \lbrack t_0,t_1\rbrack&87+193x+210x^2+9x^3&88+194x+211x^2+10x^3\\
 \lbrack t_1,1\rbrack&194+135y+6y^2&195+136y+7y^2
\end{array}
\]
Direct integration gives
\begin{equation}\label{eq:boundary-moment-bounds}
\begin{array}{c|c|c}
 \text{integral}&\text{exact value}&\text{bound}\\\hline
 \int_0^{3/4}\rho_+&457537229543/468750000000&<1\\
 \int_0^1\rho_+&482885476987/312500000000&<31/20\\
 \int_0^1(1-t)\rho_-&14208703332849601/23437500000000000&>3/5\\
 \int_0^1(1-t)^2\rho_+&167665405177023883/468750000000000000&<9/25\\
 \int_0^{3/4}(1-t)\rho_-&6296417960370113/11718750000000000&>53/100
\end{array}
\end{equation}
The first bound in \eqref{eq:boundary-moment-bounds} proves
$\int_0^{3/4}e^{-st}\rho_1(t)dt<1$, hence $z>3/4$.
In particular $\psi=\rho_1$ on the whole ramp.

Put $G(t)=\mathcal B(t)-\mathcal I(t)$, extending $\mathcal I$ to earlier times
by $\mathcal I(t)=\int_t^1F_1$.
Before $t_0$, $F_1=\Lambda(1)-\Gamma=-1/(4b)$, and
$\mathcal I(t_0)=c_1(1-\beta)-\tfrac34\Gamma=\tfrac{13}{16b}-b$. Hence, on
$[0,t_0]$, $\mathcal I(t)=\mathcal I(t_0)-\frac{t_0-t}{4b}=\frac{3+t}{4b}-b$; in
particular $\mathcal I'=1/(4b)$ there, so $G'=\mathcal B'-1/(4b)$.
The lower bound in \eqref{eq:boundary-head-formula} and minimization of
the resulting quadratic over all real $t$ give
\[
 G(t)\ge b-\frac3{4b}-\Bigl(\frac1{4b}-\mathfrak k_0\Bigr)t+\frac b2t^2
 \ge b-\frac3{4b}-\frac1{2b}\Bigl(\frac1{4b}-\mathfrak k_0\Bigr)^2
 =\frac{40822380145732727}{20652958314427894000000}>0.
\]
This is the value $\mathcal Q(1/2)$ from the proof of
Lemma~\ref{lem:head-terminal}.
In particular $G'(t_0)\ge\mathfrak k_0+bt_0-1/(4b)$ and
$\mathcal B(t_0)>\mathfrak k_0/4+b/32=527511/8000000$.

For $t\ge t_0$, $F_1=c_1w-\Gamma$ and $w$ increases, so $\mathcal I$
is concave, with
\[
 \mathcal I(t_0)=\frac{258501987}{4187750000}>0,\qquad \mathcal I(1)=0.
\]
Thus $\mathcal I\ge0$ there. After $t_0$, $\rho_1=c_1(w-bW)+bW=c_1w-W$ because
$bc_1=1+b$, so $\rho_1'=c_1w'+w+\varpi$, while $\mathcal I''=-F_1'=-c_1w'$. By
\eqref{eq:budget-second-derivative},
\[
 G''=s^2\mathcal B+b\psi+s\rho_1-w-\varpi
 \qquad\text{for almost every }t>t_0.
\]
Before $z$, $\psi=\rho_1$ and $b+s=1$; for almost every $t\in(t_0,z)$
this gives
\begin{equation}\label{eq:E1-pre-switch}
 G''=s^2\mathcal B+\mathcal P,\qquad \mathcal P=\frac wb-W-\varpi
                                  =\frac\lambda b-\varpi.
\end{equation}
On $(t_0,t_1)$, $\mathcal P''=2w_2/b+w'+\varpi'>0$,
$\mathcal P(t_0+)=-\varpi_0<0$, and $\mathcal P(t_1-)<0$.
For the latter sign use $w(t_1)<547/500$, $b>87/100$,
$W(t_1)>2/5$, and $\varpi(t_1-)>21/20$, yielding
$\mathcal P(t_1-)<63/50-2/5-21/20=-19/100$; the bound on $\varpi(t_1-)$
follows from $\varpi_0>3/50$ and $\varpi_1>64/25$.
Convexity therefore implies $\mathcal P\le0$ on the ramp.
Furthermore $\lambda(t_0)=0$ and, using $w'\ge w_1$ and
$w'+\varpi'\ge w_1+\varpi_1$ on the ramp,
\[
 \lambda'(t_0+)=w_1+b(b+\varpi_0)>\frac{127}{100},\qquad
 \lambda''=2w_2+b(w'+\varpi')>\frac{31}{10}.
\]
Since $t_1-t_0=97/250$, integration gives
\[
 \int_{t_0}^{t_1}\lambda>
 \frac{127}{100}\frac{(t_1-t_0)^2}2+
 \frac{31}{10}\frac{(t_1-t_0)^3}6
 =\frac{29478397}{234375000}>\frac18.
\]
Starting from $G(t_0)>0$, suppose $G$ first vanishes on $[t_0,t_1]$.
Before that point $\mathcal B=G+\mathcal I\ge0$, so $G''\ge\mathcal P$.
Since $\mathcal P\le0$, every preceding $t$ satisfies the bound below;
the equality uses $\int_{t_0}^{t_1}\mathcal P=\frac1b\int_{t_0}^{t_1}\lambda
-\int_{t_0}^{t_1}\varpi$ and $\int_{t_0}^{t_1}\varpi=\beta=bt_0$:
\begin{align*}
 G'(t)&\ge\mathfrak k_0+bt_0-\frac1{4b}+\int_{t_0}^{t_1}\mathcal P
 =\mathfrak k_0+\frac{\int_{t_0}^{t_1}\lambda-1/4}{b}
 >\mathfrak k_0-\frac1{8b}>0.
\end{align*}
This contradicts a first zero. Thus $G,G'>0$ on the ramp.
On $(t_1,z)$, $\varpi=0$ and $\mathcal P=\lambda/b\ge0$, so $G''\ge0$
as long as $G\ge0$. Thus $G,G'>0$ on $[t_1,z]$, using their
continuity at the endpoints. Also
\[
 \mathcal I(t_1)=\frac{13035083767568461}{151345285000000000}>\frac2{25},\qquad
 \mathcal I(t_0)<\frac{31}{500}<\frac2{25}.
\]
A concave function on $[t_0,t_1]$ is at least its smaller endpoint
value, so $\mathcal I(t)\ge \mathcal I(t_0)$ there; increasing $G$
then proves \eqref{eq:boundary-ramp-floor}.

We next show that $\mathcal B(1)>0$. Since
$\theta(t)=\theta(0)e^{st}+b\int_0^te^{s(t-r)}\psi(r)\,dr$, substituting
into \eqref{eq:budget} and exchanging the order of integration gives
\[
 \mathcal B(t)=\theta(0)\frac{\sinh(st)}s
 -\int_0^te^{-s(t-r)}\rho_1(r)\,dr
 +b\int_0^t\frac{\sinh(s(t-r))}s\psi(r)\,dr.
\]
Use $\theta(0)=\mathfrak k+b>41/40$, $\sinh x\ge x$,
$e^{-x}\le1-x+x^2/2$ for $x\ge0$, and
$\psi\ge\rho_1\mathbf1_{[0,3/4]}$.
Termwise, $\theta(0)\sinh(s)/s\ge\theta(0)>41/40$;
$\int_0^1e^{-s(1-r)}\rho_1\le\int_0^1\rho_+-s\int_0^1(1-r)\rho_-+\frac{s^2}2\int_0^1(1-r)^2\rho_+$;
and $b\int_0^1\frac{\sinh(s(1-r))}s\psi\ge b\int_0^{3/4}(1-r)\rho_-$.
The moment bounds \eqref{eq:boundary-moment-bounds} and $b=1-s$ imply
\[
 \mathcal B(1)>\frac{41}{40}-\frac{31}{20}
       +\frac{3s}{5}-\frac{9s^2}{50}+\frac{53b}{100}
 =\frac1{200}+\frac{7s}{100}-\frac{9s^2}{50}>\frac1{200}>0,
\]
where the last step uses $s<7/18$.
Finally, $\mathcal B'+s\mathcal B=\theta-\rho_1$ and
$\theta\le\Gamma$, $\rho_1\ge0$ give $(e^{st}\mathcal B)'\le\Gamma e^{st}$,
so $\mathcal B(t)\le\Gamma(1-e^{-st})/s\le\Gamma t\le\Gamma$. After $z$,
$\psi=0$, and $\varpi=0$ because $z>t_1$; the formula for $G''$ above and
$\rho_1\le c_1w$ give
\[
 G''=s^2\mathcal B+s\rho_1-w
 \le s^2\Gamma+(sc_1-1)w<\frac9{100}-\frac{11}{20}\frac45
 =-\frac7{20}<0,
\]
using $s<1/5$, $\Gamma,c_1<9/4$, and $w\ge b>4/5$.
Thus $G$ is concave on $[z,1]$ with positive endpoint values
$G(z)>0$ and $G(1)=\mathcal B(1)>0$.
This proves $G>0$ throughout. Together with $\mathcal I\ge0$ after $t_0$
and \eqref{eq:boundary-head-formula}, it proves
$\mathcal B\ge0$ and \eqref{eq:boundary-terminal}.
\end{proof}

\subsection{Bounds for constant profiles}\label{app:all-intervals}
Fix $q_0=157/200$, so that $1-q_0=43/200$ and $q_0-1/2=57/200$. The
next lemma collects four properties of the constant-profile integrands
$F_q$ in \eqref{eq:Fq-definition}. Section~\ref{app:complete} combines
them with the earlier lemmas.

\begin{lemma}[Constant-profile bounds]\label{lem:terminal-reduction}
Let $q\in[1/2,1]$.
\begin{enumerate}[label=(\roman*),leftmargin=*]
\item If $q\ge q_0$, then $F_q$ changes sign exactly once on $[0,1]$,
from negative to positive. Consequently, for $0\le\tau\le T\le1$,
\[
 \int_\tau^TF_q\le\max\Bigl\{0,\int_\tau^1F_q\Bigr\}.
\]
\item If $q\le q_0$, then $F_q\le7/80$ on $[0,1]$.
\item If $t_0\le\tau\le1$, then
$\int_\tau^1F_q\le\int_\tau^1F_1=\mathcal I(\tau)$, with $\mathcal I$ as in
\eqref{eq:boundary-terminal}.
\item If $t_1\le t\le1$, then $F_q(t)\le F_1(t)$ and $F_1(t)>0$.
\end{enumerate}
\end{lemma}
\begin{proof}
Besides \eqref{eq:scalar-brackets}, direct substitution gives
\begin{equation}\label{eq:terminal-mass-comparisons}
\begin{gathered}
 w(1/2)=\frac{162590879617}{161796000000}\in(1,101/100),\quad
 W(1/2)=\frac{84149166489889}{125553696000000}<\frac{27}{40},\\
 \frac{1093}{1000}<w(t_1)<\frac{547}{500},\quad
 W(t_1)<\frac{101}{250},\quad
 \varpi(t_1-)=\frac{12848017}{12125000}<\frac{1061}{1000},\quad
 w(1)>\frac{1137}{1000}.
\end{gathered}
\end{equation}
Also $w_1>233/500$,
$w_1-(b+\varpi_0)/2=-10610093/40449000000<0$, and $w_1+\varpi_1>3>4w_2$.

For $q\ge q_0$, \eqref{eq:Fq-after-t0} gives on the ramp
\[
 \frac{F_q'}{\cprof(q)}=w'-(1-q)(w+\varpi)
 >\frac{233}{500}-\frac{43}{200}\frac{431}{200}
 =\frac{107}{40000}>0,
\]
using $w'\ge w_1>233/500$ and
$w+\varpi\le w(t_1)+\varpi(t_1-)<547/500+1061/1000=431/200$ on the ramp.
On $[t_1,1]$, $F_{q_0}'/\cprof(q_0)=w_3-(1-q_0)w\le w_3-(1-q_0)b<0$,
while
\[
 F_{q_0}(1)>\cprof(q_0)\frac{1137}{1000}-\Gamma
             =\frac{196721081}{54440750000}>0.
\]
Furthermore
\[
 \lambda(t_1)-W(t_1)=w(t_1)-(b+1)W(t_1)
 >\frac{1093}{1000}-(b+1)\frac{101}{250}
 =\frac{21068437}{62500000}>0.
\]
Since $\partial_q\Lambda(q)=(1-2q)/b$, we have
$\partial F_q(t)/\partial q=[\lambda(t)-(2q-1)W(t)]/b\ge[\lambda(t)-W(t)]/b$,
which is positive at $t=t_1$ and equals $\lambda(1)/b\ge0$ at $t=1$. As
$F_{q_0}$ decreases on $[t_1,1]$ and $F_{q_0}(1)>0$, also
$F_{q_0}(t_1)>0$. Hence $F_q(t_1)>0$ and $F_q(1)>0$ for $q\ge q_0$. Since
$F_q''=-\cprof(q)(1-q)w_3\le0$ on $[t_1,1]$, it is positive throughout
that interval. On the ramp it increases from the negative value
$F_q(t_0)=-(q-1/2)^2/b$, and before $t_0$ it equals this negative
constant. Hence it changes sign exactly once, from negative to positive.
If $T$ lies before the sign change, then $\int_\tau^TF_q\le0$; otherwise
$F_q>0$ on $[T,1]$, so $\int_\tau^TF_q\le\int_\tau^1F_q$. Thus
\begin{equation}\label{eq:large-profile-endpoints}
 \int_\tau^TF_q\le\max\left\{0,\int_\tau^1F_q\right\}
 \qquad(q\ge q_0).
\end{equation}
This proves (i). With $q=1$, it also shows that $F_1>0$ on $[t_1,1]$.

For $1/2\le q\le q_0$ we claim the uniform pointwise bound
\begin{equation}\label{eq:small-profile-bound}
 F_q(t)\le\frac7{80}\qquad(0\le t\le1).
\end{equation}
First $F_{1/2}\le0$: it vanishes before $t_0$, has negative right
derivative at $t_0$ because $w_1-(b+\varpi_0)/2<0$, and on the ramp
$F_{1/2}''/\cprof(1/2)=2w_2-(w_1+\varpi_1+2w_2x)/2<0$.
After $t_1$, $F_{1/2}'=\cprof(1/2)(w_3-w/2)<0$.
The preceding monotonicity of $F_{q_0}$ gives
\[
 F_{q_0}(t)\le F_{q_0}(t_1)
 <\cprof(q_0)\left(\frac{547}{500}+(1-q_0)\frac{101}{250}\right)-\Gamma
 =\frac{2368186813}{27220375000}<\frac7{80}.
\]
For $t\ge1/2$, the increasing ratio $\lambda/W$ and
\[
 \lambda(1/2)-\frac{57}{100}W(1/2)
 >1-\left(b+\frac{57}{100}\right)\frac{27}{40}
 =\frac{272899}{10^7}>0
\]
imply $\lambda(t)\ge\frac{57}{100}W(t)$, so
$\partial F_q/\partial q=[\lambda-(2q-1)W]/b\ge0$ for $q\le q_0$, because
$2q-1\le57/100$. Hence $F_q(t)\le F_{q_0}(t)\le F_{q_0}(t_1)<7/80$ for
$t\ge1/2$, since $F_{q_0}$ increases on the ramp and decreases after
$t_1$. For $t\le1/2$, put $\xi=(q-1/2)/(q_0-1/2)\in[0,1]$. Since $F_q$
is quadratic in $q$ with leading coefficient $-W/b$, interpolation
between $q=1/2$ and $q=q_0$ gives the exact identity
\[
 F_q=(1-\xi)F_{1/2}+\xi F_{q_0}
                   +\frac Wb(q-1/2)(q_0-q).
\]
Here $F_{1/2}\le0$. For $t_0\le t\le1/2$,
$\xi F_{q_0}(t)\le\max\{0,F_{q_0}(1/2)\}$ because $F_{q_0}$ increases on
the ramp. Also $W\le1$ and $(q-1/2)(q_0-q)\le(57/400)^2$. Hence, using
\eqref{eq:terminal-mass-comparisons},
\[
 F_q(t)<\cprof(q_0)\left(\frac{101}{100}+(1-q_0)\frac{27}{40}\right)
       -\Gamma+\frac1b\left(\frac{57}{400}\right)^2
 =\frac{1909633207}{31109000000}<\frac1{16}<\frac7{80}.
\]
Before $t_0$, $F_q\le0$ directly. This proves
\eqref{eq:small-profile-bound}, which is (ii).

We next compare the integrals $\int_\tau^1F_q$ for $\tau\ge t_0$. Put
$\mathcal W_\tau=\int_\tau^1W$ and $\Omega_\tau=\int_\tau^1\lambda$.
The mass identities and $J<5/8$ give
\[
 \Omega_{t_0}-\mathcal W_{t_0}=1+t_0-(1+b)J
 >\frac54-(1+b)\frac58=\frac58(1-b)>0.
\]
Since $\lambda/W$ increases, its $W$-weighted terminal average
increases as $\tau$ increases. Therefore $\Omega_\tau\ge\mathcal W_\tau$ for
$\tau\ge t_0$, and
\begin{equation}\label{eq:terminal-profile-one}
 \int_\tau^1(F_q-F_1)
 =\frac{1-q}{b}(q\mathcal W_\tau-\Omega_\tau)\le0.
\end{equation}
Here $\Lambda(q)-\Lambda(1)=q(1-q)/b$ and $\cprof(q)-\cprof(1)=-(1-q)/b$,
and the sign follows from $q\mathcal W_\tau\le\mathcal W_\tau\le\Omega_\tau$.
After $t_0$, $F_1=\cprof(1)w-\Gamma=c_1w-\Gamma$, so
$\int_\tau^1F_1=\mathcal I(\tau)$ for $\tau\ge t_0$. This proves (iii).

Finally, let $t_1\le t<1$. Since $\lambda/W$ increases and
$\lambda(t_1)>W(t_1)$, we have $\lambda(t)>W(t)$. Hence
\[
 \frac{\partial F_q(t)}{\partial q}
 =\frac{\lambda(t)-(2q-1)W(t)}b
 \ge\frac{\lambda(t)-W(t)}b>0,
\]
so $F_q(t)\le F_1(t)$. At $t=1$, $W(1)=0$ and the derivative equals
$\lambda(1)/b\ge0$. Together with $F_1>0$ on $[t_1,1]$ from the proof of
(i), this proves (iv).
\end{proof}

\subsection{Completing the proof of Lemma~\ref{lem:scalar}}\label{app:complete}
Lemma~\ref{lem:scalar} has three parts: the mass identities, the
properties of $\theta_\pm$, and the scalar
inequality~\eqref{eq:scalar}. The preceding sections prove each of
them.
\begin{proof}[Proof of Lemma~\ref{lem:scalar}]
Section~\ref{app:parameters} verifies \eqref{eq:mass-identities} for the
weights in \eqref{eq:beta-c}--\eqref{eq:mixing}; this is the first part
of the lemma.

For the second part, Lemma~\ref{lem:boundary-estimates} gives a unique
$z\in(0,1)$ with $\int_0^ze^{-st}\rho_1(t)\,dt=1$, shows that $z>3/4$,
and proves \eqref{eq:theta-identities}. The two functions differ by
\[
 \theta_+(t)-\theta_-(t)
 =be^{st}\Bigl[\int_0^t\rho_1(r)e^{-sr}\,dr-1\Bigr].
\]
The bracket is strictly increasing in $t$, because $\rho_1>0$ almost
everywhere, and it vanishes at $t=z$. Hence $\theta_+<\theta_-$ before
$z$ and $\theta_+>\theta_-$ after $z$, so they cross exactly
once, at $z$. Both are absolutely continuous, and therefore
so is their minimum $\theta$.

It remains to prove \eqref{eq:scalar}. Fix $0\le\tau\le T\le1$ and a
nonincreasing $q:[\tau,T]\to[1/2,1]$. If $\tau=T$, the left side of
\eqref{eq:scalar} is zero, and $\mathcal B(\tau)\ge0$ by
Lemma~\ref{lem:boundary-estimates}. Let $\tau<T$.

We first reduce to a constant profile. Let $\overline q$ be the
$W$-weighted average of $q$ on $[\tau,T]$, as in
Lemma~\ref{lem:ordered-scalar-reduction}; it lies in $[1/2,1]$. With
$\mathcal U_{\tau,T}=\int_\tau^TW$ and $L_{\tau,T}=\int_\tau^T\lambda$, that lemma and the
identity $\Lambda(\overline q)=\cprof(\overline q)(b+1-\overline q)$ give
\[
 \int_\tau^T\bigl\{\Lambda(q(t))W(t)+\cprof(q(t))\lambda(t)-\Gamma\bigr\}dt
 \le\Lambda(\overline q)\,\mathcal U_{\tau,T}+\cprof(\overline q)L_{\tau,T}-\Gamma(T-\tau)
 =\int_\tau^TF_{\overline q}(t)\,dt .
\]
It therefore suffices to prove
\begin{equation}\label{eq:constant-profile-goal}
 \int_\tau^TF_q\le\mathcal B(\tau)
 \qquad\text{for every constant }q\in[1/2,1].
\end{equation}
Below we use $\int_\tau^TF_q=P_q(T)-P_q(\tau)$ and consider five cases
according to $\tau$ and $q$.

Suppose first that $\tau\le t_0$ and $q\le q_0$. By
Lemma~\ref{lem:prefix-inequality}, $P_q(T)\le0$. Since $\tau\le t_0$,
$P_q(\tau)=-\tau(q-1/2)^2/b$. As $(q-1/2)^2\le(57/200)^2<b\mathfrak k_0$,
the lower bound in \eqref{eq:boundary-head-formula} gives
\[
 \int_\tau^TF_q\le\frac{(q-1/2)^2}{b}\tau\le\mathfrak k_0\tau
 \le\mathcal B(\tau).
\]

Next suppose that $\tau\le t_0$ and $q\ge q_0$. By
Lemma~\ref{lem:terminal-reduction}(i), Lemma~\ref{lem:head-terminal},
and $\mathcal B\ge0$,
\[
 \int_\tau^TF_q\le\max\{0,P_q(1)-P_q(\tau)\}\le\max\{0,\mathcal B(\tau)\}
 =\mathcal B(\tau).
\]

Now suppose that $t_0\le\tau\le t_1$ and $q\le q_0$. By
Lemma~\ref{lem:terminal-reduction}(ii) and $T-\tau\le1-t_0=3/4$,
followed by \eqref{eq:boundary-ramp-floor},
\[
 \int_\tau^TF_q\le\frac7{80}\cdot\frac34=\frac{21}{320}<\mathcal B(\tau).
\]

If $\tau\ge t_0$ and $q\ge q_0$, then by
Lemma~\ref{lem:terminal-reduction}(i) and~(iii), followed by
$\mathcal B\ge0$ and \eqref{eq:boundary-terminal},
\[
 \int_\tau^TF_q\le\max\Bigl\{0,\int_\tau^1F_q\Bigr\}
 \le\max\{0,\mathcal I(\tau)\}\le\mathcal B(\tau).
\]

Finally, if $\tau\ge t_1$, then by Lemma~\ref{lem:terminal-reduction}(iv),
$F_q\le F_1$ and $F_1>0$ on $[\tau,1]$. Hence, by
\eqref{eq:boundary-terminal},
\[
 \int_\tau^TF_q\le\int_\tau^TF_1\le\int_\tau^1F_1=\mathcal I(\tau)
 \le\mathcal B(\tau).
\]

These cases are exhaustive: the first two cover $\tau\le t_0$, the
third and fourth cover $t_0\le\tau\le t_1$, and the last covers
$\tau\ge t_1$. This
proves \eqref{eq:constant-profile-goal}, hence \eqref{eq:scalar}, and
completes the proof of Lemma~\ref{lem:scalar}.
\end{proof}

\section{Sharp geometric recovery}\label{sec:geometric}
This appendix determines the largest coefficient attainable by a
deterministic downward recovery map that depends only on the state.
The proof also gives bounded affine comparison vectors, including at
boundary points of the cube, by taking supporting slopes (supergradients) of the concave
closure. This is an existence statement: computing a suitable vector
from polynomially many objective queries is a separate requirement.
Sections~\ref{sec:structural} and~\ref{sec:online} meet that requirement
at coefficient $\alpha_0$ by constructing a field linear in the
objective and a conditionally unbiased one-query estimator.

\subsection{Recovery relative to the concave closure}
The relevant comparison value allows arbitrary distributions of
comparators with a prescribed mean. It therefore describes the
simultaneous requirement imposed by all comparators at a given state.

For a continuous function $f:[0,1]^n\to\R$, define its concave closure
on the full cube by
\begin{equation}
 f^+(x)=\sup\{\E f(O): O\in[0,1]^n\text{ almost surely},\ \E O=x\}.
 \label{gcon:envelope}
\end{equation}
Thus $f^+(x)$ compares all distributions with mean $x$; their support
may lie anywhere in the cube. The supremum is attained by a distribution
with at most $n+1$ support points, as proved in
Section~\ref{gcon:finite-normalization}.
Set
\begin{equation}
\begin{aligned}
 b_\star&=\frac{\sqrt3}{2},& s_\star&=1-b_\star,\\
 \alpha_\star&=\frac{b_\star}{1+b_\star}=2\sqrt3-3,&
 \Gamma_\star&=\alpha_\star^{-1}
       =1+b_\star+\frac{1}{4b_\star}.
\end{aligned}
\label{gcon:constant}
\end{equation}
These constants define the coordinatewise map
\begin{equation}
 (h_\star(x))_i=
 \begin{cases}
  1-e^{-x_i},&x_i\le s_\star,\\
  x_i/(b_\star+x_i),&x_i>s_\star.
 \end{cases}
 \label{gcon:map}
\end{equation}
At the threshold, the exponential branch specifies the value.
The map has an upward right-hand jump there. The following theorem
relates its value to the concave closure and determines its optimal
universal coefficient.

\begin{theorem}[Sharp recovery and bounded affine comparisons]
\label{gcon:main-existence}
Let $f:[0,1]^n\to\R_{\ge0}$ be $C^1$ and DR-submodular.
For every $x\in[0,1]^n$ the following statements hold.
\begin{enumerate}[label=(\roman*),leftmargin=*]
\item The map \eqref{gcon:map} satisfies
\begin{equation}
 0\le h_\star(x)\le x,\qquad
 f(h_\star(x))\ge\alpha_\star f^+(x).
 \label{gcon:main}
\end{equation}
In particular, $h_\star(x)\in K$ whenever $x$ belongs to a
down-closed set $K$.
\item There is a vector
$g_x\in[\nabla f(\one),\nabla f(\zero)]$, where the interval is
coordinatewise, such that for every $o\in[0,1]^n$,
\begin{equation}
 f(o)\le f^+(x)+\ip{g_x}{o-x}
 \le\Gamma_\star f(h_\star(x))+\ip{g_x}{o-x}.
 \label{gcon:field}
\end{equation}
Consequently $\norm{g_x}\le H(f)$, with $H(f)$ defined in
\eqref{eq:gradient-envelope}.
\item The coefficient $\alpha_\star$ is optimal uniformly over
dimensions among deterministic maps $T_n:[0,1]^n\to[0,1]^n$ that
are independent of $f$ and satisfy $T_n(x)\le x$. More precisely,
if
\[
 f(T_n(x))\ge c f^+(x)
\]
for every dimension $n$, every such $f$, and every $x$, then
$c\le\alpha_\star$. This includes maps whose coordinates depend
on the entire vector $x$ and maps that are discontinuous.
\end{enumerate}
\end{theorem}

The coefficient arises from the balance
$1+t+1/(4t)=(1+t)/t$, whose positive solution is $t=b_\star$;
four multilinear objectives show that no larger universal coefficient
is possible for these recovery maps. Averaging \eqref{gcon:field}
over a mean-$x$ comparator law gives \eqref{gcon:main}.
The reverse implication follows from a supporting hyperplane of $f^+$.

The exact intercept $f^+(x)$ also prevents an objective-linear choice
of these supporting vectors in general. For a direct example, fix
$x=(2/5,2/5,2/5)$ and let
$f_{ij}(z)=z_i+z_j-2z_iz_j$ for $1\le i<j\le3$.
Its concave closure is
$f_{ij}^+(z)=\min\{z_i+z_j,2-z_i-z_j\}$, so the unique supporting
slope at $x$ is $e_i+e_j$. In contrast,
$f_\triangle=f_{12}+f_{13}+f_{23}$ has
$f_\triangle^+(z)=2$ whenever
$z\in[0,1]^3$ and $1\le z_1+z_2+z_3\le2$:
its Boolean values are two on nonempty proper subsets, whose indicators
have precisely this convex hull, and zero on the other two subsets.
Multilinearity bounds $f_\triangle$ by two throughout the cube.
Since $x$ is interior to this constant region, its supporting slope
is uniquely zero. Writing $g_x(f)$ for any selection of the exact
supporting slope, these identities give
\begin{equation}\label{gcon:nonlinear-support-slopes}
 \sum_{1\le i<j\le3}g_x(f_{ij})=2\one,
 \qquad g_x(f_{12}+f_{13}+f_{23})=\zero.
\end{equation}
A supporting vector in the first inequality of~\eqref{gcon:field}
is also a supergradient of $f^+$, by averaging over comparator laws.
Consequently no choice satisfying that exact-intercept inequality
can be linear in $f$ for the whole class. The signed linear sampler
of Lemma~\ref{lem:one-query} therefore cannot implement these
particular supporting vectors. The relaxed intercept $\Gamma_\star f(h_\star(x))$ allows additional
choices of comparison vectors. Constructing such a field with a
polynomial-query oracle at coefficient $\alpha_\star$ remains a
separate algorithmic question.

For a nonnegative submodular set function
$\mathfrak f:2^{[n]}\to\R_{\ge0}$, let $\widetilde F$ denote its
multilinear extension. The Boolean concave closure of $\mathfrak f$ at
$x$ is the largest expected value of $\mathfrak f$ over distributions of
sets with elementwise marginals $x$. Independently rounding each
coordinate of a cube-valued comparator to zero or one preserves both its
mean and its expected value under $\widetilde F$; hence
$\widetilde F^+$ agrees with the Boolean concave closure of $\mathfrak f$.
\citet[Theorem~2]{ChekuriLivanos2024} establish the existence of a
point $y\le x$ satisfying $\widetilde F(y)\ge e^{-1}\widetilde F^+(x)$,
with stronger bounds in some parameter regimes.
Theorem~\ref{gcon:main-existence} provides the explicit
objective-independent choice $h_\star(x)$, achieves $\alpha_\star$, and
applies to all nonnegative $C^1$ DR-submodular functions. The ordinary
ratio $\widetilde F(x)/\widetilde F^+(x)$ can still approach zero: for
$\widetilde F(z)=z_1(1-z_2)$ and $x=(\epsilon,1-\epsilon)$,
$0<\epsilon<1$, one has $\widetilde F(x)=\epsilon^2$ and
$\widetilde F^+(x)=\epsilon$. The latter value is attained by assigning
probabilities $\epsilon$ and $1-\epsilon$ to $(1,0)$ and $(0,1)$,
respectively.

\subsection{Proof of Theorem~\ref{gcon:main-existence}}
\label{gcon:sec}
The recovery inequality follows by separating the contributions of the
small and large coordinates and comparing both with a common finite
distribution. A supporting-hyperplane argument then gives the affine
comparison, and four multilinear functions prove sharpness. Fix a
nonnegative $C^1$ DR-submodular function $f$ and $x\in[0,1]^n$.

\subsubsection{A finite comparator law and affine normalization}
\label{gcon:finite-normalization}
An attaining law with finite support allows the recovery inequality to
be proved by finite weighted sums. The graph
$\mathcal C_f=\{(o,f(o)):o\in[0,1]^n\}\subset\R^{n+1}$
is compact. Its convex hull is compact by Carath\'eodory's theorem
\citep[Theorems~17.1--17.2]{rockafellar1970}.
The attainable pairs $(\E O,\E f(O))$ form this convex hull.
Finite laws give its convex combinations. For an arbitrary law, suppose
its expectation pair lay outside the compact convex hull. Strict
separation would provide an affine functional $\Phi$ with
$\Phi(\E O,\E f(O))>\max_{(o,f(o))\in\mathcal C_f}\Phi(o,f(o))$.
But the affinity of $\Phi$ gives
$\Phi(\E O,\E f(O))=\E\Phi(O,f(O))$, and every integrand is at most
the displayed maximum, a contradiction. The points of this hull whose
first $n$ coordinates equal $x$ therefore form a nonempty compact set,
and its largest last coordinate is $f^+(x)$.
The function form of Carath\'eodory's theorem
\citep[Corollary~17.1.5]{rockafellar1970}, applied to $-f$ on the
cube and extended by $+\infty$ outside it, expresses the same optimum
using at most $n+1$ support points. Removing zero weights gives
\begin{equation}
 \begin{gathered}
 1\le N\le n+1,\qquad \pi_k>0,\qquad
 \sum_{k=1}^N\pi_k=1,\qquad
 \sum_{k=1}^N\pi_k o^k=x,\\
 \sum_{k=1}^N\pi_k f(o^k)=f^+(x),\qquad o^k\in[0,1]^n.
 \end{gathered}
 \label{gcon:law}
\end{equation}

The recovery map partitions the coordinates into
$S=\{i:x_i\le s_\star\}$ and $L=[n]\setminus S$.
Write $x=(u,v)$, $h_\star(x)=(a,p)$, and
$o^k=(o_S^k,o_L^k)$ in these two coordinate blocks. Thus
$a_i=1-e^{-u_i}$ for $i\in S$ and
$p_j=v_j/(b_\star+v_j)$ for $j\in L$.
Define
\begin{equation}
 r=\frac{h_\star(x)-\alpha_\star x}{1-\alpha_\star}
   =(1+b_\star)h_\star(x)-b_\star x.
 \label{gcon:residual}
\end{equation}
The two branches give $\alpha_\star x\le h_\star(x)\le x$.
On the small-coordinate block, the lower bound follows from
$1-e^{-t}\ge t-t^2/2\ge t/2\ge\alpha_\star t$
for $0\le t\le s_\star<1/2$, and $1-e^{-t}\le t$ gives the upper bound. On the large-coordinate block,
$b_\star+t\ge1$ and $(b_\star+t)^{-1}\ge(1+b_\star)^{-1}
\ge\alpha_\star$.
Thus $0\le r\le x$. On the small block, $r_i<u_i$ whenever $u_i>0$;
at $u_i=0$ one has $r_i=0$. Consequently,
\begin{equation}
 r_i<s_\star\quad(i\in S),\qquad
 r_j-s_\star=\frac{b_\star(v_j-s_\star)(1-v_j)}{b_\star+v_j}\ge0
       \quad(j\in L).
 \label{gcon:ordering}
\end{equation}
For the second relation, $r_j=(1+b_\star)p_j-b_\star v_j$ and, over the
denominator $b_\star+v_j$,
$(1+b_\star)v_j-b_\star v_j(b_\star+v_j)-(1-b_\star)(b_\star+v_j)
=b_\star(v_j-s_\star)(1-v_j)$.
The coordinates are ordered by decreasing $r_i$, with ties broken by
original index, so that the large block precedes the small block.
For the following formulas, the coordinates are
relabeled in this order. Thus $i=1,\ldots,n$ is both the order position
and the relabeled coordinate index, $C_i=\{1,\ldots,i\}$, and
$C_0=\varnothing$. Subscripts $<i$ and $>i$ restrict a vector to
coordinates before and after $i$ in this order. Put
\begin{equation}
 d_i=f(\one_{C_i})-f(\one_{C_{i-1}})=\int_0^1\partial_i f(\one_{<i},t,0_{>i})\,dt,\qquad
 \ell(y)=f(0)+\sum_{i=1}^n d_i y_i.
 \label{gcon:minorant}
\end{equation}
The greedy lower bound of \citet[Proposition~3.2]{bach2013},
extended from the vertices by coordinatewise concavity as in
Lemma~\ref{lem:normalization}, gives $f(y)\ge\ell(y)$ throughout
the cube. Hence $F=f-\ell$ is nonnegative and DR-submodular.
The vertices $\one_{C_i}$ satisfy
$F(0)=F(0_S,\one_L)=0$ by telescoping.
Set $r_{n+1}=0$. Summation by parts gives
\begin{equation}
 \ell(r)=(1-r_1)f(0)+\sum_{i=1}^n(r_i-r_{i+1})f(\one_{C_i})\ge0.
 \label{gcon:minorant-payment}
\end{equation}
The coefficients in this display are nonnegative, because $r$ is
sorted nonincreasingly and $r_1\le x_1\le1$, and they sum to one.

\subsubsection{The large-coordinate contribution}
Since $F(0_S,\one_L)=0$, subtracting
$e=F(\cdot,\one_L)$ leaves a nonnegative remainder. Define
\begin{equation}
 e(y)=F(y,\one_L),\quad
 R(y,l)=F(y,l)-e(y),\quad R_0(l)=R(0,l),\quad
 \Delta_y(l)=R(y,l)-R_0(l).
 \label{gcon:slices}
\end{equation}
Then $e\ge0$ is DR and $e(0)=0$. For each fixed $y$,
$R(y,\cdot)$ is DR and vanishes at $\one_L$. It is nonnegative:
$R(0,l)=F(0,l)\ge0$, and for $i\in S$,
\begin{equation}
 \partial_iR(y,l)=\partial_iF(y,l)
                  -\partial_iF(y,\one_L)\ge0.
 \label{gcon:residual-increasing}
\end{equation}
Further, $\Delta_y(\one_L)=0$ and
$\partial_j\Delta_y(l)=\partial_jF(y,l)-\partial_jF(0,l)$ is
nonpositive for $j\in L$ and nonincreasing in $y$.
If $S$ is empty, $e=0$; if $L$ is
empty, $R=R_0=0$. The terms involving an empty block then vanish.

The large-coordinate block is ordered increasingly by
$\mu_j=v_j(1-v_j)$, with ties broken by the original coordinate index. This order is independent of the order used to
define $\ell$. In the large-coordinate chain formulas below,
$j=1,\ldots,|L|$ denotes position in this order, and vector components retain their corresponding coordinate
values. The sums in the following comparisons are zero when $L$ is empty.
For a $C^1$ function $V$ on $[0,1]^L$ set
\begin{align}
 A_j^L(V)&=\frac1{p_j}\int_0^{p_j}
     \partial_jV(p_{<j},t,0_{>j})\,dt,\nonumber\\
 B_j^L(V)&=\frac1{1-p_j}\int_{p_j}^1
     \partial_jV(p_{<j},t,\one_{>j})\,dt.
 \label{gcon:secant-def}
\end{align}
Here $p_j,1-p_j>0$ even at boundary states of the large-coordinate block.
For DR-submodular $V$, the signed endpoint comparison
\eqref{eq:signed-endpoint}, with the secants in
\eqref{gcon:secant-def}, gives the following bound. Its derivation
uses only DR-submodularity of the function and the prefix structure of
the secants, not a particular order of the large coordinates, so it
applies with the order by $\mu_j$:
\begin{equation}
 V(l)\le V(p)+\sum_{j=1}^{|L|}
       \big[l_j(1-p_j)A_j^L(V)-p_j(1-l_j)B_j^L(V)\big].
 \label{gcon:secant-majorant}
\end{equation}

Define the constant
\begin{equation}
 C_{R_0}=R_0(p)+\sum_{j=1}^{|L|}
   \big[v_j(1-p_j)A_j^L(R_0)-p_j(1-v_j)B_j^L(R_0)\big].
 \label{gcon:baseline}
\end{equation}
If $L$ is empty, then $C_{R_0}=R_0(p)=0$, so
\eqref{gcon:baseline-payment} below is immediate. For the following
telescoping calculation, assume $|L|\ge1$.
Forward telescoping gives
$\sum_{j=1}^{|L|}p_jA_j^L(R_0)=R_0(p)-R_0(0)\le R_0(p)$.
For reverse telescoping let
$Q_j=R_0(p_{\le j},\one_{>j})$, so that
$Q_0=R_0(\one_L)=0$, $Q_{|L|}=R_0(p)$, and
$(1-p_j)B_j^L(R_0)=Q_{j-1}-Q_j$. Thus
\begin{align}
 \sum_{j=1}^{|L|}\mu_j(1-p_j)B_j^L(R_0)
 &=\sum_{j=1}^{|L|-1}(\mu_{j+1}-\mu_j)Q_j
      -\mu_{|L|}R_0(p)\nonumber\\
 &\ge-\tfrac14 R_0(p).
 \label{gcon:baseline-telescope}
\end{align}
The inequality follows from $Q_j\ge0$, the ordering of the $\mu_j$,
and $0\le\mu_j\le1/4$.
The identities $v_j(1-p_j)=b_\star p_j$ and
$p_j(1-v_j)=\mu_j(1-p_j)/b_\star$ both follow from
$p_j/(1-p_j)=v_j/b_\star$. Substituting them in \eqref{gcon:baseline}
and using the two telescoping bounds gives
\begin{equation}
 C_{R_0}=R_0(p)+b_\star\sum_jp_jA_j^L(R_0)-\frac1{b_\star}\sum_j\mu_j(1-p_j)B_j^L(R_0)
 \le\Bigl(1+b_\star+\frac1{4b_\star}\Bigr)R_0(p)=\alpha_\star^{-1}R_0(p).
 \label{gcon:baseline-payment}
\end{equation}
The last equality holds because $b_\star^2=3/4$ gives
$b_\star+1/(4b_\star)=1/b_\star$, so the coefficient is
$(1+b_\star)/b_\star=\alpha_\star^{-1}$.

For the remaining contribution, define
\begin{equation}
 Z_j(y)=-B_j^L(\Delta_y)\ge0.
 \label{gcon:stocks}
\end{equation}
They vanish at $y=0$ and are nondecreasing coordinatewise in $y$, by the
derivative properties established after \eqref{gcon:slices}.
Exact reverse telescoping, using $\Delta_y(\one_L)=0$, gives
\begin{equation}
 \Delta_y(p)=\sum_{j=1}^{|L|}(1-p_j)Z_j(y).
 \label{gcon:stock-identity}
\end{equation}
The comparator average of the residual function is the scalar
\begin{equation}\label{gcon:residual-average}
 \mathcal V_{\mathrm{res}}=\sum_{k=1}^N\pi_k R(o_S^k,o_L^k).
\end{equation}
For each $k$, \eqref{gcon:secant-majorant} with
$V=R(o_S^k,\cdot)=R_0+\Delta_{o_S^k}$ and $l=o_L^k$ gives
\[
\begin{aligned}
 R(o^k)\le{}&R_0(p)+\Delta_{o_S^k}(p)\\
 &+\sum_j\Bigl[(o_L^k)_j(1-p_j)A_j^L(R_0+\Delta_{o_S^k})
 -p_j\bigl(1-(o_L^k)_j\bigr)B_j^L(R_0+\Delta_{o_S^k})\Bigr].
\end{aligned}
\]
The secants are linear in the function. Averaged over $k$ with weights
$\pi_k$, the $R_0$-terms give $C_{R_0}$, because $\sum_k\pi_ko_L^k=v$. In the
$\Delta$-terms, $A_j^L(\Delta_{o_S^k})\le0$, $-B_j^L(\Delta_{o_S^k})=Z_j(o_S^k)$,
and $\Delta_{o_S^k}(p)=\sum_j(1-p_j)Z_j(o_S^k)$ by
\eqref{gcon:stock-identity}. Dropping the nonpositive forward terms
therefore yields
\begin{equation}
 \mathcal V_{\mathrm{res}}\le C_{R_0}+\sum_{j=1}^{|L|}\left[
 (1-p_j)\sum_{k=1}^N\pi_k Z_j(o_S^k)
 +p_j\sum_{k=1}^N\pi_k\left(1-(o_L^k)_j\right)Z_j(o_S^k)\right].
 \label{gcon:stock-cost}
\end{equation}

\subsubsection{A common average for the two coordinate blocks}
We call the index $k$ of the support point $o^k$ in \eqref{gcon:law}
its label. Independent random choices of labels produce a distribution
whose mean is the small-coordinate recovery point $a$. This distribution
controls both the terms $Z_j$ and the function $e$.

For each label $k$, define an inclusion probability
$\theta_k$ and a vector $t^k$ by
\begin{equation}
 \theta_k=1-e^{-\pi_k},\qquad
 t_i^k=\frac{1-e^{-\pi_k(o_S^k)_i}}{\theta_k}
 \quad(i\in S).
 \label{gcon:inflation}
\end{equation}
Concavity of $\omega\mapsto1-e^{-\pi_k\omega}$ implies
$o_S^k\le t^k\le\one_S$.
For $A\subseteq[N]$ put
\begin{equation}
 \nu_A=\prod_{k\in A}\theta_k
          \prod_{k\in[N]\setminus A}(1-\theta_k),\qquad
 y_A=\one_S-\prod_{k\in A}(\one_S-t^k),
 \label{gcon:finite-law}
\end{equation}
with coordinatewise products and the usual empty-product convention.
The weights describe independent label choices. They satisfy
\begin{equation}
 \begin{aligned}
 &\sum_{A\subseteq[N]}\nu_A=1,\\
 &\sum_{A\subseteq[N]}\nu_Ay_{A,i}
 =1-\sum_{A\subseteq[N]}\nu_A\prod_{k\in A}(1-t_i^k)
 =1-\prod_{k=1}^N(1-\theta_k t_i^k)
 =1-e^{-\sum_{k=1}^N\pi_k(o_S^k)_i}=a_i.
 \end{aligned}
 \label{gcon:mean}
\end{equation}
The second equality uses the independence of the label choices, and the
third uses $\theta_kt_i^k=1-e^{-\pi_k(o_S^k)_i}$.

The distribution just constructed satisfies the following comparison
for every $C^1$ DR function $V$ on $[0,1]^S$:
\begin{equation}
 V(a)\ge\sum_{A\subseteq[N]}\nu_AV(y_A).
 \label{gcon:jensen}
\end{equation}
The law \eqref{gcon:finite-law} consists of independent Bernoulli
label choices, each initially shared by all coordinates. Let
$c\mathbin\oplus t=c+(\one_S-c)\odot t$, where $\odot$ denotes
coordinatewise multiplication. Conditional on all other choices,
splitting one label's shared choice into independent choices on
disjoint coordinate blocks $I,J$ changes the expectation by
the following quantity. Here $\theta$ is that label's inclusion
probability, $c$ is the contribution from the other choices, and
$t_I,t_J$ are the restrictions of its vector $t^k$ to $I$ and $J$,
with zeros elsewhere:
\begin{equation}
 \theta(1-\theta)\big[
 V(c\oplus t_I)+V(c\oplus t_J)-V(c)
       -V(c\oplus t_I\oplus t_J)\big]\ge0.
 \label{gcon:rectangle}
\end{equation}
The inequality holds by DR-submodularity. Write
$\eta_I=(\one_S-c)\odot t_I\ge0$. Since $I$ and $J$ are disjoint,
$c\oplus t_I=c+\eta_I$ and $c\oplus t_I\oplus t_J=c\oplus t_J+\eta_I$. As
$\nabla V$ is coordinatewise nonincreasing and $c\oplus t_J\ge c$,
\[
 V(c+\eta_I)-V(c)=\int_0^1\ip{\nabla V(c+r\eta_I)}{\eta_I}\,dr
 \ge\int_0^1\ip{\nabla V(c\oplus t_J+r\eta_I)}{\eta_I}\,dr
 =V(c\oplus t_I\oplus t_J)-V(c\oplus t_J).
\] The comparison now follows in two steps. First, split every
label's shared choice coordinate by coordinate, one label and one
coordinate at a time. By \eqref{gcon:rectangle}, no split decreases the
expectation, and each coordinate's marginal law, hence its mean $a_i$,
is unchanged. After all splits, the coordinates of the random point are
independent. Second, for this product law, Jensen's inequality applied
to one coordinate at a time, conditionally on the others, gives
expected value at most $V(a)$, since $V$ is concave in every
coordinate. The original expectation is at most the product-law
expectation, which proves \eqref{gcon:jensen}.
For $V(y)=F(y,p)$, this gives
$F(a,p)\ge\sum_{A\subseteq[N]}\nu_AF(y_A,p)$. The comparison is applied to
$F(\cdot,p)$ itself: the residual $R(\cdot,p)$ need not be
DR-submodular in the small coordinates, and its average is bounded
separately below, using only the monotonicity properties established
after \eqref{gcon:slices}.

For the functions $Z_j$, monotonicity gives the following
threshold comparison. If
$Z_j(o_S^k)>\tau$ and $k\in A$, then
$y_A\ge t^k\ge o_S^k$ and hence $Z_j(y_A)>\tau$.
Writing
$m_j(\tau)=\sum_{k=1}^N\pi_k\mathbf1\{Z_j(o_S^k)>\tau\}$, the
probability that no label $k$ with $Z_j(o_S^k)>\tau$ is selected is
$\prod_{k:\,Z_j(o_S^k)>\tau}(1-\theta_k)=e^{-m_j(\tau)}$. Hence
\begin{equation}
 \sum_{\substack{A\subseteq[N]\\Z_j(y_A)>\tau}}\nu_A\ge1-e^{-m_j(\tau)}.
 \label{gcon:threshold}
\end{equation}
Layer-cake integration gives
\begin{align}
 \sum_{k=1}^N\pi_k Z_j(o_S^k)&=\int_0^\infty m_j(\tau)\,d\tau,
 \nonumber\\
 \sum_{k=1}^N\pi_k\left(1-(o_L^k)_j\right)Z_j(o_S^k)
 &\le\int_0^\infty\min\{1-v_j,m_j(\tau)\}\,d\tau.
 \label{gcon:layer-cake}
\end{align}
For the second inequality, each threshold set has total mass $m_j(\tau)$, and its value of
$\sum_k\pi_k\bigl(1-(o_L^k)_j\bigr)$ is at most
$\sum_{k=1}^N\pi_k\left(1-(o_L^k)_j\right)=1-v_j$.

The two integrals are controlled by
\begin{equation}
 (1-p_j)(1-e^{-\omega})\ge\alpha_\star
 \left[(1-p_j)\omega+p_j\min\{1-v_j,\omega\}\right]
 \qquad(0\le\omega\le1).
 \label{gcon:scalar}
\end{equation}
Multiplying \eqref{gcon:scalar} by
$(b_\star+v_j)/\alpha_\star=(b_\star+v_j)(1+b_\star)/b_\star$ and using
$1-p_j=b_\star/(b_\star+v_j)$ and $p_j=v_j/(b_\star+v_j)$ turns it into
$(1+b_\star)(1-e^{-\omega})\ge b_\star\omega+v_j\min\{1-v_j,\omega\}$.
Thus it suffices to prove
\begin{equation}
 \phi_\delta(\omega)
 =(1+b_\star)(1-e^{-\omega})-b_\star\omega
 -(1-\delta)\min\{\delta,\omega\}\ge0
 \qquad(\delta,\omega\in[0,1]),
 \label{gcon:phi}
\end{equation}
with $\delta=1-v_j$.
For $\omega\le\delta$, the bound $1-e^{-\omega}\ge\omega-\omega^2/2$ gives
$\phi_\delta(\omega)\ge\delta\omega-(1+b_\star)\omega^2/2
\ge(1-b_\star)\omega^2/2$, using $\delta\ge\omega$ in the last step.
On $[\delta,1]$, $\phi_\delta$ is concave, so its endpoint values suffice.
The left endpoint was checked, and
\begin{equation}
 \phi_\delta(1)=(1+b_\star)(1-e^{-1})-b_\star-\delta(1-\delta)
 \ge\tfrac34-(1+b_\star)e^{-1}
 \ge\tfrac34-2e^{-1}>0.
 \label{gcon:phi-endpoint}
\end{equation}
Here $\delta(1-\delta)\le1/4$, $b_\star<1$, and $e>8/3$ were used.
For each $j$, since $m_j(\tau)\le\sum_k\pi_k=1$, inequality
\eqref{gcon:scalar} applies with $\omega=m_j(\tau)$ for every $\tau\ge0$.
By \eqref{gcon:threshold},
\[
 \sum_A\nu_AZ_j(y_A)=\int_0^\infty\sum_{A:\,Z_j(y_A)>\tau}\nu_A\,d\tau
 \ge\int_0^\infty\bigl(1-e^{-m_j(\tau)}\bigr)\,d\tau .
\]
Multiplying by $1-p_j$, applying \eqref{gcon:scalar} under the integral,
and then using \eqref{gcon:layer-cake} gives
\begin{equation}\label{gcon:large-coordinate-payment}
\begin{aligned}
 (1-p_j)\sum_{A\subseteq[N]}\nu_AZ_j(y_A)
 \ge\alpha_\star\Big[&(1-p_j)\sum_{k=1}^N\pi_kZ_j(o_S^k)\\
 &+p_j\sum_{k=1}^N\pi_k\bigl(1-(o_L^k)_j\bigr)Z_j(o_S^k)\Big].
\end{aligned}
\end{equation}
The bound \eqref{gcon:baseline-payment} gives
$R_0(p)\ge\alpha_\star C_{R_0}$. Adding it to
\eqref{gcon:large-coordinate-payment} for all $j$ and then using
\eqref{gcon:stock-cost} proves
\begin{align}
 \sum_{A\subseteq[N]}\nu_AR(y_A,p)
 &=R_0(p)+\sum_{j=1}^{|L|}(1-p_j)\sum_{A\subseteq[N]}\nu_AZ_j(y_A)\nonumber\\
 &\ge\alpha_\star\mathcal V_{\mathrm{res}}.
 \label{gcon:large-block-payment}
\end{align}
The function $e$ is controlled by a second consequence of the
affine lower bound. When $S$ is empty, $e=0$ and its contribution
vanishes. Otherwise, for any nonnegative $C^1$ DR function
$V:[0,1]^S\to\R_{\ge0}$ and finite
probability law with points $z^1,\ldots,z^m$, weights
$\lambda_1,\ldots,\lambda_m$, and mean
$\bar z=\sum_{\iota=1}^m\lambda_\iota z^\iota$, one has
\begin{equation}
 \sum_{\iota=1}^m\lambda_\iota V(z^\iota)\ge
       (1-\|\bar z\|_\infty)V(0).
 \label{gcon:retention-ineq}
\end{equation}
To prove it, let $n_S=|S|$, relabel the coordinates of $S$ so that
$\bar z_1\ge\cdots\ge\bar z_{n_S}$, put
$C'_j=\{1,\ldots,j\}$ and $\bar z_{n_S+1}=0$, and let
$\ell_V$ be the greedy affine minorant \eqref{gcon:minorant} of $V$ for
this order. Then
\begin{align*}
 \sum_{\iota=1}^m\lambda_\iota V(z^\iota)
 &\ge\sum_{\iota=1}^m\lambda_\iota\ell_V(z^\iota)=\ell_V(\bar z)\\
 &=(1-\bar z_1)V(0)
   +\sum_{j=1}^{n_S}\bigl(\bar z_j-\bar z_{j+1}\bigr)
     V(\one_{C'_j})\\
 &\ge(1-\|\bar z\|_\infty)V(0).
\end{align*}
The equality in the second line is the summation by parts of
\eqref{gcon:minorant-payment}, and the last step uses that every
coefficient and every vertex value is nonnegative.

Relabel the comparators so that $M_1\ge M_2\ge\cdots\ge M_N$, where
$M_k=e(t^k)$, with a fixed tie rule; the law \eqref{gcon:finite-law} is
invariant under this relabeling. The nonempty subsets $A\subseteq[N]$ are
partitioned by their first label $k$. For a fixed $k$, their total weight is
$\theta_k\prod_{k'=1}^{k-1}(1-\theta_{k'})$.
Conditionally on that first label,
$y_A=t^k\oplus y_{A\setminus\{k\}}$,
and the later labels retain their independent product law.
Their coordinate mean is
$1-e^{-\sum_{k'=k+1}^N\pi_{k'}(o_S^{k'})_i}\le a_i<s_\star$.
The composed function $y\mapsto e(t^k\oplus y)$ is nonnegative
and DR, since its coordinate derivatives are those of $e$ times
nonnegative constant factors. Inequality~\eqref{gcon:retention-ineq} therefore
gives a lower bound of $b_\star=1-s_\star$ times $M_k$. Hence
\begin{align}
 \sum_{A\subseteq[N]}\nu_Ae(y_A)
 &\ge b_\star\sum_{k=1}^N\theta_k
             \prod_{k'=1}^{k-1}(1-\theta_{k'})M_k\nonumber\\
 &\ge b_\star(1-e^{-1})\sum_{k=1}^N\pi_k e(t^k).
 \label{gcon:terminal-inflated}
\end{align}
To justify the last inequality, note that since
$M_1\ge\cdots\ge M_N$, the first selected label carries the largest
value; so the middle sum is the expected maximum $M_k$ among selected
labels, with value zero if none is selected. Its layer-cake integral is
$\int_0^\infty(1-e^{-\sum_{\substack{1\le k\le N\\M_k>\tau}}\pi_k})\,d\tau$.
Since $\sum_{k:M_k>\tau}\pi_k\le1$, the scalar inequality
$1-e^{-\omega}\ge(1-e^{-1})\omega$ for $0\le\omega\le1$ gives the last bound in \eqref{gcon:terminal-inflated}.

It remains to compare the values at $t^k$ with those at the original
comparators. For $0\le y\le1$
and $0<\chi\le1$,
\begin{equation}
 0\le1-e^{-\chi y}-y(1-e^{-\chi})
       \le\frac{\chi^2}{2}y(1-y).
 \label{gcon:inflation-error}
\end{equation}
Both differences vanish at the endpoints. The first is concave;
the difference between the upper bound and the middle expression
is also concave, with second derivative
$\chi^2(e^{-\chi y}-1)\le0$. This proves both inequalities.
Applying \eqref{gcon:inflation-error} with $\chi=\pi_k$ and
$y=(o_S^k)_i$, and dividing by $\theta_k\ge\pi_k(1-e^{-1})$ (by concavity of
$\chi\mapsto1-e^{-\chi}$ on $[0,1]$), gives
\[
 t_i^k-(o_S^k)_i\le
 \frac{\pi_k(o_S^k)_i\bigl(1-(o_S^k)_i\bigr)}{2(1-e^{-1})}.
\]
Since $\pi_k(o_S^k)_i\le u_i\le s_\star$, it follows that
\begin{equation}
 t_i^k-(o_S^k)_i\le\rho\left(1-(o_S^k)_i\right),\qquad
 \rho=\frac{s_\star}{2(1-e^{-1})}<1.
 \label{gcon:inflation-relative}
\end{equation}
Thus $t^k=(1-\rho)o_S^k+\rho d^k$ with
$d^k=o_S^k+\rho^{-1}(t^k-o_S^k)$, which satisfies
$o_S^k\le d^k\le\one_S$ by \eqref{gcon:inflation-relative}.
Concavity of $e$ in this nonnegative direction, together with
$e\ge0$, gives $e(t^k)\ge(1-\rho)e(o_S^k)$.
Consequently \eqref{gcon:terminal-inflated} implies
\begin{equation}
 \sum_{A\subseteq[N]}\nu_Ae(y_A)\ge
 b_\star(1-e^{-1}-s_\star/2)\sum_{k=1}^N\pi_k e(o_S^k)
 =\gamma\sum_{k=1}^N\pi_k e(o_S^k),
 \label{gcon:terminal-payment}
\end{equation}
where $\gamma=b_\star(1-e^{-1}-s_\star/2)$, using $\rho(1-e^{-1})=s_\star/2$.
Since $e^{-1}<3/8$ and $s_\star=1-b_\star$,
$\gamma>b_\star\bigl(\tfrac58-\tfrac{1-b_\star}2\bigr)=\tfrac{b_\star}8+\tfrac{b_\star^2}2
=\tfrac{3+b_\star}8$. Moreover
$(3+b_\star)(1+b_\star)-8b_\star=3-4b_\star+b_\star^2=\tfrac{15}4-4b_\star>0$,
because $4b_\star=2\sqrt3<15/4$; hence $\gamma>b_\star/(1+b_\star)=\alpha_\star$.

The comparison \eqref{gcon:jensen}, followed by
\eqref{gcon:large-block-payment} and \eqref{gcon:terminal-payment},
gives
\begin{align}
 F(a,p)&\ge\sum_{A\subseteq[N]}\nu_A\big[e(y_A)+R(y_A,p)\big]\nonumber\\
 &\ge\gamma\sum_{k=1}^N\pi_k e(o_S^k)
      +\alpha_\star\mathcal V_{\mathrm{res}}\nonumber\\
 &\ge\alpha_\star\sum_{k=1}^N\pi_k F(o^k).
 \label{gcon:normalized-recovery}
\end{align}
The last step uses $e\ge0$ and $\gamma>\alpha_\star$.
These inequalities also cover empty blocks: if $S$ is empty,
$F(a,p)=R_0(p)$ and \eqref{gcon:baseline-payment} suffices; if $L$ is empty,
only the $e$ terms remain.
By \eqref{gcon:law},
$f^+(x)=\sum_k\pi_k[F(o^k)+\ell(o^k)]=\sum_k\pi_kF(o^k)+\ell(x)$. Since
$\ell$ is affine and $h_\star(x)=(1-\alpha_\star)r+\alpha_\star x$, we
have $\ell(h_\star(x))=(1-\alpha_\star)\ell(r)+\alpha_\star\ell(x)$. Subtracting
gives the identity below, and the inequality follows from
\eqref{gcon:normalized-recovery} and \eqref{gcon:minorant-payment}:
\begin{align}
 f(h_\star(x))-\alpha_\star f^+(x)
 &=F(a,p)-\alpha_\star\sum_{k=1}^N\pi_kF(o^k)
       +(1-\alpha_\star)\ell(r)\nonumber\\
 &\ge0.
 \label{gcon:gap}
\end{align}
This proves assertion (i).

\subsubsection{Bounded supporting slopes on the closed cube}
The concave closure has coordinate increment bounds inherited from
the endpoint derivatives of $f$. These bounds control its supporting
slopes uniformly and allow them to extend to the boundary.

Mixtures of comparator laws show that $f^+$ is concave. It is finite
by compactness and has the coordinate increment bounds
\begin{equation}
 \partial_i f(\one)t\le f^+(x+te_i)-f^+(x)
 \le\partial_i f(\zero)t
 \qquad(0\le t\le1-x_i).
 \label{gcon:envelope-lipschitz}
\end{equation}
For $t>0$, replace $O_i$ in an attaining mean-$x$ law by
$O_i+\delta$ with $\delta=t(1-O_i)/(1-x_i)$. The new coordinate lies in
$[O_i,1]$ because $t\le1-x_i$, and $\E\delta=t$, so the mean increases by
$te_i$. Since $\partial_if\ge\partial_if(\one)$ on the cube,
$f(O+\delta e_i)-f(O)\ge\partial_if(\one)\delta$; taking expectations
proves the lower bound. Conversely, replace $O_i$ in an attaining
mean-$(x+te_i)$ law by $O_i-\delta$ with $\delta=tO_i/(x_i+t)\ge0$,
which has mean $t$, so the mean decreases by $te_i$. Since
$\partial_if\le\partial_if(\zero)$, the old expected objective minus the
new one is at most $\partial_if(\zero)t$, proving the upper bound. The case $t=0$ is immediate. Changing
coordinates one at a time in \eqref{gcon:envelope-lipschitz} gives
\begin{equation}\label{gcon:envelope-continuity}
 |f^+(y)-f^+(x)|\le\sum_{i=1}^n H_i(f)|y_i-x_i|
                    \le H(f)\norm{y-x}.
\end{equation}
In particular, $f^+$ is continuous on the closed cube.

At an interior point $x$, the supporting-hyperplane theorem for finite
concave functions gives a supergradient $g_x$ satisfying
$f^+(o)\le f^+(x)+\ip{g_x}{o-x}$ throughout the cube.
Testing $x\pm te_i$ in this inequality and using
\eqref{gcon:envelope-lipschitz} gives
$\partial_i f(\one)\le(g_x)_i\le\partial_i f(\zero)$.
For a boundary point $x$, interior points approaching $x$ have
supergradients that lie in the same compact box; a convergent
subsequence and \eqref{gcon:envelope-continuity} give the supporting inequality
at $x$ for every comparator $o$. Since $f\le f^+$,
assertion (i) proves \eqref{gcon:field}.
The endpoint box gives the stated norm bound.

This establishes assertion (ii).

\subsubsection{Sharpness from four multilinear objectives}
An objective-independent map must satisfy the recovery inequality for
several objectives at the same state and output. The following four
functions impose incompatible requirements above $\alpha_\star$.

Suppose deterministic maps $T_n:[0,1]^n\to[0,1]^n$, independent
of the objective, satisfy $T_n(x)\le x$ and
$f(T_n(x))\ge c f^+(x)$ for every nonnegative $C^1$ DR function
and every dimension. It suffices to use dimension four. Fix
$0<\epsilon<1/2$ and write
\begin{equation}
 x_\epsilon=(\epsilon,1-\epsilon,1/2,1/2),\qquad
 y=T_4(x_\epsilon)\le x_\epsilon.
 \label{gcon:sharp-state}
\end{equation}
The following nonnegative multilinear DR functions, which we call test
functions, give necessary inequalities at this same state and output;
the first three are cuts $o_i(1-o_j)$:
\begin{equation}
\begin{array}{c|c|c}
 f(o)&f^+(x_\epsilon)&\text{necessary inequality}\\ \hline
 o_1(1-o_2)&\epsilon&c\le1-y_2\\
 o_3(1-o_4)&1/2&c\le2y_3(1-y_4)\\
 o_4(1-o_3)&1/2&c\le2y_4(1-y_3)\\
 o_2(1-o_3o_4)&1-\epsilon&c(1-\epsilon)\le y_2(1-y_3y_4).
\end{array}
\label{gcon:sharp-witnesses}
\end{equation}
For a cut $o_i(1-o_j)$, its concave closure at $x$ equals
$\min\{x_i,1-x_j\}$. The pointwise inequalities
$o_i(1-o_j)\le o_i,1-o_j$ give the upper bound. To attain it, draw
$U\sim\Unif[0,1]$ and set $O_i=\one\{U\le x_i\}$ and
$O_j=\one\{U>1-x_j\}$, with the remaining coordinates sampled
independently at their prescribed means. Then
$\E[O_i(1-O_j)]=\min\{x_i,1-x_j\}$. For the cubic test function, the upper bound is
$o_2$; making coordinates $3$ and $4$ mutually exclusive attains
it while supplying the other means independently.
The cubic is DR because its nonzero mixed derivatives are
$-o_4,-o_3,-o_2$, and its pure second derivatives vanish.
For the first row, the guarantee is
$c\epsilon\le y_1(1-y_2)\le\epsilon(1-y_2)$; division by
$\epsilon>0$ gives the displayed inequality.

Put $\varrho=\sqrt{y_3y_4}\in[0,1/2]$, so that
$\min\{y_3,y_4\}\le\varrho$ and $y_3y_4=\varrho^2$. The middle two
test functions imply
\begin{equation}
 c\le2\left(\min\{y_3,y_4\}-y_3y_4\right)\le2\varrho(1-\varrho).
 \label{gcon:sharp-increasing}
\end{equation}
The first test function gives $y_2\le1-c$. Since $1-\varrho^2>0$,
substituting into the last test function gives
$c(1-\epsilon)\le(1-c)(1-\varrho^2)$. Collecting the terms in $c$
and dividing by $2-\epsilon-\varrho^2>0$ yields
\begin{equation}
 c\le\frac{1-\varrho^2}{2-\epsilon-\varrho^2}.
 \label{gcon:sharp-decreasing}
\end{equation}
At $\epsilon=0$, $2\varrho(1-\varrho)$ is increasing on $[0,1/2]$ and
$(1-\varrho^2)/(2-\varrho^2)$ is decreasing. Their unique crossing satisfies
\begin{equation}
 (\varrho-1)^2(2\varrho^2+2\varrho-1)=0,\qquad
 \varrho_\star=\frac{\sqrt3-1}{2}.
 \label{gcon:sharp-crossing}
\end{equation}
Indeed,
$2\varrho(1-\varrho)(2-\varrho^2)-(1-\varrho^2)=(\varrho-1)^2(2\varrho^2+2\varrho-1)$.
At $\varrho_\star$, $\varrho_\star^2=1-\sqrt3/2=s_\star$ and
$2\varrho_\star(1-\varrho_\star)=(\sqrt3-1)-(2-\sqrt3)=\alpha_\star$, so both
values there equal $\alpha_\star$.
The two monotonicity assertions follow from derivatives
$2-4\varrho\ge0$ and $-2\varrho/(2-\varrho^2)^2\le0$, respectively.
For every $\epsilon$, if $\varrho\le \varrho_\star$, then
\eqref{gcon:sharp-increasing} already gives $c\le\alpha_\star$.
Otherwise $\varrho^2>\varrho_\star^2=s_\star$. For fixed
$\epsilon\in(0,1/2)$, the map $q\mapsto(1-q)/(2-\epsilon-q)$ has derivative
$(\epsilon-1)/(2-\epsilon-q)^2<0$, so \eqref{gcon:sharp-decreasing} gives
$c\le(1-s_\star)/(2-\epsilon-s_\star)=b_\star/(1+b_\star-\epsilon)$, that is,
\begin{equation}
 c\le\frac{\alpha_\star}{1-\epsilon/(1+b_\star)}.
 \label{gcon:sharp-limit}
\end{equation}
Letting $\epsilon\downarrow0$ proves $c\le\alpha_\star$.
This proves assertion (iii) and completes the theorem.\qed

\setlength{\bibsep}{2pt}
\bibliography{references}
\end{document}